\documentclass[11pt,a4paper,final]{article}
\usepackage[left=1in,right=1in,top=1in,bottom=1in]{geometry}
\usepackage{amsmath}

\usepackage{amsthm}
\usepackage{XCharter}
\usepackage[xcharter,vvarbb]{newtxmath}

\usepackage[numbers,sort&compress]{natbib}
\usepackage{xcolor}
\usepackage{bm}
\usepackage{eucal}
\usepackage{tabularx}
\usepackage{hyperref}
\usepackage{cleveref}
\usepackage[ruled]{algorithm2e} 
\usepackage{authblk}
\usepackage{graphicx}
\usepackage{multirow}
\usepackage{enumitem}
\usepackage{thm-restate}
\newtheoremstyle{spaced}
  {0.5em}          
  {0em}          
  {\normalfont}  
  {}             
  {\bfseries}    
  {.}            
  { }            
  {}             

\theoremstyle{spaced}

\usepackage{wrapfig}
\usepackage{subcaption}

\usepackage{parskip}

\usepackage[suppress]{color-edits}
\addauthor{kh}{red}
\addauthor{sp}{blue}
\addauthor{TODO}{magenta}

\DeclareMathOperator*{\argmin}{argmin}

\newcommand{\cL}{\mathcal{L}}

\newcommand{\cS}{\mathcal{S}}
\newcommand{\cD}{\mathcal{D}}
\newcommand{\alg}{\mathrm{alg}}

\renewcommand \vec [1]{\bm{#1}}
\newcommand{\norm}[1]{\lVert#1\rVert}

\newcommand{\core}{\mathsf{core}}

\renewcommand{\epsilon}{\varepsilon}

\newtheorem{theorem}{Theorem}[section]

\newtheorem{corollary}[theorem]{Corollary}

\newtheorem{example}[theorem]{Example}
\newtheorem{definition}[theorem]{Definition}

\begin{document}

\title{\textbf{In-Context Credit Assignment via the Core}}

\author[1]{Keegan Harris}
\author[2]{Siddharth Prasad}
\author[3]{Asher Trockman}

\affil[1]{UC Berkeley, \texttt{keegan.harris@berkeley.edu}}
\affil[2]{Toyota Technological Institute at Chicago, \texttt{sprasad@ttic.edu}}
\affil[3]{Google Research, \texttt{trockman@google.com}}

\date{}

\maketitle

\begin{abstract}
    We propose incentive-aligned mechanisms for \emph{in-context credit assignment}: the task of assigning credit for AI-generated content among creators whose intellectual property appears in the context window. 
Our approach is based on the {\em least core} solution concept from cooperative game theory, which distributes value in a way that is as stable as possible by ensuring that no subset of creators is significantly under-compensated relative to the value they could generate on their own.
We develop algorithms for approximating the least core, which leverage novel routines for constraint seeding and constraint separation.
On a web retrieval credit assignment task, we find that our approaches approximate the least core using an order of magnitude fewer LLM calls compared to alternative methods.\looseness-1




\end{abstract}

\section{Introduction}
The creator economy is currently undergoing a massive transformation as a result of generative AI tools. 
Across domains such as text, images, and video, AI tools have enabled cheaper and faster production of content, in ways which have not previously been possible (e.g.~\cite{anderson2025making, kanna2024much}). 
In many use cases, 
this process involves synthesizing creators’ intellectual property (IP) that appears in the model's \emph{context window} into a final user-facing content item.\looseness-1 

For example, suppose a user queries a language model which uses a web search tool to aggregate and summarize information across many websites (e.g. code repositories, news articles) when generating its response to the query. 
As another example, suppose a user prompts a text-to-video model to generate a video featuring the likenesses of several individuals (e.g. celebrities or fictional characters). 
%
%
In each of these settings, the resulting content typically yields some sort of value for the end user and/or AI platform (measured in terms of revenue, user engagement, or otherwise). 
The goal of this paper is to design mechanisms that distribute credit for this value across the creators whose IP appears in-context, in a way that respects the incentives of the creators. 
Concretely, we aim to answer the following question:\looseness-1

\vspace{-5mm}

{\centerline \em \noindent{\bf Q:} \emph{What is an economically principled and computationally tractable way to allocate credit for AI-generated content among a set of creators whose IP appears in the model's context window?}\looseness-1}

Our first conceptual contribution is to model the in-context credit assignment problem as a \emph{coalitional game}, where creators work together to generate some value. 
In this setting, there is an unknown ground-truth reward function that determines the (counterfactual) value that would have been generated by any coalition of the creators. 
Rewards can be binary (e.g. ``Do the aggregated web pages contain enough information to answer the user's query?'') or continuous (e.g. ``What is the user engagement that would have been generated if only a subset of likenesses were used in the video?''). 
Instead of assuming access to the ground-truth reward function, we posit access to an {\em estimated} reward function (e.g., via a learned reward model or LLM-as-a-Judge). 
Our solution concept is the {\em least core}~\cite{gillies1959solutions}: the set of reward allocations that compensates every coalition of creators by minimizing the largest \emph{reward deficit} (i.e. amount of under-compensation relative to the counterfactual value generated) across all coalitions.\footnote{We choose the least core as our solution concept because it explicitly prioritizes the creators' incentives, whereas other cooperative solution concepts do not. See Appendix~\ref{app:discussion} for further discussion on this point.\looseness-1}\looseness-1 

%

We start by investigating the sensitivity of the least core to differences between the estimated and ground-truth reward functions.
We show that the sensitivity of the least-core deficit can be bounded in terms of reward errors weighted by a {\em balanced} distribution (i.e. the marginal membership probability for each {\em player} is the same) over coalitions. Our result can be viewed as a quantitative analog of the celebrated {\em Bondareva-Shapley theorem} from cooperative game theory~\citep{shapley1967balanced}. 
%
%
When rewards are binary, we provide a negative result: an estimated reward function with exponentially-few misclassifications can yield a high deficit even when the ground-truth deficit is zero. Nonetheless, we show that sensitivity of the deficit is bounded by the worst-case probability of misclassification over a balanced distribution.

Next we develop general-purpose algorithms for finding allocations in the estimated least core. 
While the problem of computing the (estimated) least core can be formulated as a linear program (LP), the number of constraints is exponential and so it cannot be solved directly for more than a modest number of creators. 
%
%
As such, we develop iterative separation-oracle based techniques.\footnote{A separation oracle is a function which certifies whether a particular point is in a particular convex set, and returns a separating hyperplane if it is not.}
%
We show that the ellipsoid method instantiated with a sampling-based separation oracle for finding under-compensated coalitions finds, with high probability, an approximate least-core allocation with polynomial sample complexity and polynomial runtime. 
Practically, we implement a {\em constraint generation (CG)} algorithm and instantiate it with a variety of different seeding and separation routines for finding under-compensated coalitions.\looseness-1 

We evaluate our methods on a semi-synthetic web retrieval credit assignment task, where the goal is use an LLM to aggregate information across many different data sources in order to answer a set of questions.  
We find that our CG algorithm instantiated with LLM-based separation oracles approximates least-core allocations using an order of magnitude fewer LLM calls compared to baselines.\looseness-1
\subsection{Related Work}

The least core is intractable to compute in general. 
As a result, there is a line of work on learning/approximating it from data. 
In particular,~\citet{balcan2015learning,balkanski2017statistical,yan2021if} provide sample complexity bounds for learning the core given sample evaluations of the value function on coalitions drawn from some underlying distribution;~\citet{gemp2024approximating} provide iterative convex optimization algorithms for core computation that avoid solving LPs. In contrast, we harness LLMs within constraint generation to iteratively solve the least-core LP.\looseness-1

%
\citet{wolsey1976nucleolus} describes a generic constraint generation algorithm for the core (but only discusses the separation subproblem in the abstract) and~\citet{nguyen2016finding} develop constraint generation algorithms for nucleolus computation (a refinement of the core that is computationally hard to even approximate~\cite{yan2021if}).
%
%
In the realm of combinatorial auctions (which can be modeled as a type of coalitional game), constraint generation methods for computing core outcomes have been developed and successfully deployed in real-world high-stakes auctions~\citep{day2007fair,day2012quadratic,prasad2026weakest}.\looseness-1

%
Adaptations of the Shapley value~\cite{shapley1953value} to determine the value of different features and data points in machine learning settings have become popular in recent years (e.g.,~\cite{lundberg2017unified, ghorbani2019data, jia2019towards, wang2024data}).  
\citet{yan2021if} argue for the use of the core in such settings instead. 

There is also a recent line of work on \emph{context attribution}, which aims to increase interpretability and/or fairness by identifying the parts of
the context that led to a model's generated output~\cite{cohen2024contextcite, liu2024attribot, qi2024model, wang2025tracllm, hirsch2025laquer, shen2025transparentize, xiao2025tokenshapley, chang2025jopa, li2025attributing, zhang2025taught, nematov2025source, pan2025context, patel2025maxshapley, de2025improving}. 
Like these papers, we are interested in credit assignment for data which appears in-context. 
However unlike this line of work, we are motivated by economic considerations, not interpretability or fairness. 
Our credit assignment algorithms based on the least core explicitly prioritize the {\em incentives} of the creators, which is a critical aspect missing from prior work. 
(See Appendix~\ref{app:discussion} for more details.)
\section{Setting and Background}

We cast the problem of in-context content credit assignment in the language of cooperative game theory. 
There is a set $N = \{1,\ldots,n\}$ {\em creators}, where each $i\in N$ is the creator of some IP that is used in the content generation process. 
%
We model rewards via a coalitional value function $v:2^{ N}\to[0,1]$. 
The total reward generated as a result of the content created by $N$ is the quantity to be distributed among the creators in $N$ and is normalized to $v(N) = 1$. 
Likewise, the baseline reward generated by the empty set is zero. 
%
%
For a coalition of creators $ S\subseteq N$, $v( S)\in [0,1]$ represents the counterfactual reward that the subset $S$ of creators would have generated on their own. 
While counterfactual rewards are generated according to the ground-truth value function $v$, we only assume that the platform has access to some \emph{estimate} of $v$, which we denote by $\widehat{v}:2^{ N}\to[0,1]$ (also normalized so that $\widehat{v}(N) = 1$ and $\widehat{v}(\emptyset) = 0$; $v$ and $\widehat{v}$ agree on the grand coalition $N$ since that is the reward observed by the platform).
This framing captures two important examples of AI content generation:\looseness-1

%

\begin{enumerate}[leftmargin=*, itemsep=0pt, topsep=0pt, parsep=0pt]
    \item{\em LLM web search.} A user queries an LLM, which draws on several web sources when generating its answer. 
    $N$ is the set of webpages referenced in the LLM response.
    The reward $v(N)$ generated from this interaction is proportional to user engagement/satisfaction with the platform.\footnote{For example, $v(N)$ might encode whether or not the returned answer correctly answers the user's query. More practically, $v(N)$ could also incorporate ad revenue and/or some value proportional to the amortized fee the user pays per LLM query.\looseness-1} 
    For a coalition $S \subseteq N$, $v(S)$ denotes the counterfactual reward of the response that would have been generated if the LLM were only allowed to draw from webpages in $S$ when generating its answer.
    An LLM-as-a-Judge or reward function obtained from human feedback may be used as a value function estimate $\widehat{v}$.
    \item{\em AI video generation.} A prompt author generates an AI video featuring the likenesses of the creators in $N$. 
    The video accrues some reward $v(N)$, which may be ad revenue proportional to the amount of engagement (e.g. clicks, likes, or watch-time) it receives. 
    For a subset of creators $S\subseteq N$, $v(S)$ is the counterfactual reward that would have been generated from a video containing only those $S$ creators.
    In video generation, $\widehat{v}$ may be an internal tool used to judge video quality, learned from, e.g., the performance of other videos on the platform. (Here, $v(S)$ is not directly observable short of creating a new video and letting it accrue reward on the platform.)
\end{enumerate}



\begin{definition}
    A value-sharing scheme is a vector $\vec{u}=(u_1,\ldots, u_N)\in[0,1]^{N}$ such that $\sum_{i\in N}u_i = 1$, where $u_i \in [0, 1]$ is creator $i$'s share of the reward.
\end{definition}

The core is the set of value-sharing schemes such that each subset of creators $S \subseteq N$ receives total value at least $v(S)$. 
The core is rarely nonempty.\footnote{E.g., for binary $v$, the core is nonempty if and only if there exists a {\em veto} player $i$ s.t. $v(S) = 1$ only if $i \in S$.\looseness-1} Therefore, we consider a relaxation, the $\varepsilon$-core, where each subset is under-compensated by an amount at most $\varepsilon$ compared to their value $v(S)$.\looseness-1

\begin{definition}[$\varepsilon$-(estimated) core]
    For $\varepsilon \geq 0$, the 
    $\varepsilon$-core and $\varepsilon$-estimated core are the polytopes
    \begin{equation*}
    \begin{aligned}
        \core_{\varepsilon}&=\left\{\vec{u}\in[0,1]^{N}:
        \textstyle\sum_{i\in  N} u_i = 1;\;\;
        \textstyle\sum_{i \in  S} u_i \geq v( S)-\varepsilon, \;\; \forall  S \subseteq  N
        \right\}, \\
        \widehat{\core}_{\varepsilon}&=\left\{\vec{u}\in[0,1]^{N}:
        \textstyle\sum_{i\in  N} u_i = 1;\;\;
        \textstyle\sum_{i \in  S} u_i \geq \widehat{v}( S)-\varepsilon, \;\; \forall  S \subseteq  N
        \right\}.      
    \end{aligned}
    \end{equation*}
\end{definition}
The least core is the $\varepsilon$-core with the smallest value of $\varepsilon$ for which the $\varepsilon$-core is nonempty. 

\begin{definition}[Least core]
    Let $\varepsilon^* = \min\{\varepsilon\ge 0 : \core_{\varepsilon}\neq\emptyset\}$ and $\widehat{\varepsilon} = \min\{\varepsilon\ge 0 : \widehat{\core}_{\varepsilon}\neq\emptyset\}$.
    The least core is the set $\core_{\varepsilon^*}$. The least estimated core is the set $\widehat{\core}_{\widehat{\varepsilon}}$. The quantities $\varepsilon^*$ and $\widehat{\varepsilon}$ solve the following linear programs. 

\begin{tabularx}{\linewidth}{@{}XX@{}}
\noindent\textbf{Least Core LP:}
\begin{equation}\label{eq:least-core-lp}
    \varepsilon^*= \min_{\varepsilon,\vec{u}}\left\{\varepsilon : \vec{u}\in\core_{\varepsilon}\right\}
\end{equation}
&
\noindent\textbf{Least Estimated Core LP:}
\begin{equation}\label{eq:least-estimated-core-lp}
    \widehat{\varepsilon}= \min_{\varepsilon,\vec{u}} \left\{\varepsilon : \vec{u}\in\widehat{\core}_{\varepsilon}\right\}
    \end{equation}
\end{tabularx}
\end{definition}
Ideally we would like to implement a value-sharing scheme $\vec{u}\in\core_{\varepsilon^*}$. 
However, since the true value function $v$ is unobservable, allocations in $\core_{\varepsilon^*}$ cannot be exactly computed. 
As a result, we focus on computing schemes in $\widehat{\core}_{\widehat{\varepsilon}}$, and show that under reasonable conditions these value-sharing schemes are not too far from those in $\core_{\varepsilon^*}$. 

The least (estimated) core is a set of value-sharing schemes. 
We break ties among schemes in the set by minimizing the $\ell_2$ norm. 
This allocation is called the \emph{egalitarian} least (estimated) core.\looseness-1 


\section{Sensitivity of the Least Core}\label{sec:sensitive}

This section bounds the difference between the solutions $\varepsilon^*$ and $\widehat{\varepsilon}$ of LP~\ref{eq:least-core-lp} and LP~\ref{eq:least-estimated-core-lp} as a function of the true and estimated value functions. 
We begin with the general setting, before specializing to binary rewards.
Results for supermodular value functions are included in the Appendix.\looseness-1

While one can always upper-bound $|\varepsilon^* - \widehat{\varepsilon}|$ with the worst-case deviation $\max_{S \subseteq N} |v(S) - \widehat{v}(S)|$, this bound may be overly pessimistic, as it depends only on the single most misestimated coalition. 
In contrast to this worst-case viewpoint, we use a dual characterization in terms of {\em balanced} distributions over coalitions. (A balanced distribution is one such that all creators have equal marginal inclusion probabilities.)
This perspective allows us to obtain a refined sensitivity bound that weights estimation errors according to a worst-case balanced distribution.

\begin{restatable}{theorem}{dbsb}\label{thm:dbsb}
    Let $\Lambda = \{\vec{\lambda}\in [0,1]^{2^n} : \sum_{S\subseteq N}\lambda_S = 1, \sum_{S\ni i}\lambda_S = \sum_{S\ni j}\lambda_S~\forall i,j\in N\}$ denote the set of balanced distributions over coalitions. 
    Then, 
    $|\varepsilon^* - \widehat{\varepsilon}|\le\max_{\vec{\lambda}\in\Lambda}\sum_{S \subseteq N}\lambda_S|v(S)-\widehat{v}(S)|$.
\end{restatable}

Theorem~\ref{thm:dbsb} can be viewed as a quantitative analog of the {\em Bondareva-Shapley theorem}~\citep{shapley1967balanced} which characterizes whether or not the core is empty based on balancedness (see also~\citet{kovalenkov2001exact}).
%
%
One simple consequence of Theorem~\ref{thm:dbsb} is that a least estimated core scheme $\widehat{\vec{u}}$ is always in a $2 \max_{S\subset N}\left|v(S)-\widehat{v}( S)\right|$-expansion of the least core:

\begin{corollary}\label{cor:core-containment}    $\widehat{\core}_{\widehat{\varepsilon}}\subseteq\core_{\widehat{\varepsilon}+\delta}\subseteq\core_{\varepsilon^*+2\delta}$, where $\delta=\max_{S\subseteq N}\left|v(S)-\widehat{v}( S)\right|$. 
\end{corollary}


In the binary setting (where $v(S)$ and $\widehat{v}(S)$ take values in $\{0, 1\}$), small pointwise disagreement between $v$ and $\widehat{v}$ does not imply small error in the least-core deficit. 
In particular, even if $v$ and $\widehat{v}$ disagree on an exponentially-small fraction of coalitions, the resulting least-core deficits $\varepsilon^*$ and $\widehat{\varepsilon}$ may differ substantially. 
The following example illustrates this phenomenon.

\begin{example}
Suppose $n = |N|$ is even, and let $A = \{1,\ldots,n/2\}$ and $B = \{n/2+1,\ldots,n\}$. 
Define
$v(S) = \mathbf{1}[A \subseteq S]$ and $\widehat{v}(S) = \mathbf{1}[A \subseteq S \ \wedge\ B \subseteq S]$.
Then $v(S) \neq \widehat{v}(S)$ if and only if $S = B \cup T$ for some $T \subseteq N \setminus A$.
The fraction of coalitions on which $v$ and $\widehat{v}$ disagree is therefore
$\frac{2^{n/2}}{2^n - 1} \approx 2^{-n/2}$. 
Despite this exponentially small disagreement, the least-core deficits differ significantly: $\varepsilon^* = 0$ but $\widehat{\varepsilon} = 1/2$.
\end{example}

This example demonstrates that the fraction of coalitions on which $v$ and $\widehat{v}$ disagree is not the right quantity for controlling the sensitivity of the least-core deficit in the binary setting. 
Instead, what matters is how these disagreements are distributed across coalitions.
Instantiating the duality-based bound from Theorem~\ref{thm:dbsb} yields a sharper characterization in the binary setting: the sensitivity is governed by the worst-case disagreement probability under a balanced distribution over coalitions.

\begin{corollary}
    Let $v$ and $\widehat{v}$ take values in $\{0,1\}$, and let $\Lambda$ be the set of balanced distributions over coalitions defined in Theorem~\ref{thm:dbsb}. 
    Then,
    $|\varepsilon^* - \widehat{\varepsilon}|\le\max_{\vec{\lambda}\in\Lambda}\sum_{S : v(S)\neq\widehat{v}(S)}\lambda_S.$
\end{corollary}

These results quantify the error incurred by optimizing with respect to the estimated value function $\widehat{v}$ instead of $v$. 
In Section~\ref{sec:comp}, we address the complementary problem of computing an estimated least core allocation given access to $\widehat{v}$. 
\section{Computing the Least Core}\label{sec:comp}
We now develop algorithms for efficiently approximating the least estimated core. 
Our algorithms are variants of the classical ellipsoid and constraint generation (CG) methods~\cite{grotschel1988geometric, kelley1960cutting}. 
These methods solve a convex optimization problem (in our case, LP~\eqref{eq:least-estimated-core-lp}) by iteratively refining a feasible region using separating hyperplanes, without explicitly enumerating all constraints. 
Avoiding constraint enumeration is critical for core computation, as LP~\eqref{eq:least-estimated-core-lp}, which has $2^N$ constraints, cannot be written in memory for even a moderate number of creators.\looseness-1

The ellipsoid method~\citep{grotschel1981ellipsoid,grotschel1988geometric} maintains a shrinking outer approximation (an ellipsoid) of the feasible region using separating hyperplanes.
While the ellipsoid method enjoys worst-case polynomial runtime guarantees, it is inefficient to run in practice. 
CG is a more practical algorithm: 
instead of maintaining a shrinking ellipsoidal approximation to the feasible region, it progressively refines a \emph{polyhedral} approximation by explicitly accumulating violated constraints.
While CG does not have polynomial runtime guarantees, it is a popular method for solving optimization problems in practice~\citep{desaulniers2006column}.\looseness-1

Both the ellipsoid method and CG require the existence of a \emph{separation oracle} capable of providing a hyperplane which separates a given point from the feasible region---in our setting this is the task of finding an under-compensated coalition, if one exists. 
In this section, we show how to use these algorithms in our setting, and we provide several separation oracles for the in-context credit assignment problem.\looseness-1 

\begin{definition}[Separation oracle]
    A separation oracle for the $\varepsilon$-estimated core takes as input a candidate value-sharing scheme $\vec{u}\in[0,1]^N$ and deficit $\varepsilon \geq 0$, and returns a coalition $ S$ such that $\sum_{i\in S}u_i + \varepsilon < \widehat{v}(S)$, or returns $\mathtt{None}$ if no such coalition exists.\looseness-1
\end{definition}

Since there is no generic polynomial-time exact separation oracle for core computation, the separation oracles we provide are either approximate or based on LLM prompting.\looseness-1 

\subsection{Polynomial-Time Core Computation via the Ellipsoid Method}\label{sec:ellipsoid}

We first consider the separation oracle that randomly samples $m$ coalitions and returns an under-compensated coalition if one exists (formally defined in Algorithm~\ref{alg:samp-sep}). 
What follows is a guarantee on the \emph{$\delta$-probable} least estimated core: 
Theorem~\ref{thm:ellipse} says that with high probability, a $1 - \delta$ fraction of creator coalitions will have a value deficit that is at most as large as their deficit under the least estimated core, when the number of samples in the sampling-based oracle is large enough.\looseness-1 

\begin{restatable}{theorem}{corellipsoid}\label{thm:ellipse}
    Let $\cD$ be a distribution over coalitions.
    Suppose the ellipsoid method (Alg.~\ref{alg:ellipsoid}) is instantiated such that $\mathrm{GetUnhappyCoalition}$ is the sampling-based separation oracle (Alg.~\ref{alg:samp-sep}) with $m = O(\log(n/\gamma)/\delta)$ samples, and let $\vec{u}^{(\alg)}$ be the returned value-sharing scheme. 
    Then w.p. at least $1 - \gamma$, $\mathbb{P}_{ S \sim \cD}\left[\sum_{i \in  S} u_i^{(\alg)} + \widehat{\epsilon} \geq \widehat{v}( S) \right] \geq 1 - \delta,$
    %
    %
    where $\widehat{\varepsilon}$ is the least estimated core deficit.\looseness-1 
\end{restatable}

It is worth comparing our results in Theorem~\ref{thm:ellipse} to those of Theorem 1 in~\citet{yan2021if}, which obtains a value-sharing scheme in the $\delta$-probable least core by randomly sampling $O((n + \log(1/\gamma))/\delta^2)$ coalitions and solving the least core LP with just those coalitions' constraints.  
Ignoring logarithmic factors, our total sample complexity across all rounds scales as $\widetilde{O}(n^2/\delta)$ (since the ellipsoid method runs for $\widetilde{O}(n^2)$ rounds), whereas theirs scales as $\widetilde{O}(n/\delta^2)$. 
Hence the two bounds are complementary and are better in different regimes. 

%
The quadratic objective in the egalitarian least core admits an exact separation oracle through its epigraph, and so the only source of randomness is the sampling-based oracle.
Therefore, we can recover a probabilistic approximation of the egalitarian least estimated core by running two instances of the ellipsoid method: one for the estimated least core deficit LP and one for the egalitarian least estimated core QP.\looseness-1

\begin{restatable}{corollary}{elipsqp}\label{cor:eagl}
    Let $\varepsilon^{(\alg)}$ denote the estimated core deficit from Theorem~\ref{thm:ellipse}, and consider running the ellipsoid method with the same setup as Theorem~\ref{thm:ellipse} to solve 
    $\min_{\vec{u}'\in\widehat{\core}_{\varepsilon^{(\alg)}}}\frac{1}{2} \norm{\vec{u}'}_2^2$. 
    Then w.p. at least $1-2\gamma$, the returned solution minimizes $\|\vec{u}\|_2^2$ among all feasible value-sharing schemes.\looseness-1
\end{restatable} 
When combined with the sensitivity bounds in Section~\ref{sec:sensitive}, our computational guarantees yield end-to-end bounds on the stability of the resulting allocation under the \emph{true} value function $v$. 

\subsection{Constraint Generation Algorithms}
We now propose practical CG methods for (egalitarian) least estimated core computation.
For a collection of coalitions $\cS\subseteq 2^N$, let $\widehat{\mathrm{LP}}(\cS)$ denote the restricted least estimated core LP
\begin{equation}\label{eq:lpS}
    \min_{\vec{u}\in [0,1]^n}\left\{\varepsilon:
    \begin{aligned}
    \textstyle\sum_{i\in  N} u_i = 1;\;\;
    \textstyle\sum_{i \in  S} u_i \geq \widehat{v}( S)-\varepsilon,~\forall  S\in\cS
    \end{aligned}
    \right\}.\;       
\end{equation}
%
%
Similarly, given a collection of coalitions $\cS\subseteq 2^N$ and a candidate least deficit $\varepsilon$, let $\widehat{\mathrm{QP}}(\cS,\varepsilon)$ denote the restricted egalitarian least estimated core quadratic program\looseness-1
\begin{equation}\label{eq:qpS}
    \min_{\vec{u}\in[0,1]^n}\left\{\norm{\vec{u}}_2^2 : 
    \begin{aligned}
        \textstyle\sum_{i\in N}u_i = 1;\;\; 
        \textstyle\sum_{i\in S} u_i \ge \widehat{v}(S) - \varepsilon,~\forall S\in\cS        
    \end{aligned}
    \right\}.
\end{equation}

\begin{algorithm}[t]
\SetAlgoLined
\SetAlgoNoEnd
\SetAlgoSkip{}
Initialize $\cS^1\subseteq 2^N$\label{line:seeding}\\
\For{$t = 1,\ldots, T$}{
    Solve $\widehat{\mathrm{LP}}(\cS^t)$; let $\varepsilon$ be the returned optimal deficit \\
    Solve $\widehat{\mathrm{QP}}(\cS^t,\varepsilon)$; let $\vec{u}$ be the returned value-sharing scheme\\
    {\bf if} $\mathrm{GetUnhappyCoalition}(\vec{u},\varepsilon)=\mathtt{None}$ \textbf{return} $\vec{u},\varepsilon$
    $\cS^{t+1}\leftarrow\cS^t\cup\mathrm{GetUnhappyCoalition}(\vec{u},\varepsilon)$.
}
{\bf return} $\vec{u},\varepsilon$
\caption{Joint Constraint Generation for Egalitarian Least Core}\label{alg:cg-qp}
\end{algorithm}

Algorithm~\ref{alg:cg-qp} describes the joint CG algorithm that solves the least estimated core LP and the egalitarian least estimated core QP simultaneously, by growing $\cS$ at each iteration and re-solving $\widehat{\mathrm{LP}}(\cS)$ and $\widehat{\mathrm{QP}}(\cS,\varepsilon)$. 
If equipped with an {\em exact} separation oracle, Algorithm~\ref{alg:cg-qp} terminates with the egalitarian least estimated core allocation. Unlike ellipsoid-based algorithms, it is generally not possible to provide (polynomial) bounds on the number of CG rounds.
Nonetheless, with a sampling-based separation oracle, we prove that (1) the returned allocation is in the least estimated core with high probability, and (2) the returned allocation minimizes $\ell_2$ norm among all feasible allocations (Theorem~\ref{thm:cg}).\looseness-1 

\citet{yan2021if} (and the corresponding implementation available in the {\tt pyDVL} package~\citep{pyDVL_2025}), approximate the egalitarian least core by (1) solving the least core LP using their sampling method then (2) solving the \emph{full} QP using the returned least core deficit. 
A major shortcoming of this approach is if the least core LP is too large to be written down in memory---as is the premise that necessitates~\citeauthor{yan2021if} to develop sampling methods---then so is the QP. 
Our methods address this issue.\looseness-1 

Finally, we note that our algorithms may be instantiated with \emph{any} separation oracle, not just the sampling-based method of Algorithm~\ref{alg:samp-sep}. 
As one example, Algorithm~\ref{alg:llm-sep} uses an LLM to identify violated coalitions. This approach is inspired by the line of work on \emph{LLM-as-a-Judge}~\cite{zheng2023judging}, where LLMs are used to evaluate candidate outputs. In our setting, the LLM plays a related role by acting as a learned separation oracle: given a candidate allocation, it proposes coalitions that are likely to violate the core constraints, effectively serving as an adversarial tester that surfaces counterexamples.
While this oracle does not admit the same guarantees as the sampling-based method, it offers a flexible and practical heuristic for guiding the search toward informative constraints, particularly in settings where domain knowledge can be encoded through prompting.\looseness-1 
%



\subsection{Coalition Seeding for Binary Value Functions}

We now study how the choice of the initial constraint set in CG (Alg.~\ref{alg:cg-qp}, line 1) affects performance in the special case of binary rewards. 
Binary rewards are important in practice (e.g. ``Does a set of documents suffice to answer a query?''), and admit additional structure that can be exploited.\looseness-1

\begin{definition}[Minimal winning coalition]
    Let $\widehat{v} : 2^N\to\{0,1\}$ be a (estimated) binary value function that is {\em monotone}, i.e. $\widehat{v}(T)\geq \widehat{v}(S)$ for all $T\supseteq S$. A {\em minimal winning coalition (MWC)} for $\widehat{v}$ is a coalition $S\subseteq N$ such that $\widehat{v}(S) = 1$, but $\widehat{v}(T) = 0$ for all strict subsets $T\subset S$.
\end{definition}

Intuitively, MWCs are the smallest sets of creators that can independently generate value.
A key structural property is that MWCs form a complete description of the core constraints: any winning coalition contains at least one MWC, and constraints for larger coalitions are redundant once MWCs are enforced.\looseness-1

\begin{restatable}{lemma}{lemmwc}\label{lem:mwc}
    Given as input a list of all MWCs $\{S_1,\ldots, S_k\}$ of $\widehat{v}$, a least estimated core allocation can be found in $\mathsf{poly}(n, k)$ time by solving LP~\eqref{eq:lpS}. Furthermore, the egalitarian least core can be computed in $\mathsf{poly}(n, k)$ time by solving QP~\eqref{eq:qpS}.
\end{restatable}

However we rarely will know all MWCs in advance. 
The next result shows that CG can efficiently recover any missing MWCs, provided the separation oracle identifies highly violated coalitions.\looseness-1
%

\begin{restatable}{theorem}{mwcpred}\label{thm:mwc_pred}
    Let $\cL = \{S_1,\ldots, S_k\}$ denote the list of MWCs that CG is seeded with and let $\eta = |\{T\subset N : T\text{ is a MWC}, T\notin \cL\}|$ be the number of missing MWCs. 
    Suppose we are given access to a \emph{max-violation separation oracle}, i.e. $\mathrm{GetUnhappyCoalition}(\vec{u}) \in \argmin_{S : \widehat{v}(S) = 1}\sum_{i\in S}u_i$. 
    Then Algorithm~\ref{alg:cg-qp} computes the egalitarian least estimated core in $\mathsf{poly}(n, k + \eta)$ runtime with $\eta + 1$ oracle calls.\looseness-1 
\end{restatable}

These results suggest that effective seeding strategies for binary rewards should aim to approximate the set of MWCs, as the number of CG iterations is governed by how many MWCs are missing.
Motivated by this, we experiment with several different coalition seeding strategies in~\Cref{sec:experiments}.

\section{Experiments}
\label{sec:experiments}

We evaluate our CG methods on a semi-synthetic webpage retrieval benchmark, using both the BrowseComp-Plus dataset~\cite{chen2025browsecomp} and a dataset of our own creation. 
We use Gemini~3 Flash~\cite{gemini3flash_modelcard_2025} as the value function estimate (via LLM-as-a-Judge grading), and Gemini~3 Pro~\cite{gemini3pro_modelcard_2025} for generating our synthetic dataset. 
We test both egalitarian and non-egalitarian variants of our CG algorithms and find that they substantially outperform the sampling baseline of \citet{yan2021if} in terms of the number of LLM calls required. 
The datasets used in our experiments consists of question/answer pairs, each associated with a set of $n \in \{10, 32\}$ documents that serve as the creators in our coalitional game:\looseness-1


\textbf{BrowseComp-Plus~\citep{chen2025browsecomp}.}
This dataset assesses multi-hop question answering in a web retrieval context. Each query is paired with \emph{gold} documents (individually sufficient to answer), \emph{evidence} documents (each containing partial facts that yield the answer in aggregate), and \emph{negative} documents (irrelevant or containing distracting information). When queried individually, we found that only 64\% of gold documents lead to the correct answer, possibly due to our prompting strategy. 
However through manual inspection
we observed several label quality issues: gold documents are repeated verbatim in the evidence set, some negative documents contain query-relevant facts, and some gold documents do not actually contain the answer. A detailed example is provided in Appendix~\ref{app:browsecomp-details}.\looseness-1

\textbf{Synthetic.}
Given these issues with BrowseComp-Plus, we also create our own multi-hop Q/A dataset to ensure the existence of coalitions with genuine interaction effects. We prompt Gemini~3 Pro with high thinking budget to propose questions that are answerable only through a combination of 2--10 evidence documents, or a single gold document. We filter for topical uniqueness and draw negative documents from sufficiently unrelated questions. By construction, 100\% of standalone gold documents lead to the correct answer.\looseness-1

\begin{wrapfigure}[17]{r}{0.5\textwidth}
    \vspace{-10pt}
    \centering
    \includegraphics[height=6.7cm]{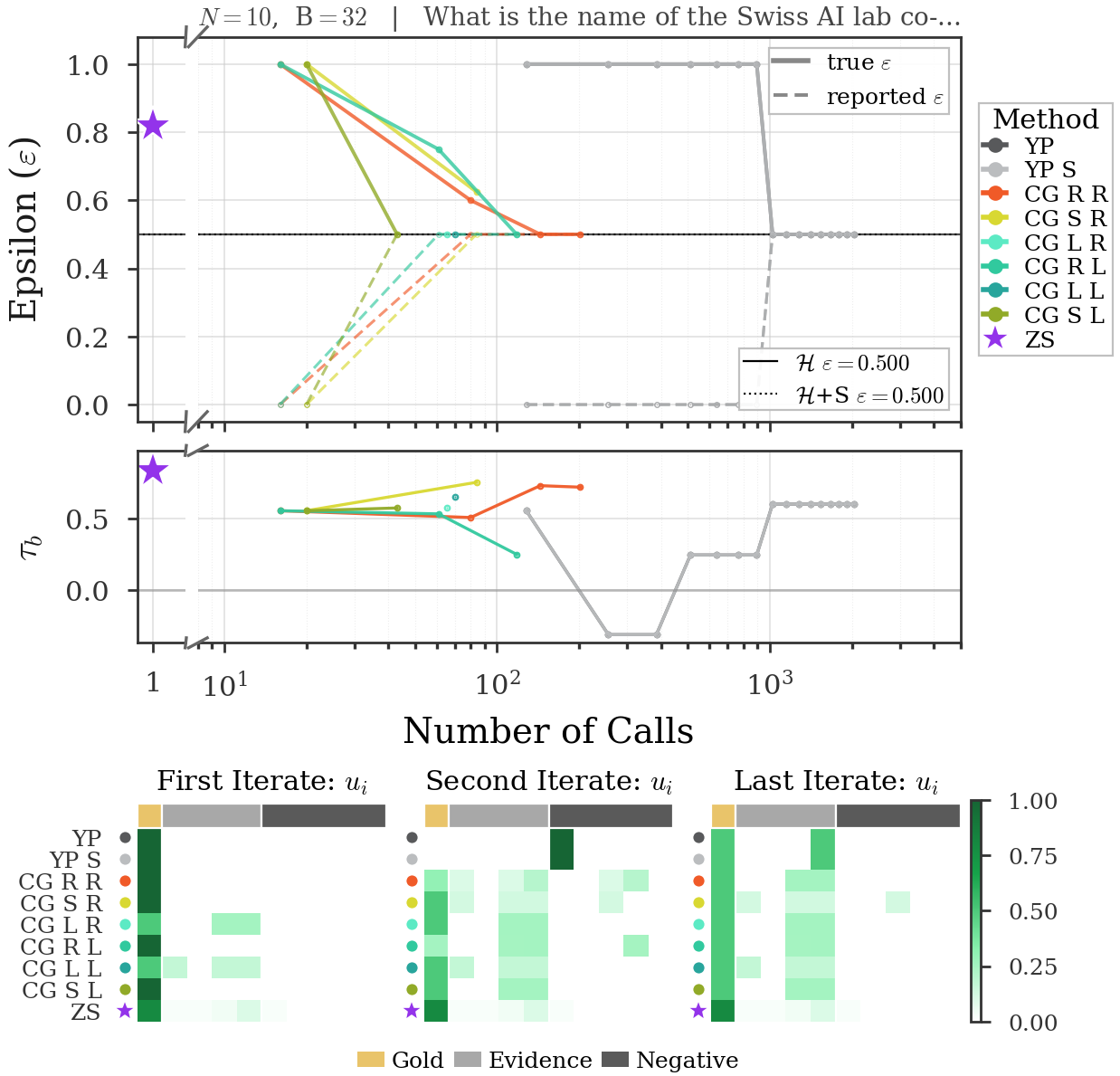}
    \caption{Convergence at tractable scale: Synthetic, $n{=}10$ (egal.). \textbf{CG} matches in ${\sim}30$ calls.\looseness-1}
    \label{fig:trace-synth10}
\end{wrapfigure}


We compare three families of algorithms: random sampling (\textbf{YP}; Algorithm~\ref{alg:yp}), our constraint generation approach (\textbf{CG}; Algorithm~\ref{alg:cg-qp}), and a zero-shot (\textbf{ZS}) method in which we ask the LLM to directly output a least core utility vector in a single query. We assign binary rewards $\widehat{v}(S) \in \{0,1\}$ to each coalition $S$ by prompting the model to answer the question using only the documents in $S$ (instructing it to ignore world knowledge), then grading the response for correctness with a second call. Each evaluation therefore requires two LLM calls.\looseness-1

\textbf{Evaluation Metrics.}
A notable difficulty is the lack of a ground-truth $(\varepsilon^*, \vec{u}^*)$ to compare to. Even restricting to small document sets where we can solve the full LP, intrinsic noise in LLM responses prevents exact evaluation. We therefore construct a holdout set $\mathcal{H}$ of $2^{12}$ coalitions sampled u.a.r., and for each method's utility vector $\vec{u}$ we compute the \emph{holdout} $\varepsilon$, $\widehat{\varepsilon}_{\mathcal{H}}(\vec{u})$:\looseness-1

\begin{definition}
$\widehat{\varepsilon}_{\mathcal{H}}(\vec{u})$ measures the out-of-sample quality of a utility vector $\vec{u}$ against held-out coalitions $\mathcal{H}$. Due to noise in LLM responses (e.g., a single mistake where the model uses world knowledge yields $\varepsilon = 1$), we use the 99th-percentile undercompensation rather than the maximum: $\widehat{\varepsilon}_{\mathcal{H}}(\vec{u}) = Q_{0.99}\!\left(\left\{ \widehat{v}(S) - \textstyle\sum_{i \in S} u_i \;:\; S \in \mathcal{H} \right\}\right).$%
\end{definition}%
Lower values of $\widehat{\varepsilon}_{\mathcal{H}}(\vec{u})$ indicate a more stable allocation, as fewer coalitions are substantially under-compensated.
We additionally report Kendall's $\tau_b$ rank correlation with document-type labels (Gold ${>}$ Evidence ${>}$ Negative), gold false-negative and negative false-positive rates, \% nonzero utility entries (measuring egalitarianness), and the total number of LLM calls before convergence---the primary axis on which our methods improve over prior work.

\begin{figure*}[t]
    \centering
    \includegraphics[width=\textwidth,trim=0 0 0 30,clip]{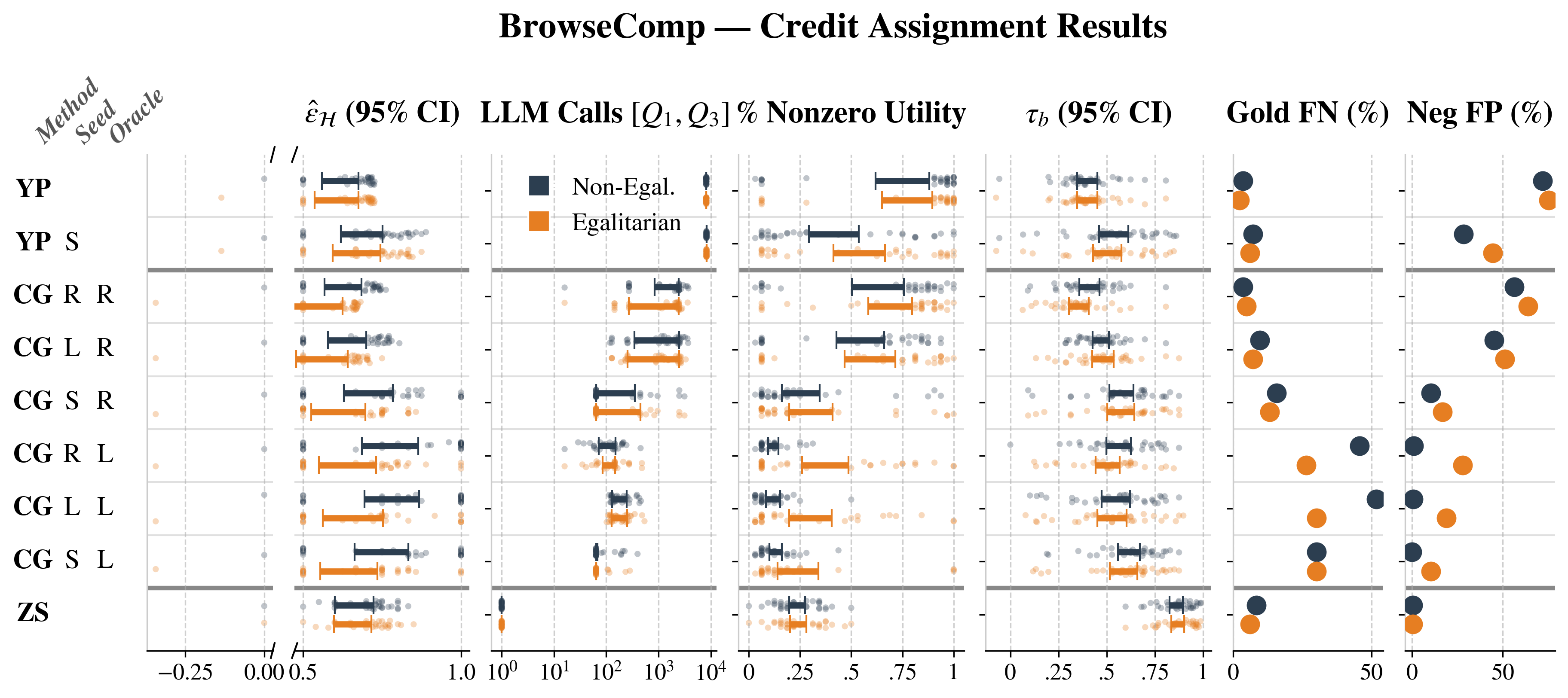}

    \caption{Credit assignment results on BrowseComp-Plus ($n{=}32$) for non-egalitarian (dark) and egalitarian (light) settings. Bars show 95\% bootstrap CIs of the sample means; points are individual problem instances. See Figure~\ref{fig:forest-synth} for the analogous Synthetic results.}
    \label{fig:forest}
\end{figure*}

\textbf{Algorithm Configurations.}
We test three coalition seedings for the initial constraint set: random coalitions (\textbf{R}), singleton coalitions (\textbf{S}), and LLM-proposed coalitions (\textbf{L}), where the LLM proposes the minimal winning coalitions it believes are sufficient to answer the question. For the separation oracle, we use batches of $2^6$ coalitions per round: the random oracle (\textbf{R}) samples uniformly, while the LLM oracle (\textbf{L}) proposes coalitions it suspects are under-compensated. As baselines, we evaluate \textbf{YP} (random sampling + LP solve), \textbf{YP\,S} (with singletons), and \textbf{ZS} (single-call zero-shot).

\begin{wrapfigure}[11]{r}{0.35\textwidth}
    \centering
    \vspace{-14pt}
    \includegraphics[height=3.9cm]{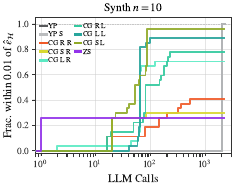}
    \caption{Calls-to-target of $\widehat{\varepsilon}_{\mathcal{H}} \pm .01$. More in App.~\ref{app:cdf-plots}.}
    \label{fig:cdf-synth10}
\end{wrapfigure}


We set a maximum budget of $B = 2^{12}$ checked coalitions per method. \textbf{YP} exhausts its full budget, then solves the LP; \textbf{CG} methods terminate early when no violated constraint is found. All experiments use batches of 64 coalitions per \textbf{CG} round. We repeat across 32 problem instances per dataset. Aggregate results are summarized in Figure~\ref{fig:forest}.

\looseness=-1
\textbf{CG} methods with a random oracle achieve comparable or better $\widehat{\varepsilon}_{\mathcal{H}}$ than \textbf{YP} while using 3--17$\times$ fewer calls. On BrowseComp-Plus, \textbf{CG(R,R)} achieves $\widehat{\varepsilon}_{\mathcal{H}} \approx 0.64$ (matching \textbf{YP}) at median ${\sim}2200$ calls vs.\ 8192. On Synthetic, \textbf{CG(R,R)} achieves $\widehat{\varepsilon}_{\mathcal{H}} \approx 0.47$ vs.\ 0.48 for \textbf{YP}, with similar call savings. LLM-oracle \textbf{CG} variants converge fastest (median 64--145 calls) but at moderately higher $\varepsilon$ (0.51--0.80). 
Singleton seeding (\textbf{S}) is effective for fast convergence, establishing documents' standalone values in the first round.\looseness-1

%
LLM-oracle \textbf{CG} methods produce sparser allocations (\% nonzero $\approx$ 12--25\%), concentrating utility on fewer documents. This yields low negative FP rates ($\leq$2\%) but higher gold FN rates (up to 47\%). The sparsity arises because the LLM oracle targets coalitions it expects to be high-value, so the LP only sees constraints involving genuinely useful documents and concentrates credit accordingly. By contrast, \textbf{YP}'s random sampling inevitably includes coalitions where negative documents appear alongside gold or evidence documents. 
When such coalitions happen to produce correct answers, credit leaks to the negatives. This explains \textbf{YP}'s broad-but-imprecise allocations (\% nonzero $\approx$ 77\%, negative FP 72\%, gold FN 3.6\%). The egalitarian QP objective successfully redistributes credit, roughly doubling \% nonzero for \textbf{CG} methods and reducing gold FN rates at the cost of increased negative FP, without materially degrading $\varepsilon$.\looseness-1

To understand the optimization dynamics, we present representative traces in Figure~\ref{fig:traces}. We use \textbf{YP}'s final $\widehat{\varepsilon}_{\mathcal{H}}$ as a reference for each instance, since \textbf{YP} exhausts its full budget and provides the strongest baseline.
In Synthetic $n{=}32$ (Figure~\ref{fig:traces}, left), \textbf{YP} converges to $\widehat{\varepsilon}_{\mathcal{H}} = 0.65$ after 8{,}192 calls. \textbf{CG(R,R)} matches this, reaching $\widehat{\varepsilon}_{\mathcal{H}} \approx 0.65$ (within 0.01) in only ${\sim}1{,}900$ calls, a 4$\times$ reduction. The LLM-oracle variants \textbf{CG(L,R)} and \textbf{CG(L,S)} converge even faster at 145 and 119 calls, though they overshoot the baseline by ${\sim}0.10$ ($\widehat{\varepsilon}_{\mathcal{H}} \approx 0.75$). \textbf{ZS} overshoots further ($\widehat{\varepsilon}_{\mathcal{H}} \approx 0.80$; 0.15 above \textbf{YP}). The reported $\varepsilon$ is biased downward as it reflects only constraints seen so far; we therefore track both.\looseness-1

\looseness=-1
In BrowseComp $n{=}32$ (Figure~\ref{fig:traces}, right), \textbf{YP} reaches $\widehat{\varepsilon}_{\mathcal{H}} = 0.62$ at 6{,}784 calls. \textbf{CG(R,R)} essentially matches it ($\widehat{\varepsilon}_{\mathcal{H}} \approx 0.63$) in only ${\sim}1{,}000$ calls---an 8$\times$ reduction. Remarkably, the LLM-oracle variants \textbf{CG(L,ZS)} and \textbf{CG(L,L)} \emph{undershoot} the baseline, achieving $\widehat{\varepsilon}_{\mathcal{H}} = 0.50$ (0.12 below \textbf{YP}) in just 104--128 calls. Perhaps our strongest evidence comes from the tractable $n{=}10$ scale (Figure~\ref{fig:trace-synth10}; full results in Appendix~\ref{app:n10}), where the coalition space is small enough for \textbf{YP} to be effectively exact: \textbf{CG} recovers the holdout optimum ($\widehat{\varepsilon}_{\mathcal{H}} = 0.5$) in just 30--60 calls, while \textbf{YP} requires $\approx2^{10}$ calls to reach the same point and \textbf{ZS} overshoots at $\widehat{\varepsilon}_{\mathcal{H}} \approx 0.82$. Aggregating across all $n{=}10$ instances, Figure~\ref{fig:cdf-synth10} confirms that \textbf{CG(L,L)} reaches within 0.01 of $\widehat{\varepsilon}_{\mathcal{H}}$ for ${>}90\%$ of instances in under 100 calls.

\begin{figure*}[t]
    \centering
    \begin{subfigure}[t]{0.44\textwidth}
        \includegraphics[height=7.1cm,trim=0 0 30 10,clip]{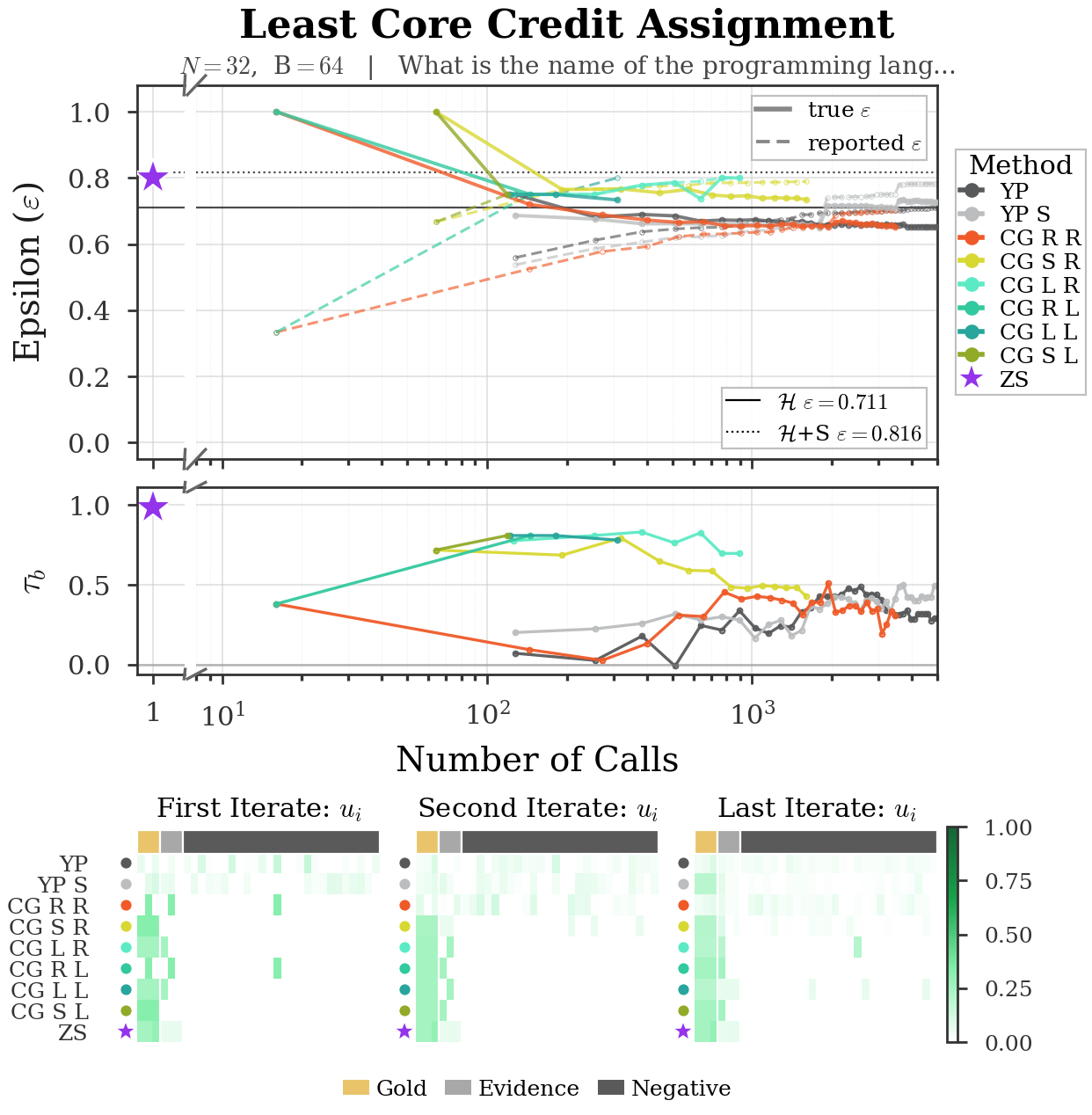}
        \caption{Synthetic, $n{=}32$}
        \label{fig:trace-synth32}
    \end{subfigure}\hfill%
    \begin{subfigure}[t]{0.52\textwidth}%
        \includegraphics[height=7.1cm,trim=0 0 0 10,clip]{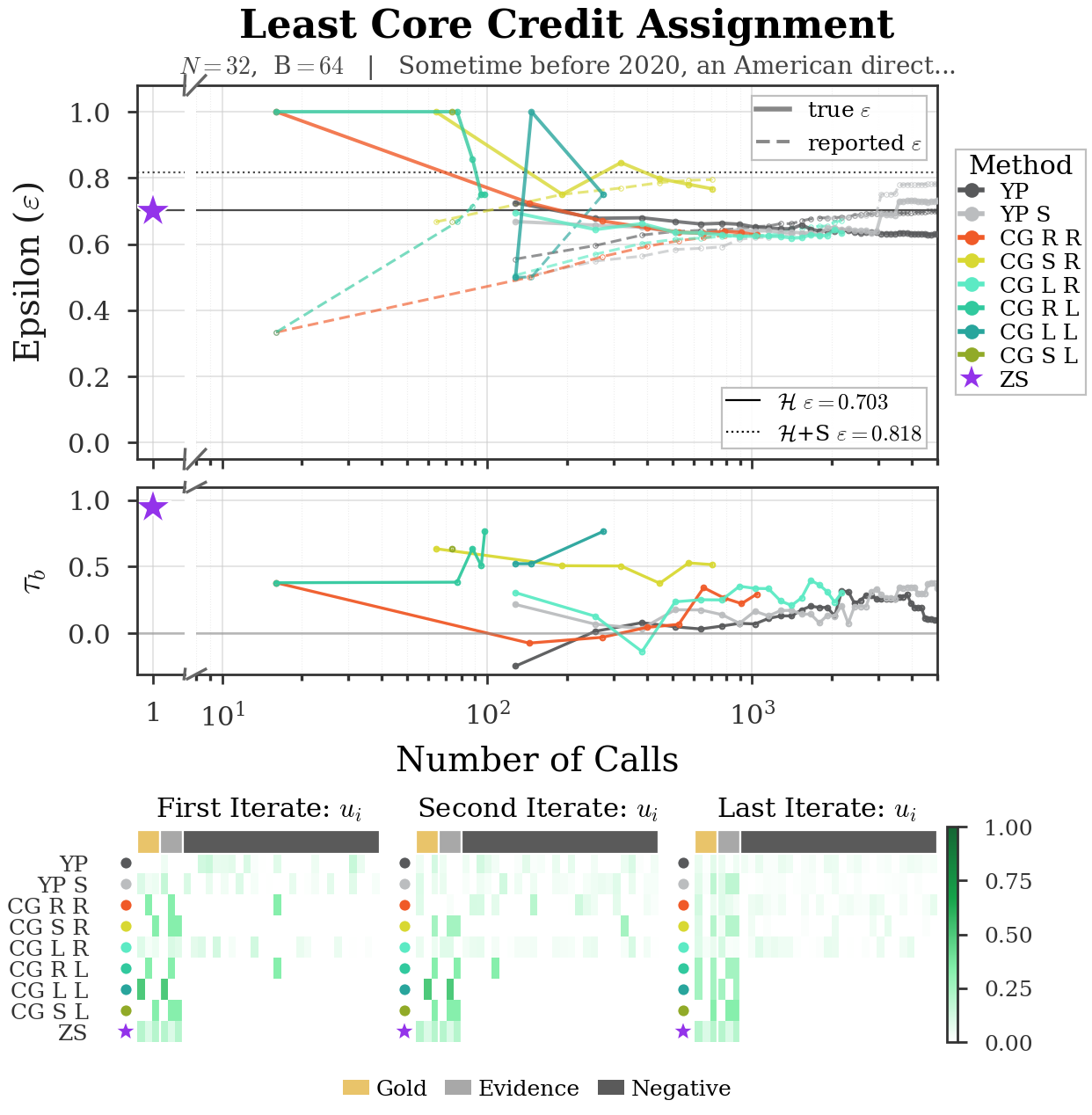}
        \caption{BrowseComp, $n{=}32$}
        \label{fig:trace-bc32}
    \end{subfigure}
    \caption{Convergence traces for representative $n{=}32$ instances. \textbf{Top panels:} holdout $\widehat{\varepsilon}_{\mathcal{H}}$ (solid) and reported $\varepsilon$ (dashed), with $\tau_b$ below. \textbf{Bottom:} credit heatmaps at the first, second, and last iterates.\looseness-1}
    \label{fig:traces}
\end{figure*}

\textbf{Summary.}
Our \textbf{CG} methods achieve comparable or better holdout $\varepsilon$ relative to \textbf{YP} while using 1--2 orders of magnitude fewer value-function queries. 
The egalitarian variants successfully redistribute credit more broadly with only minor impacts on $\varepsilon$. 
\textbf{ZS} is a surprisingly strong baseline: 
a single LLM call with no value-function feedback achieves mean $\widehat{\varepsilon}_{\mathcal{H}}$ within ${\sim}20\%$ of \textbf{YP} on BrowseComp-Plus and ${\sim}10\%$ on Synthetic, with high $\tau_b$. 
While this level of accuracy may suffice in some low-stakes settings, \textbf{CG} is useful when tighter $\varepsilon$ or formal guarantees are needed: 
with a random separation oracle, \textbf{CG} inherits the probabilistic guarantees of Theorem~\ref{thm:cg}, and matches or improves upon \textbf{YP}'s $\widehat{\varepsilon}_{\mathcal{H}}$ using 3--17$\times$ fewer calls. 
In most settings, we recommend using LLM-seeded \textbf{CG} (\textbf{CG(L,$\cdot$)}), which uses the LLM's guess as a warm start and iteratively refines it, combining the speed of \textbf{ZS} with the iterative tightening of \textbf{CG}. 
Additional convergence traces, $n{=}10$ experiments, and detailed dataset statistics are provided in Appendix~\ref{app:experiments}.\looseness-1

\section{Conclusions}\label{sec:conc}

We formalize in-context credit assignment for AI-generated content as a cooperative game between creators whose contributions appear in a model’s context window. 
Our credit assignment scheme is based on the least core, which minimizes the maximum under-compensation across all coalitions. 
We establish sensitivity guarantees that relate deviations in the least-core allocation to structured errors in the estimated reward function. 
We then develop scalable methods based on separation oracles and demonstrate empirically that our methods can approximate least-core allocations with an order of magnitude fewer value function queries compared to other approaches.  
More broadly, our work highlights the importance of integrating economic incentives into the design of AI platforms, especially as questions of credit attribution and compensation become increasingly central. 
%
\looseness-1


\section*{Acknowledgements}
KH was supported in part by the Simons Institute for the Theory of Computing, and part of this work was conducted when he was visiting the Institute.

\bibliographystyle{plainnat}
\bibliography{main/refs}

@article{shapley1971cores,
  title={Cores of convex games},
  author={Shapley, Lloyd S},
  journal={International Journal of Game Theory},
  volume={1},
  number={1},
  pages={11--26},
  year={1971},
  publisher={Springer}
}

@inproceedings{yan2021if,
  title={If You Like {S}hapley Then You’ll Love the Core},
  author={Yan, Tom and Procaccia, Ariel D},
  booktitle={Proceedings of the AAAI Conference on Artificial Intelligence},
  volume={35},
  pages={5751--5759},
  year={2021}
}

@article{balkanski2017statistical,
  title={Statistical cost sharing},
  author={Balkanski, Eric and Syed, Umar and Vassilvitskii, Sergei},
  journal={Advances in Neural Information Processing Systems},
  volume={30},
  year={2017}
}

@inproceedings{gemp2024approximating,
  title={Approximating the Core via Iterative Coalition Sampling},
  author={Gemp, Ian and Lanctot, Marc and Marris, Luke and Mao, Yiran and Du{\'e}{\~n}ez-Guzm{\'a}n, Edgar and Perrin, Sarah and Gyorgy, Andras and Elie, Romuald and Piliouras, Georgios and Kaisers, Michael and others},
  booktitle={Proceedings of the 23rd International Conference on Autonomous Agents and Multiagent Systems},
  pages={669--678},
  year={2024}
}

@article{day2012quadratic,
  title={Quadratic core-selecting payment rules for combinatorial auctions},
  author={Day, Robert W and Cramton, Peter},
  journal={Operations Research},
  volume={60},
  number={3},
  pages={588--603},
  year={2012},
  publisher={INFORMS}
}

@article{grotschel1981ellipsoid,
  title={The ellipsoid method and its consequences in combinatorial optimization},
  author={Gr{\"o}tschel, Martin and Lov{\'a}sz, L{\'a}szl{\'o} and Schrijver, Alexander},
  journal={Combinatorica},
  volume={1},
  number={2},
  pages={169--197},
  year={1981},
  publisher={Springer}
}

@inproceedings{prasad2026weakest,
  title={Weakest bidder types and new core-selecting combinatorial auctions},
  author={Prasad, Siddharth and Balcan, Maria-Florina and Sandholm, Tuomas},
  booktitle={Proceedings of the AAAI Conference on Artificial Intelligence},
  volume={40},
  pages={17206--17214},
  year={2026}
}

@article{nguyen2016finding,
  title={Finding the nucleoli of large cooperative games},
  author={Nguyen, Tri-Dung and Thomas, Lyn},
  journal={European Journal of Operational Research},
  volume={248},
  number={3},
  pages={1078--1092},
  year={2016},
  publisher={Elsevier}
}

@article{wolsey1976nucleolus,
  title={The nucleolus and kernel for simple games or special valid inequalities for 0--1 linear integer programs},
  author={Wolsey, LA},
  journal={International Journal of Game Theory},
  volume={5},
  number={4},
  pages={227--238},
  year={1976},
  publisher={Springer}
}

@article{day2007fair,
  title={Fair payments for efficient allocations in public sector combinatorial auctions},
  author={Day, Robert W and Raghavan, Subramanian},
  journal={Management science},
  volume={53},
  number={9},
  pages={1389--1406},
  year={2007},
  publisher={INFORMS}
}

@article{chen2025browsecomp,
  title={Browse{C}omp-{P}lus: A more fair and transparent evaluation benchmark of deep-research agent},
  author={Chen, Zijian and Ma, Xueguang and Zhuang, Shengyao and Nie, Ping and Zou, Kai and Liu, Andrew and Green, Joshua and Patel, Kshama and Meng, Ruoxi and Su, Mingyi and others},
  journal={arXiv preprint arXiv:2508.06600},
  year={2025}
}

@article{shapley1953value,
  title={A value for n-person games},
  author={Shapley, Lloyd S},
  year={1953},
  publisher={Princeton University Press Princeton},
  journal={Contributions to the Theory of Games II}
}

@article{gillies1959solutions,
  title={Solutions to general non-zero-sum games},
  author={Gillies, Donald B},
  journal={Contributions to the Theory of Games},
  volume={4},
  number={40},
  pages={47--85},
  year={1959}
}

@inproceedings{balcan2015learning,
  title={Learning cooperative games},
  author={Balcan, Maria-Florina and Procaccia, Ariel D and Zick, Yair},
  booktitle={Proceedings of the 24th International Joint Conference on Artificial Intelligence},
  pages={475--481},
  year={2015}
}

@article{lundberg2017unified,
  title={A unified approach to interpreting model predictions},
  author={Lundberg, Scott M and Lee, Su-In},
  journal={Advances in neural information processing systems},
  volume={30},
  year={2017}
}

@inproceedings{ghorbani2019data,
  title={Data {S}hapley: Equitable valuation of data for machine learning},
  author={Ghorbani, Amirata and Zou, James},
  booktitle={International Conference on Machine Learning},
  pages={2242--2251},
  year={2019},
  organization={PMLR}
}

@book{desaulniers2006column,
  title={Column Generation},
  author={Desaulniers, Guy and Desrosiers, Jacques and Solomon, Marius M},
  volume={5},
  year={2006},
  publisher={Springer Science \& Business Media}
}

@misc{pyDVL_2025, title={py{DVL}: The python library for data valuation}, url={https://pydvl.org/stable/},author={TransferLab Team}, year={2025}, month={Mar}}

@article{kovalenkov2001exact,
  title={An exact bound on epsilon for nonemptiness of epsilon cores of games},
  author={Kovalenkov, Alexander and Wooders, Myrna Holtz},
  journal={Mathematics of Operations Research},
  volume={26},
  number={4},
  pages={654--678},
  year={2001},
  publisher={INFORMS}
}

@article{shapley1967balanced,
  title={On balanced sets and cores},
  author={Shapley, Lloyd S},
  journal={Naval Research Logistics Quarterly},
  volume={14},
  number={4},
  pages={453--460},
  year={1967},
  publisher={John Wiley \& Sons}
}

@inproceedings{jia2019towards,
  title={Towards efficient data valuation based on the {S}hapley value},
  author={Jia, Ruoxi and Dao, David and Wang, Boxin and Hubis, Frances Ann and Hynes, Nick and G{\"u}rel, Nezihe Merve and Li, Bo and Zhang, Ce and Song, Dawn and Spanos, Costas J},
  booktitle={The 22nd International Conference on Artificial Intelligence and Statistics},
  pages={1167--1176},
  year={2019},
  organization={PMLR}
}

@article{wang2024data,
  title={Data {S}hapley in one training run},
  author={Wang, Jiachen T and Mittal, Prateek and Song, Dawn and Jia, Ruoxi},
  journal={arXiv preprint arXiv:2406.11011},
  year={2024}
}

@article{cohen2024contextcite,
  title={Contextcite: Attributing model generation to context},
  author={Cohen-Wang, Benjamin and Shah, Harshay and Georgiev, Kristian and Madry, Aleksander},
  journal={Advances in Neural Information Processing Systems},
  volume={37},
  pages={95764--95807},
  year={2024}
}

@inproceedings{anderson2025making,
  title={Making {AI}-enhanced videos: Analyzing generative {AI} use cases in {Y}ou{T}ube content creation},
  author={Anderson, Torin and Niu, Shuo},
  booktitle={Proceedings of the Extended Abstracts of the CHI Conference on Human Factors in Computing Systems},
  pages={1--7},
  year={2025}
}

@article{kanna2024much,
  title={How much research is being written by large language models},
  author={Kanna, P},
  journal={Human-Centered Artificial Intelligence. Stanford University. Retrieved August},
  volume={16},
  year={2024}
}

@book{grotschel1988geometric,
  title={Geometric algorithms and combinatorial optimization},
  author={Gr{\"o}tschel, Martin and Lov{\'a}sz, L{\'a}szl{\'o} and Schrijver, Alexander},
  year={1988},
  publisher={Springer}
}

@article{kelley1960cutting,
  title={The cutting-plane method for solving convex programs},
  author={Kelley, Jr, James E},
  journal={Journal of the society for Industrial and Applied Mathematics},
  volume={8},
  number={4},
  pages={703--712},
  year={1960},
  publisher={SIAM}
}

@article{liu2024attribot,
  title={Attribot: A bag of tricks for efficiently approximating leave-one-out context attribution},
  author={Liu, Fengyuan and Kandpal, Nikhil and Raffel, Colin},
  journal={arXiv preprint arXiv:2411.15102},
  year={2024}
}

@inproceedings{qi2024model,
  title={Model internals-based answer attribution for trustworthy retrieval-augmented generation},
  author={Qi, Jirui and Sarti, Gabriele and Fern{\'a}ndez, Raquel and Bisazza, Arianna},
  booktitle={Proceedings of the 2024 Conference on Empirical Methods in Natural Language Processing},
  pages={6037--6053},
  year={2024}
}

@inproceedings{wang2025tracllm,
  title={$\{$TracLLM$\}$: A Generic Framework for Attributing Long Context $\{$LLMs$\}$},
  author={Wang, Yanting and Zou, Wei and Geng, Runpeng and Jia, Jinyuan},
  booktitle={34th USENIX Security Symposium (USENIX Security 25)},
  pages={3845--3864},
  year={2025}
}

@inproceedings{hirsch2025laquer,
  title={Laquer: Localized attribution queries in content-grounded generation},
  author={Hirsch, Eran and Slobodkin, Aviv and Wan, David and Stengel-Eskin, Elias and Bansal, Mohit and Dagan, Ido},
  booktitle={Proceedings of the 63rd Annual Meeting of the Association for Computational Linguistics (Volume 1: Long Papers)},
  pages={15355--15370},
  year={2025}
}

@inproceedings{shen2025transparentize,
  title={Transparentize the internal and external knowledge utilization in LLMs with trustworthy citation},
  author={Shen, Jiajun and Zhou, Tong and Chen, Yubo and Qiu, Delai and Liu, Shengping and Liu, Kang and Zhao, Jun},
  booktitle={Findings of the Association for Computational Linguistics: ACL 2025},
  pages={17858--17877},
  year={2025}
}

@inproceedings{xiao2025tokenshapley,
  title={TokenShapley: Token Level Context Attribution with Shapley Value},
  author={Xiao, Yingtai and Zhu, Yuqing and Samyoun, Sirat and Zhang, Wanrong and Wang, Jiachen T and Du, Jian},
  booktitle={Findings of the Association for Computational Linguistics: ACL 2025},
  pages={3882--3894},
  year={2025}
}

@inproceedings{chang2025jopa,
  title={JOPA: Explaining Large Language Model’s Generation via Joint Prompt Attribution},
  author={Chang, Yurui and Cao, Bochuan and Wang, Yujia and Chen, Jinghui and Lin, Lu},
  booktitle={Proceedings of the 63rd Annual Meeting of the Association for Computational Linguistics (Volume 1: Long Papers)},
  pages={22106--22122},
  year={2025}
}

@article{li2025attributing,
  title={Attributing response to context: A jensen-shannon divergence driven mechanistic study of context attribution in retrieval-augmented generation},
  author={Li, Ruizhe and Chen, Chen and Hu, Yuchen and Gao, Yanjun and Wang, Xi and Yilmaz, Emine},
  journal={arXiv preprint arXiv:2505.16415},
  year={2025}
}

@article{zhang2025taught,
  title={Who taught the lie? responsibility attribution for poisoned knowledge in retrieval-augmented generation},
  author={Zhang, Baolei and Xin, Haoran and Chen, Yuxi and Liu, Zhuqing and Yi, Biao and Li, Tong and Nie, Lihai and Liu, Zheli and Fang, Minghong},
  journal={arXiv preprint arXiv:2509.13772},
  year={2025}
}

@article{nematov2025source,
  title={Source attribution in retrieval-augmented generation},
  author={Nematov, Ikhtiyor and Kalai, Tarik and Kuzmenko, Elizaveta and Fugagnoli, Gabriele and Sacharidis, Dimitris and Hose, Katja and Sagi, Tomer},
  journal={arXiv preprint arXiv:2507.04480},
  year={2025}
}

@article{pan2025context,
  title={Context Attribution with Multi-Armed Bandit Optimization},
  author={Pan, Deng and Murugesan, Keerthiram and Moniz, Nuno and Chawla, Nitesh},
  journal={arXiv preprint arXiv:2506.19977},
  year={2025}
}

@article{patel2025maxshapley,
  title={MaxShapley: Towards Incentive-compatible Generative Search with Fair Context Attribution},
  author={Patel, Sara and Zhou, Mingxun and Fanti, Giulia},
  journal={arXiv preprint arXiv:2512.05958},
  year={2025}
}

@article{de2025improving,
  title={Improving Context-Attribution with Semi-Supervised Cross-Encoders},
  author={De Grandisa, Luca and Granataa, Francesco Maria and Costaa, Davide and Lanzaa, Antonio and Oroa, Ermelinda},
  journal={European Conference on Artificial Intelligence (ECAI)},
  year={2025}
}

@techreport{gemini3pro_modelcard_2025,
  title        = {Gemini 3 Pro Model Card},
  author       = {{Google DeepMind}},
  institution  = {Google DeepMind},
  year         = {2025},
  month        = {November},
  url          = {https://storage.googleapis.com/deepmind-media/Model-Cards/Gemini-3-Pro-Model-Card.pdf},
  note         = {Model card},
}

@techreport{gemini3flash_modelcard_2025,
  title        = {Gemini 3 Flash Model Card},
  author       = {{Google DeepMind}},
  institution  = {Google DeepMind},
  year         = {2025},
  month        = {December},
  url          = {https://storage.googleapis.com/deepmind-media/Model-Cards/Gemini-3-Flash-Model-Card.pdf},
  note         = {Model card},
}

@article{zheng2023judging,
  title={Judging llm-as-a-judge with mt-bench and chatbot arena},
  author={Zheng, Lianmin and Chiang, Wei-Lin and Sheng, Ying and Zhuang, Siyuan and Wu, Zhanghao and Zhuang, Yonghao and Lin, Zi and Li, Zhuohan and Li, Dacheng and Xing, Eric and others},
  journal={Advances in neural information processing systems},
  volume={36},
  pages={46595--46623},
  year={2023}
}


\newpage
\appendix
\section{Discussion: Cooperative Solution Concepts}\label{app:discussion}
The Shapley value and the core are two of the most canonical methods for determining payments in coalitional game settings. 
The Shapley value is notable for being the is the only allocation that simultaneously satisfies the following four properties: efficiency (the sum of Shapley values equals the total value), symmetry (equivalent agents receive equivalent values), linearity (if the value of every coalition is a linear combination of two value functions, then each player’s Shapley value is the same linear combination of their Shapley values under those functions), and null player (adding a dummy player does not change the Shapley values of others). 
However it is generally not an equilibrium credit allocation scheme.\looseness-1 

In contrast, the least core may be viewed as the cooperative analog of the Nash equilibrium, as under core payments the incentive of any subset of players to unilaterally deviate is minimized. 
The egalitarian least core also satisfies all Shapley properties, with the exception of linearity.
We choose to use the least core instead of the Shapley values as our solution concept, as the core better aligns with our goal of aligning creator incentives by preventing any subset of creators from being systematically under-compensated for their contribution to the AI-generated content item.\looseness-1 

\section{Appendix for Section~\ref{sec:sensitive}: Sensitivity of the Least Core}


\dbsb*
\begin{proof}
    By LP duality, we have
    \begin{equation*}
            \varepsilon^* = \max_{\vec{\lambda},\mu}\left\{\sum_{S\subseteq N}\lambda_Sv(S) - \mu : 
                \sum_{S\subseteq N} \lambda_S = 1,~\sum_{S\ni i}\lambda_S = \mu~\forall~ i\in N,~\lambda_S\ge 0~\forall~S\subseteq N\right\}
    \end{equation*} and
    \begin{equation*}
            \widehat{\varepsilon} = \max_{\vec{\lambda},\mu}\left\{\sum_{S\subseteq N}\lambda_S\widehat{v}(S) - \mu : 
                \sum_{S\subseteq N} \lambda_S = 1,~\sum_{S\ni i}\lambda_S = \mu~\forall~ i\in N,~\lambda_S\ge 0~\forall~S\subseteq N\right\}.
    \end{equation*}
    Let $\vec{\lambda}^*,\mu^*$ be a dual optimal solution for $\varepsilon^*$. 
    We have $$\varepsilon^* = \sum_{S\subseteq N}\lambda_S^*v(S) - \mu^*\quad\text{and}\quad\widehat{\varepsilon}\ge \sum_{S\subseteq N}\lambda_S^*\widehat{v}(S) - \mu^*.$$
    So, $$\varepsilon^* - \widehat{\varepsilon} \le \sum_{S\subseteq N}\lambda_S^*(v(S) - \widehat{v}(S)).$$ Symmetrically, considering a dual optimal solution $\widehat{\vec{\lambda}},\widehat{\mu}$ for $\widehat{\varepsilon}$, we get $$\widehat{\varepsilon}-\varepsilon^*\le\sum_{S\subseteq N}\widehat{\lambda}_S(\widehat{v}(S) - v(S)).$$
    Therefore $$\left|\varepsilon^* - \widehat{\varepsilon}\right|\le \max\left\{\sum_{S\subseteq N}\lambda_S^*(v(S) - \widehat{v}(S)), \sum_{S\subseteq N}\widehat{\lambda}_S(\widehat{v}(S) - v(S))\right\}\le \max_{\vec{\lambda}\in\Lambda}\sum_{S\subseteq N}\lambda_S|v(S) -\widehat{v}(S)|.$$
\end{proof}

\subsection{Supermodular Value Functions}\label{sec:super}

The general sensitivity bound developed above applies to arbitrary value functions, but it does not directly yield a constructive procedure for computing approximately stable allocations. 
In this subsection, we show that when the value function exhibits additional structure---specifically, supermodularity---it is possible to obtain approximate core allocations in closed form, without explicitly solving the least-core optimization problem.

Supermodularity corresponds to a setting with increasing marginal returns: adding a creator to a larger coalition provides at least as much incremental value as adding them to a smaller one.
Supermodularity implies that the core is non-empty in the exact case. 
We show that this structure can also be leveraged in the approximate setting.

\begin{restatable}{proposition}{supercore}\label{prop:supermodular-core}
    Suppose $v$ is supermodular, i.e. for all $ S, T\subseteq N$, $v( S)+v( T)\le v( S\cup T)+v( S\cap T)$. Let $S_i = \{1,\ldots, i\}$ and define the total variation distance along this chain by  
    $\textsc{TV}(v, \widehat{v}) = \frac{1}{2} \sum_{i=1}^N |v(S_i) - \widehat{v}(S_i)|$. 
    Finally, define the allocation $\widehat{\vec{u}} \in \mathbb{R}^N$ by the marginal increments
    $\widehat{u}_i = \widehat{v}(S_i) - \widehat{v}(S_{i-1})$. 
    Then, $\widehat{\vec{u}}\in\core_{4\textsc{TV}(v, \widehat{v})}$.   
\end{restatable}

\begin{proof}
    The proof is a direct generalization of the proof that a supermodular value function induces a nonempty core~\citep{shapley1971cores}.
    By definition, $\sum_{j\in N}\widehat{u}_j = \widehat{v}( N) = 1$. 
    Let $i\in N$ be minimal such that $ T\subseteq\{1,\ldots, i\}$. 
    We show that, for all $ T\subseteq N$, $\sum_{j\in T}\widehat{u}_j\ge v( T) - \sum_{j =1}^i |v(S_j) - \widehat{v}(S_j)| - \sum_{j =1}^{i-1} |v(S_j) - \widehat{v}(S_j)|$ by induction on $| T|$ (which implies $\widehat{\vec{u}}\in\core_{4\textsc{TV}(v, \widehat{v})}$). 
    The base case of $ T=\emptyset$ holds trivially. Consider now a general $ T$. 
    We have 
    $$\begin{aligned}v( T)&\le v( T\cup S_{i-1}) + v( T\cap S_{i-1}) - v(S_{i-1}) \\ 
    &=v(S_i) + v( T\setminus\{i\}) - v(S_{i-1}) \\ 
    &\le \widehat{v}(S_i) - \widehat{v}(S_{i-1}) + v( T\setminus\{i\}) + |v(S_i) - \widehat{v}(S_i)| + |v(S_{i-1}) - \widehat{v}(S_{i-1})|\\ 
    &= \widehat{u}_i + v( T\setminus\{i\}) + |v(S_i) - \widehat{v}(S_i)| + |v(S_{i-1}) - \widehat{v}(S_{i-1})|\\
    &\le\sum_{j\in T} \widehat{u}_j + \sum_{j =1}^i |v(S_j) - \widehat{v}(S_j)| + \sum_{j =1}^{i-1} |v(S_j) - \widehat{v}(S_j)|,
    \end{aligned}$$ as desired, where the final inequality is due to the inductive hypothesis.
\end{proof}

This result shows that a simple allocation constructed from successive marginal contributions under $\widehat{v}$ is guaranteed to be approximately stable with respect to the true value function $v$, with approximation quality controlled by how well $\widehat{v}$ matches $v$ along the chain.

Next, we consider the complementary setting in which the estimated value function $\widehat{v}$ is supermodular.

\begin{restatable}{proposition}{superestcore}\label{prop:supermodular-estimated-core}
        Suppose $\widehat{v}$ is supermodular. Define $\delta$ and $\widehat{u}_i$ as in Proposition~\ref{prop:supermodular-core}. Then, $\widehat{u}\in\core_{\delta}$.
\end{restatable}

\begin{proof}
    Supermodularity of $\widehat{v}$ implies supermodularity of $\widehat{v}$ which implies $\widehat{\core}_0\neq\emptyset$ with $\widehat{\vec{u}}\in\widehat{\core}_{0}$. By Corollary~\ref{cor:core-containment}, $\widehat{\vec{u}}\in\core_{\delta}$.
\end{proof}

In this case, supermodularity of $\widehat{v}$ guarantees that the constructed allocation lies exactly in the estimated core, and therefore, by the general sensitivity bound, it is also approximately stable with respect to the true value function. 
Taken together, these results show that supermodularity enables efficient and direct construction of approximately optimal allocations without solving the full least-core LP.
\section{Appendix for Section~\ref{sec:comp}: Computing the Least Core}

\begin{algorithm}[t]
\SetAlgoLined
\textbf{Input:} number of samples $m$, value function estimate $\widehat{v}$, value-sharing scheme $\vec{u}$, $\varepsilon \geq 0$\\
\For{$j=1, \ldots, m$}{
    Sample $ S_j \sim \cD$\\
    Return $ S_j$ if $\sum_{i \in  S_{j}}  u_i + \varepsilon < \widehat{v}(S_j)$

}
\caption{GetUnhappyCoalition: sampling-based}\label{alg:samp-sep}
\end{algorithm}

\begin{algorithm}[t]
\SetAlgoLined
\textbf{Input:} number of samples $m$, value function estimate $\widehat{v}$, value-sharing scheme $u$, $\varepsilon \geq 0$\\
\For{$j=1, \ldots, m$}{
    Ask the LLM which subset it believes is most violated given $u$. Denote this subset by $S_j$.\\
    Return $ S_j$ if $\sum_{i \in  S_{j}}  u_i + \varepsilon < \widehat{v}(S_j)$

}
\caption{GetUnhappyCoalition: LLM-based}\label{alg:llm-sep}
\end{algorithm}

\begin{algorithm}[t]
\SetAlgoLined

\caption{Ellipsoid method for the least estimated core}\label{alg:ellipsoid}
Set $T = 2(n+1)(n+2)\ln(R/r)$ and initialize ellipsoid $E$ to be the ball $\{x\in\mathbb R^{n+1}:\|x-c\|_2\le R\}$, where $c=(\frac{1}{2},\dots,\frac{1}{2})\in\mathbb R^{n+1}$ and $R=\frac{\sqrt{n+1}}{2}$. 

\For{$t=1,2,\dots,T$}{
    Let $(\vec{u}^{(t)},\varepsilon^{(t)})$ be the center of $E$

    \If{$\sum_{i\in N} u_i^{(t)} \neq 1$}{
        Use the violated hyperplane $\sum_{i\in N} u_i = 1$ 
        to update $E$

        \textbf{continue}
    }

    \If{$u_i^{(t)} \notin [0,1]$ for some $i$, or $\varepsilon^{(t)} \notin [0,1]$}{
        Use the corresponding constraint to update $E$

        \textbf{continue}
    }

    $S^{(t)} \leftarrow \mathrm{GetUnhappyCoalition}(\vec{u}^{(t)},\varepsilon^{(t)})$

    \If{$S^{(t)} \neq \mathtt{None}$}{
        Use the violated hyperplane $\sum_{i\in S^{(t)}} u_i + \varepsilon = \widehat{v}(S^{(t)})$
        to update $E$
        
        \textbf{continue}
    }

    \Return $(u^{(t)},\varepsilon^{(t)})$
}
\end{algorithm}

%

\corellipsoid*
\begin{proof}
    The classical central-cut ellipsoid method in dimension $d$ makes at most $2d(d+1)\log(R/r)$ oracle calls; in our LP, $d=n+1$, so this is at most $2(n+1)(n+2)\ln(R/r)$.
    Here, $R$ is the radius of a ball containing $\mathcal{K} := \{(\vec{u}, \epsilon) \; : \; \vec{u} \in \Delta^{N}, \; \epsilon \in (0, 1)\}$
    and $r$ is the radius of a ball contained within the relative interior of $\mathcal{K}$. 
    By standard bit complexity arguments, $R/r = O(\mathrm{poly}(n))$. 

    In any given round, the probability of the sampling-based separation oracle failing is at most $(1 - \delta)^m$. 
    Taking a union bound over the $O(n^2 \log(R/r))$ calls to the separation oracle, the probability of failure $\gamma$ is at most $O((1 - \delta)^m n^2 \log(n))$. 
    Rearranging terms and solving for $m$ yields the desired sample size in the sampling-based separation oracle.\looseness-1 

    If the ellipsoid method terminates at $(\vec{u}^{(\alg)},\varepsilon^{(\alg)})$ under the clean event, then it must be the case that
    \begin{equation*}
        p(\vec{u}^{(\alg)},\varepsilon^{(\alg)}) = \mathbb{P}_{S\sim D} \left[\sum_{i\in S} u_i^{(\alg)} + \varepsilon^{(\alg)} < \widehat{v}(S)\right] < \delta,
    \end{equation*}
    because otherwise the final oracle call would have returned a violated coalition instead of allowing termination. 
    Since $\varepsilon \leq \widehat{\varepsilon}$, $$\mathbb{P}_{S \sim \cD}\left[ \sum_{i \in S} u_i^{(\alg)} + \varepsilon^{(\alg)} \geq \widehat{v}(S) \right] \leq \mathbb{P}_{S \sim \cD}\left[ \sum_{i \in S} u_i^{(\alg)} + \widehat{\varepsilon} \geq \widehat{v}(S) \right],$$ and so  
    \begin{equation*}
        \mathbb{P}_{ S \sim \cD}\left[\sum_{i \in  S} u_i^{(\alg)} + \widehat{\epsilon} \geq \widehat{v}( S) \right] \geq 1 - \delta.
    \end{equation*}
\end{proof}

\elipsqp*
\begin{proof}
Introduce an auxiliary variable $t$ and rewrite the quadratic program in epigraph form:
\[
\min_{\vec{u},t} \ t
\]
subject to
\[
\vec{u} \in [0,1]^N,\qquad \sum_{i\in N} u_i = 1,
\qquad
\sum_{i\in S} u_i \ge \hat v(S)-\varepsilon^{(\alg)}\ \ \forall S \subset N,
\qquad
t \ge \frac12\|\vec{u}\|_2^2.
\]
This is a convex optimization problem in dimension $n+1$.

We first describe a separation oracle for the epigraph constraint $t \ge \frac{1}{2}\|\vec{u}\|_2^2$:
If a candidate solution $(\vec{u}_0,t_0)$ satisfies $t_0 \ge \frac{1}{2}\|\vec{u}_0\|_2^2$, then it is feasible for the epigraph. 
Otherwise, since the function $f(\vec{u}) := \frac{1}{2}\|\vec{u}\|_2^2$ is convex and differentiable with gradient $\nabla f(\vec{u}_0)=\vec{u}_0$, the first-order supporting hyperplane inequality gives
\[
f(\vec{u}) \ge f(\vec{u}_0) + \langle \vec{u}_0, \vec{u}-\vec{u}_0\rangle
= \vec{u}_0^\top \vec{u} - \frac{1}{2}\|\vec{u}_0\|_2^2
\qquad \forall \vec{u} \in \mathbb{R}^N.
\]
Therefore every point $(\vec{u},t)$ in the epigraph satisfies $t \geq \vec{u}_0^\top \vec{u} - \frac{1}{2}\|\vec{u}_0\|_2^2$. 
Since $(\vec{u}_0,t_0)$ violates $t_0 \geq \frac{1}{2}\|\vec{u}_0\|_2^2$, we have
\[
\vec{u}_0^\top \vec{u}_0 - t_0 > \frac{1}{2}\|\vec{u}_0\|_2^2,
\]
so $(\vec{u}_0,t_0)$ lies strictly outside the halfspace, which therefore separates it from the epigraph.
Hence the quadratic epigraph admits an exact deterministic separation oracle.

The remaining constraints are the coalition constraints
\[
\sum_{i\in S} u_i \geq \widehat{v}(S)-\varepsilon^{(\alg)} \qquad \forall S\subseteq N.
\]
For these, use the same sampling-based oracle as in Theorem~\ref{thm:ellipse}, now with $\varepsilon^{(\alg)}$ fixed. Define
\[
p(\vec{u}) := \mathbb{P}_{S\sim D}\!\left[\sum_{i\in S} u_i + \varepsilon^{(\alg)} < \widehat{v}(S)\right].
\]
Exactly as in the proof of Theorem~4.2, if $p(\vec{u})\geq \delta$, then one call to Algorithm~\ref{alg:samp-sep} fails to find a violated coalition with probability at most $(1-\delta)^m$. The classical central-cut ellipsoid method makes at most $O(n^2\log(R/r))$ oracle calls, so choosing
\[
m = O\left(\frac{\log(n^2\log n/\gamma)}{\delta}\right)
\]
and taking a union bound implies that with probability at least $1-\gamma$, every oracle call behaves correctly whenever the current point violates a $\delta$-mass of coalition constraints.

Condition on this clean event. If the ellipsoid method terminates at $(\vec{u}^{(\mathrm{egal})}, t^{(\mathrm{egal})})$, then it cannot be the case that $p(\vec{u}^{(\mathrm{egal})}) \geq \delta$, because otherwise the final oracle call would have returned a violated coalition constraint rather than allowing termination. Therefore
\[
\Pr_{S\sim D}\!\left[\sum_{i\in S} u^{(\mathrm{egal})}_i + \varepsilon^{(\alg)} \geq \widehat{v}(S)\right] \geq 1-\delta.
\]
This proves the $\delta$-probable least estimated-core guarantee.

Because the optimization is carried out in epigraph form with objective $\min t$, any optimal solution satisfies
\[
t^{(\mathrm{egal})} = \frac{1}{2}\|\vec{u}^{(\mathrm{egal})}\|_2^2.
\]
Indeed, if strict inequality held, one could decrease $t$ while preserving feasibility, contradicting optimality. 
Thus $(\vec{u}^{(\mathrm{egal})}, t^{(\mathrm{egal}}))$ minimizes the $\ell_2$ norm among allocations satisfying the sampled ellipsoid feasibility conditions.
Taking a final union bound over both runs of ellipsoid obtains the desired result. 
\end{proof}

\begin{restatable}{theorem}{cgthm}\label{thm:cg}
    Let $\cD$ be a distribution over coalitions. Consider running CG (Alg.~\ref{alg:cg-qp}) instantiated such that $\mathrm{GetUnhappyCoalition}$ is the sampling-based separation oracle (Alg.~\ref{alg:samp-sep}) with number of samples $m$, and let $\tau \leq T$ be the round on which it terminates with solution $(u^{(*, \tau)}, \epsilon^{(*, \tau)})$. 
    Then if $m \geq \frac{1}{\delta} \log(T/\gamma)$, with probability at least $1 - \gamma$, $\mathbb{P}_{ S \sim \cD}\left[\sum_{i \in  S} u_i^{(*, \tau)} + \widehat{\varepsilon} \geq \widehat{v}( S) \right] \geq 1 - \delta.$
    %
    %
    Furthermore, $\vec{u}^{(*, \tau)}$ minimizes $\| \vec{u} \|_2^2$ among all feasible value-sharing schemes defined by the $\tau-1$ sampled constraints.\looseness-1 
\end{restatable}
%
\begin{proof}
    Let
    \begin{equation*}
        p(u, \epsilon) := \mathbb{P}_{ S \sim \cD} \left[ \sum_{i \in  S} u_i + \epsilon < \widehat{v}(S) \right] 
    \end{equation*}
    be the probability that a core constraint (under $\widehat{v}$) is violated by more than $\epsilon$, and suppose we'd like to control $p(u, \epsilon)$ such that it is at most $\delta$. 
    Consider the run of Algorithm~\ref{alg:samp-sep} for round $t$ of constraint generation and condition on the entire history up to round $t$. 
    It is ok if Algorithm~\ref{alg:samp-sep} fails if $p(u^{(*, t)}, \epsilon^{(*, t)}) < \delta$, as this does not affect the high-probability guarantee of the $\delta$-probable core. 
    Therefore, assume that $p(u^{(*, t)}, \epsilon^{(*, t)}) \geq \delta$. 
    The probability of Algorithm~\ref{alg:samp-sep} failing in round $t$ is  
    $(1 - p(u^{(*, t)}, \epsilon^{(*, t)}))^m \leq (1 - \delta)^m$, which we'd like to be at most $\gamma / T$. 
    Therefore, if 
    \begin{equation*}
        m \geq \frac{\log(T/\gamma)}{-\log(1 - \delta)} \geq \frac{\log(T/\gamma)}{\delta},
    \end{equation*}
    then one oracle call will fail with probability at most $\gamma/T$. 

    Taking a union bound over all $T$, we have that if $m \geq \frac{1}{\delta} \log(T/\gamma)$, then the probability of any of the $T$ oracle calls failing is at most $\gamma$. 
    %
\end{proof}

\lemmwc*
\begin{proof}
    Seed the least core LP with constraints corresponding to $S_1,\ldots, S_k$ and solve. 
    Let $S$ be any winning coalition, and let $S_j$ be a MWC such that $S_j\subseteq S$. 
    For any $\vec{u}\in [0,1]^N$, we have $\widehat{v}(S) - \sum_{i\in S} u_i = 1 - \sum_{i\in S} u_i \le 1 - \sum_{i\in S_j} u_i = \widehat{v}(S_j) - \sum_{i\in S_j}u_i$. 
    So, $\widehat{v}(S_j) - \sum_{i\in S_j} u_i\le\varepsilon \Rightarrow \widehat{v}(S) - \sum_{i\in S} u_i\le\varepsilon$ for any $\varepsilon > 0$ and therefore the constraint corresponding to $S$ is redundant.
    The same reasoning extends to egalitarian least core computation (and we obtain a polynomial run-time since the quadratic program has a convex objective).    
\end{proof}

\mwcpred*
\begin{proof}
    Seed the least core LP with constraints corresponding to $S_1,\ldots, S_k$ and run constraint generation with the given oracle. The same reasoning as in Lemma~\ref{lem:mwc} shows that for any $\vec{u}$, the oracle returns a MWC. Since there are at most $\eta$ MWCs that are not already in the LP, CG terminates in at most $\eta$ iterations. The same reasoning extends to egalitarian least core computation (and we obtain a polynomial run-time since the quadratic program has a convex objective).
\end{proof}
\section{Additional Experimental Details}
\label{app:experiments}

\begin{algorithm}[t]
\SetAlgoLined
Initialize $\cS^0\leftarrow\emptyset$\\
\For{$t = 1,\ldots, T$}{
    Sample $S^t \sim \cD$\\
    $\cS^t \leftarrow \cS^{t-1} \cup S^t$
}
    Solve $\widehat{\mathrm{LP}}(\cS^T)$; let $\varepsilon$ be the returned optimal deficit \\
    Solve $\widehat{\mathrm{QP}}(\cS^T,\varepsilon)$; let $\vec{u}$ be the returned value-sharing scheme\\
{\bf return} $\vec{u},\varepsilon$
\caption{\textbf{YP} baseline}\label{alg:yp}
\end{algorithm}

\subsection{BrowseComp-Plus Dataset Details}
\label{app:browsecomp-details}

As an example BrowseComp-Plus query: \emph{``A CEO who founded a company in the mid-1990s was raised in Southern Africa. His father was an engineer, and their relationship was not delightful. His first child sadly passed away as a result of SIDS. In the early 2020s, he had a child whose name had the internet buzzing. This CEO has a younger sister who co-founded a film streaming service. In an article from the early 2020s, at what age did the mother of the lady who claimed to be related to him give birth to her?''}

A \emph{gold} document for this question explicitly states, \emph{``...Her mother was 21...''}. A \emph{negative} document contains information that appears relevant at first glance---e.g., \emph{``...he shares 21-year-old twins...''}---which does not actually answer the question but contains salient keywords. An \emph{evidence} document contains, for example, \emph{``Tosca is the co-founder of a streaming service,''} which helps narrow down the answer.

Ideally, gold documents should receive large values in $\mathbf{u}$, evidence documents should receive nonzero but moderate values, and negative documents should receive zero. However, in our preliminary investigation, we found several violations of this intuition. In the example above, a gold document is repeated verbatim in the evidence set, violating our uniqueness assumption. Furthermore, at least one negative document contains facts relevant to the query, making it not truly negative for our purposes. We observed several similar issues: for instance, gold documents that do not contain the answer (which, in some cases, nonetheless evoke the correct response from some language models).

\subsection{Additional Experiment Discussion}
\label{app:more-discussion}

\paragraph{Compute Resources.} All value function estimation and synthetic dataset generation was performed using the Gemini 3 API.
The constraint generation algorithms and LP/QP optimizations were executed locally. Running experiments for the $n=32$ case
with 32 rows takes approximately 16 hours due to API throughput and the large number of coalitions tested for baselines.
We developed the method at small scale on a handful of held-out examples.

\paragraph{Limitations.} One limitation of our approach is that our \textbf{CG} methods require many calls, which could be expensive (though far cheaper than baselines). We speculate that the LLM compute required for such counterfactual estimation will become more abundant in the future, together with further algorithmic
improvements in efficiently estimating the core. A limitation of our experiment section is that an exact ground truth is impossible
to compute for large $n$ (though we ran the exact version for the $n=10$ case). Calculating the baseline and holdout coalition values
was a significant bottleneck for our experiments, resulting in a relatively low sample size for our study. We also found it difficult to find or construct a dataset with true interaction effects within coalitions, as language models are often able to infer the correct answer from partial information given their expansive training corpora. However, in such cases, it may be reasonable to consider the reward to be the marginal value provided to the LLM on top of its world knowledge. Finally, we could only investigate our method using one frontier language model due to the availability of API keys with lenient rate limits, though we are confident results would be similar with others.

\newpage
\subsection{Synthetic $n{=}32$ Results}
\label{app:synth-n32}

The aggregate Synthetic $n{=}32$ results are shown in \cref{fig:forest-synth}. Below we present selected convergence traces organized by \textbf{ZS} performance, using \textbf{YP}'s final holdout $\hat{\varepsilon}_{\mathcal{H}}$ as the reference point for each instance.
To report \textbf{YP}'s performance over the total number of calls, we incrementally build the full set of coalitions
by adding random batches of size $2^7$. Note that this method is not subject to same guarantees or stopping conditions as our \textbf{CG} algorithms,
and is included for illustrative purposes only.

\begin{figure*}[!htb]
    \centering
    \includegraphics[width=\textwidth]{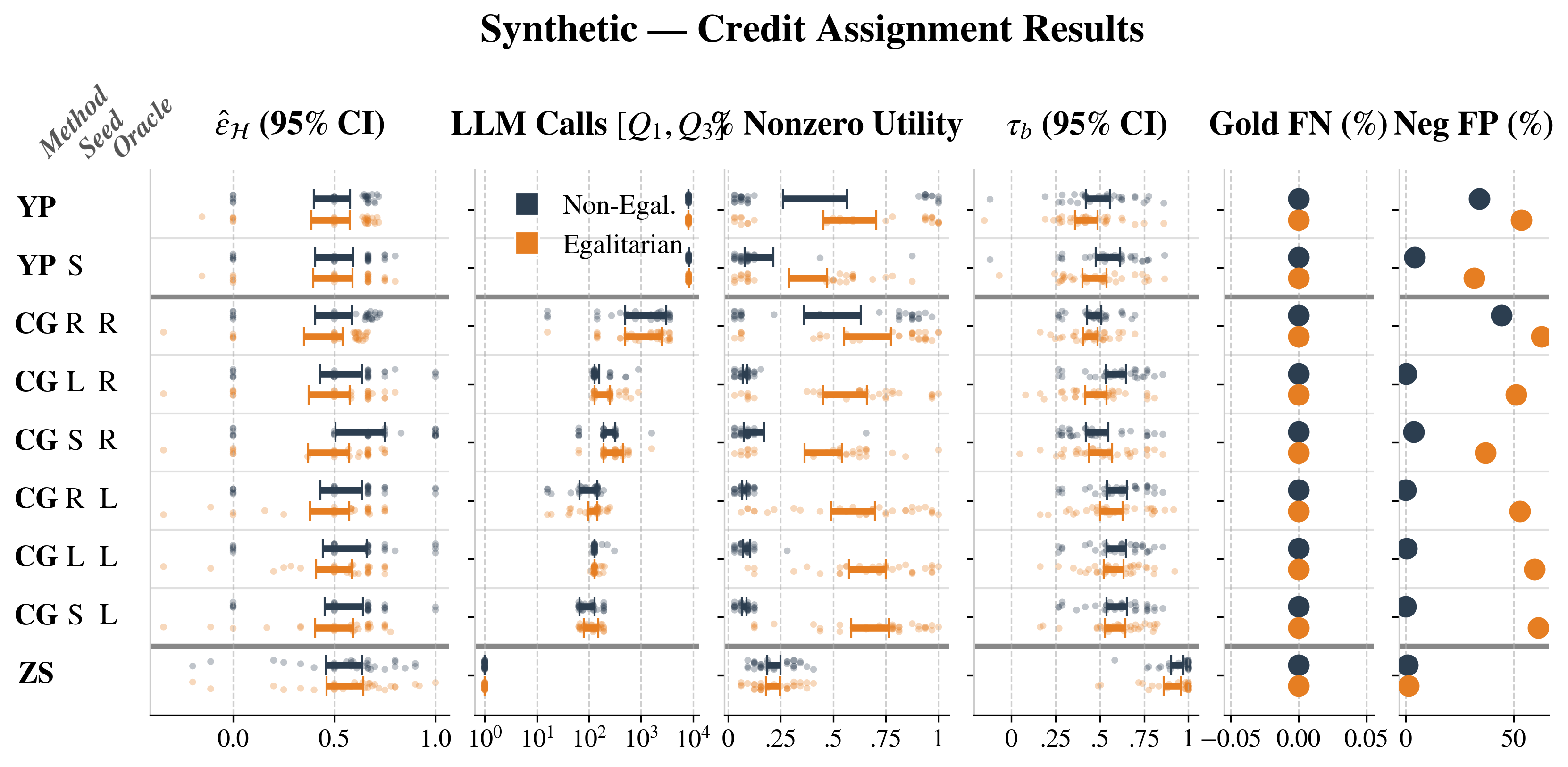}
    \caption{Credit assignment results on Synthetic ($n{=}32$) for non-egalitarian (dark) and egalitarian (light) settings. Lower $\hat{\varepsilon}_{\mathcal{H}}$ is better. Bars show 95\% bootstrap CIs; points are individual problem instances. Compare with \cref{fig:forest} for BrowseComp-Plus.}
    \label{fig:forest-synth}
\end{figure*}

\newpage
\subsubsection{Cases Where ZS Fails}

These instances (\cref{fig:zs-fail}) demonstrate that \textbf{ZS} can produce poor allocations, motivating the need for \textbf{CG}'s iterative refinement with guarantees. In Row~2, \textbf{YP} reaches $\hat{\varepsilon}_{\mathcal{H}} = 0.67$ and \textbf{CG(R,R)} matches it (within 0.04) in ${\sim}2{,}000$ calls, while \textbf{ZS} overshoots substantially at $\hat{\varepsilon}_{\mathcal{H}} = 0.90$ (+0.23 above \textbf{YP}). Row~21 shows a similar pattern: \textbf{CG(R,R)} reaches within 0.04 of \textbf{YP}'s baseline, but \textbf{ZS} overshoots by +0.23.

\begin{figure*}[!htb]
    \centering
    \begin{subfigure}[t]{0.49\textwidth}
        \includegraphics[width=\textwidth]{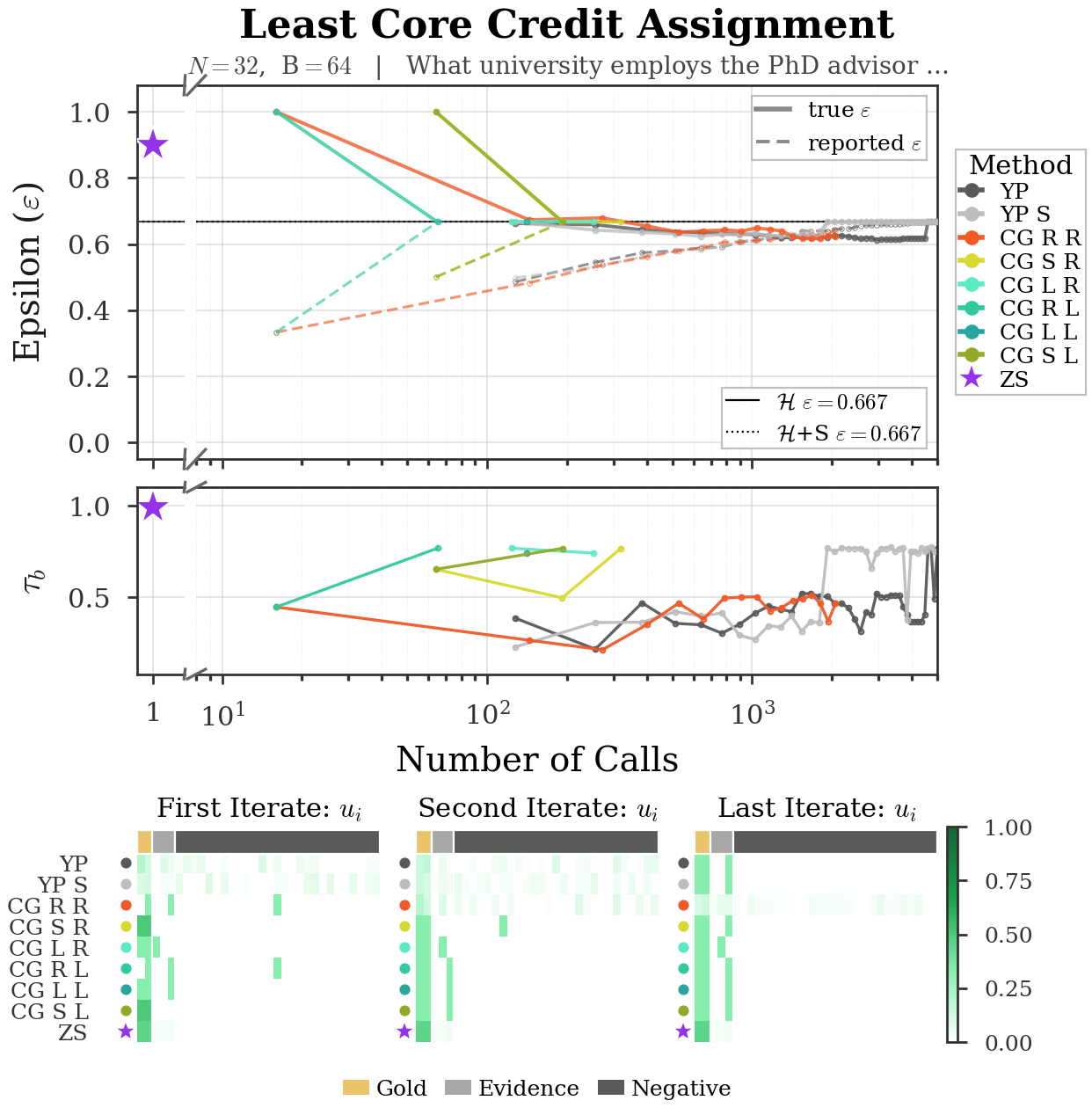}
        \caption{Synthetic $n{=}32$, Row~2 (non-egal.)}
    \end{subfigure}\hfill
    \begin{subfigure}[t]{0.49\textwidth}
        \includegraphics[width=\textwidth]{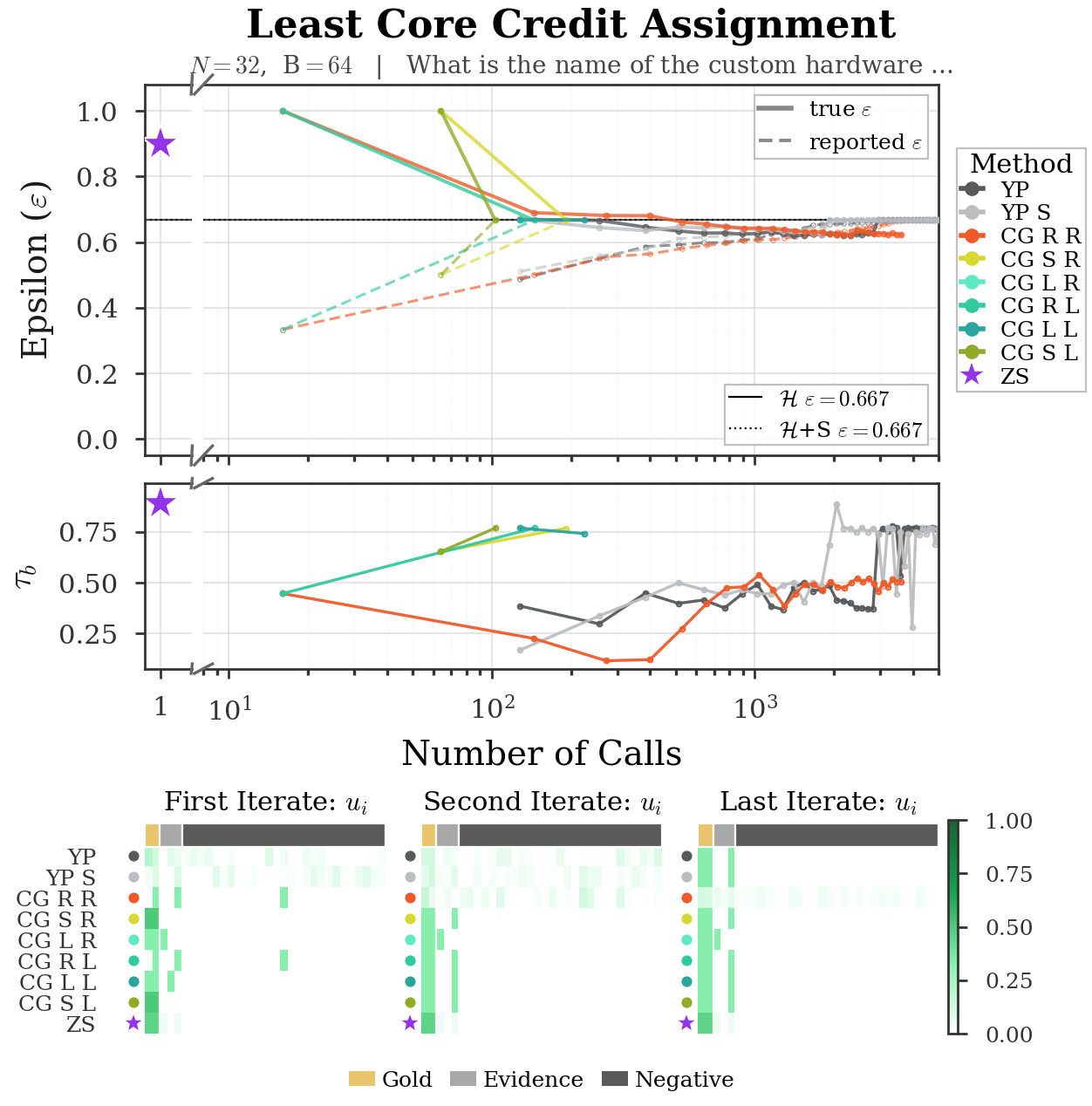}
        \caption{Synthetic $n{=}32$, Row~21 (non-egal.)}
    \end{subfigure}
    \caption{\textbf{ZS} failure cases: \textbf{ZS} overshoots \textbf{YP}'s holdout $\hat{\varepsilon}_{\mathcal{H}}$ by +0.23, while \textbf{CG(R,R)} matches the baseline with 3--4$\times$ fewer calls.}
    \label{fig:zs-fail}
\end{figure*}

\newpage
\subsubsection{Cases Where ZS Succeeds}

These instances (\cref{fig:zs-success}) show that \textbf{ZS} can sometimes produce remarkably good allocations in a single call. In Row~0, all methods---including \textbf{ZS}---match \textbf{YP}'s $\hat{\varepsilon}_{\mathcal{H}} = 0.50$ exactly, suggesting this is an easy instance. In Row~3, \textbf{ZS} again matches \textbf{YP} exactly at $\hat{\varepsilon}_{\mathcal{H}} = 0.50$, and all \textbf{CG} variants converge to the same optimum, confirming the allocation is tight.

\begin{figure*}[!htb]
    \centering
    \begin{subfigure}[t]{0.49\textwidth}
        \includegraphics[width=\textwidth]{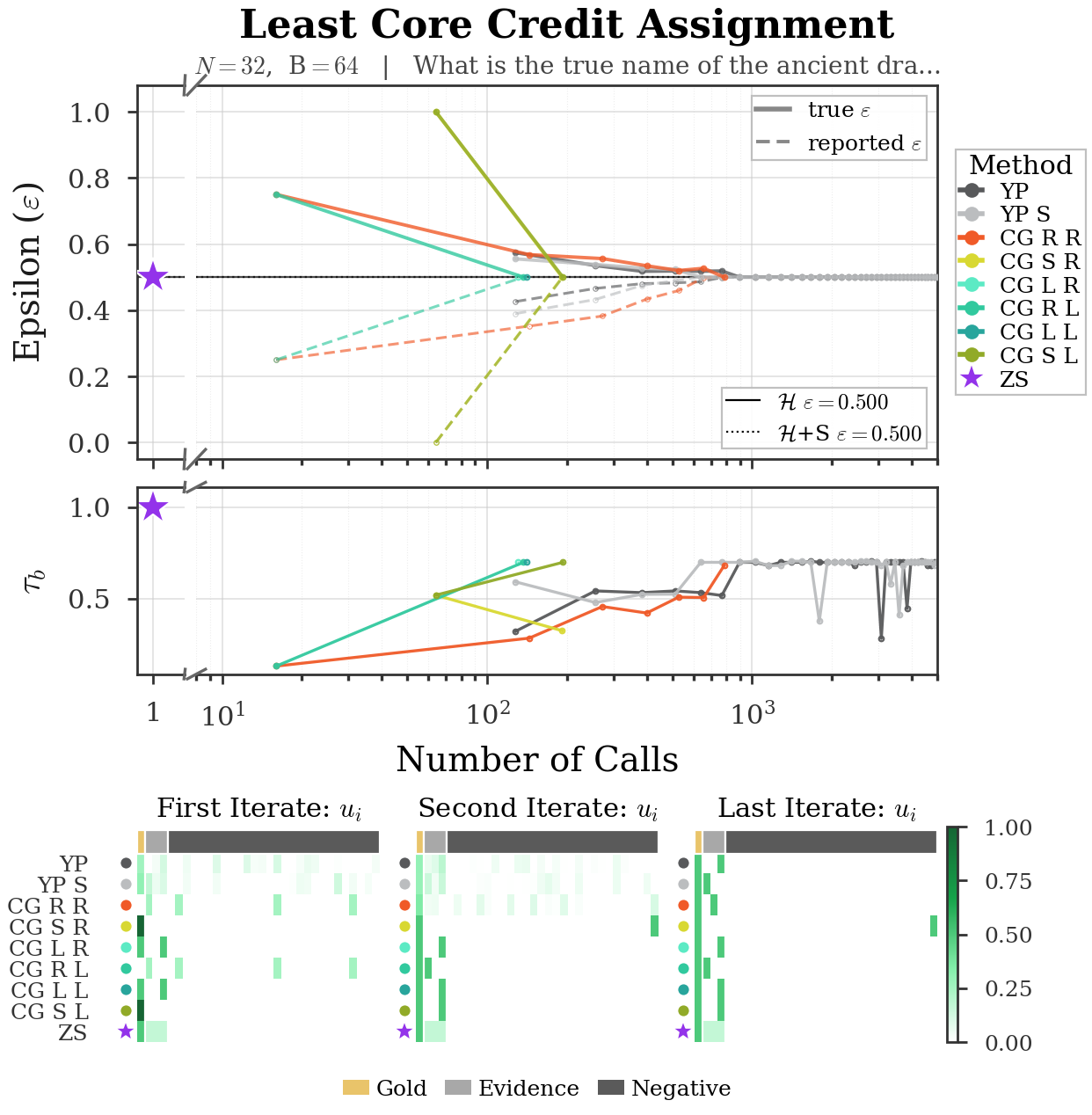}
        \caption{Synthetic $n{=}32$, Row~0 (non-egal.)}
    \end{subfigure}\hfill
    \begin{subfigure}[t]{0.49\textwidth}
        \includegraphics[width=\textwidth]{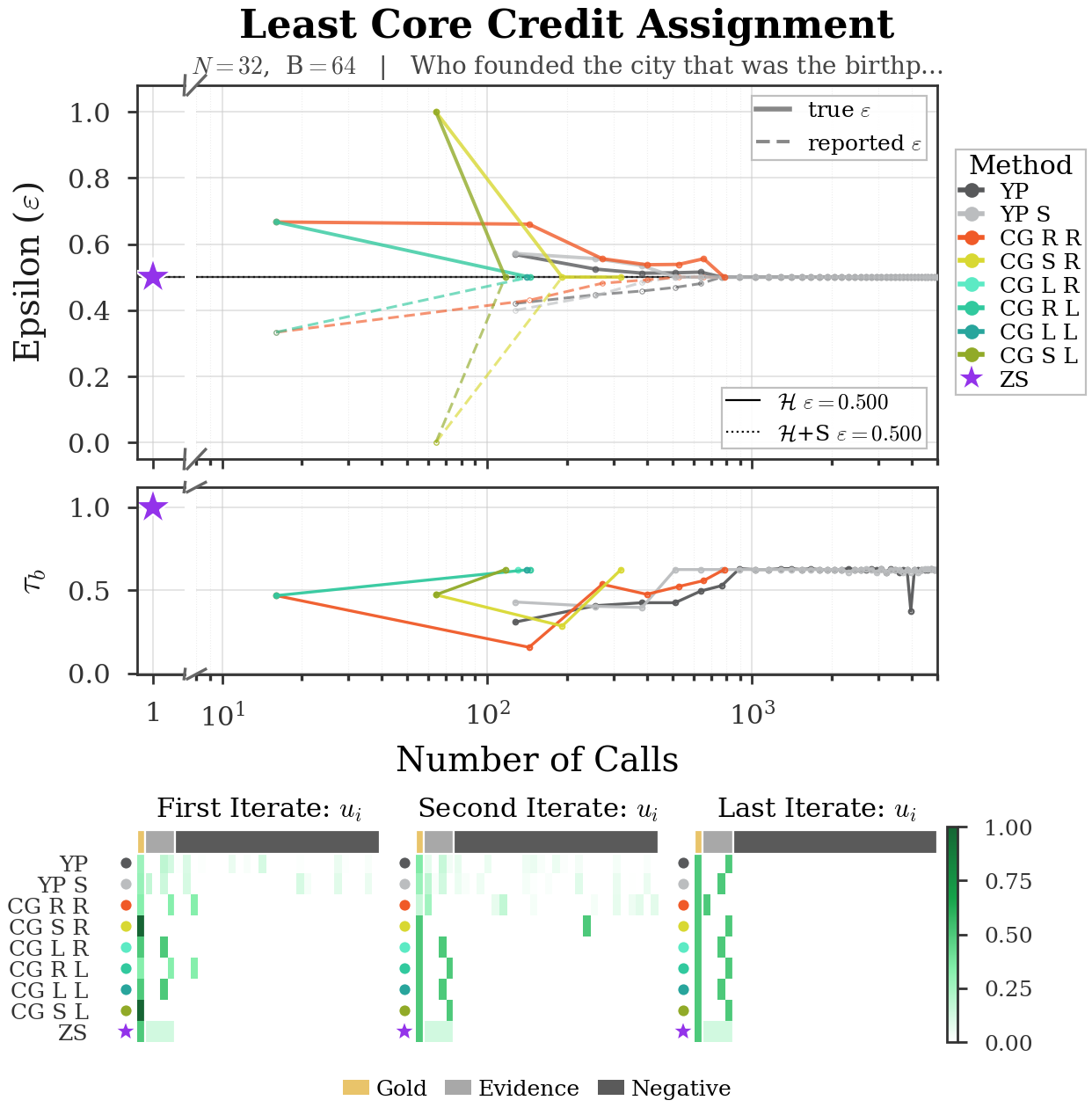}
        \caption{Synthetic $n{=}32$, Row~3 (non-egal.)}
    \end{subfigure}
    \caption{\textbf{ZS} success cases: \textbf{ZS} matches \textbf{YP}'s holdout $\hat{\varepsilon}_{\mathcal{H}} = 0.50$ in a single call, with high $\tau_b$.}
    \label{fig:zs-success}
\end{figure*}

\newpage
\subsubsection{Egalitarian Effects}

These traces (\cref{fig:egal-effects-1,fig:egal-effects-2}) compare non-egalitarian (left) and egalitarian (right) variants on the same problem instance, showing how the QP objective redistributes credit without substantially degrading $\hat{\varepsilon}_{\mathcal{H}}$. In Row~2, the egalitarian \textbf{CG(R,R)} spreads utility more broadly (visible in the heatmaps) while remaining within 0.04 of \textbf{YP}'s baseline. Rows~4, 7, and 13 exhibit the same pattern: the egalitarian heatmaps show less sparse allocations, and $\hat{\varepsilon}_{\mathcal{H}}$ is either unchanged or increases by at most ${\sim}0.07$.

\begin{figure*}[!htb]
    \centering
    \begin{subfigure}[t]{0.49\textwidth}
        \includegraphics[width=\textwidth]{appendix_figs/traces/synth_N32_row2_nonegal.png}
        \caption{Row~2 (non-egal.)}
    \end{subfigure}\hfill
    \begin{subfigure}[t]{0.49\textwidth}
        \includegraphics[width=\textwidth]{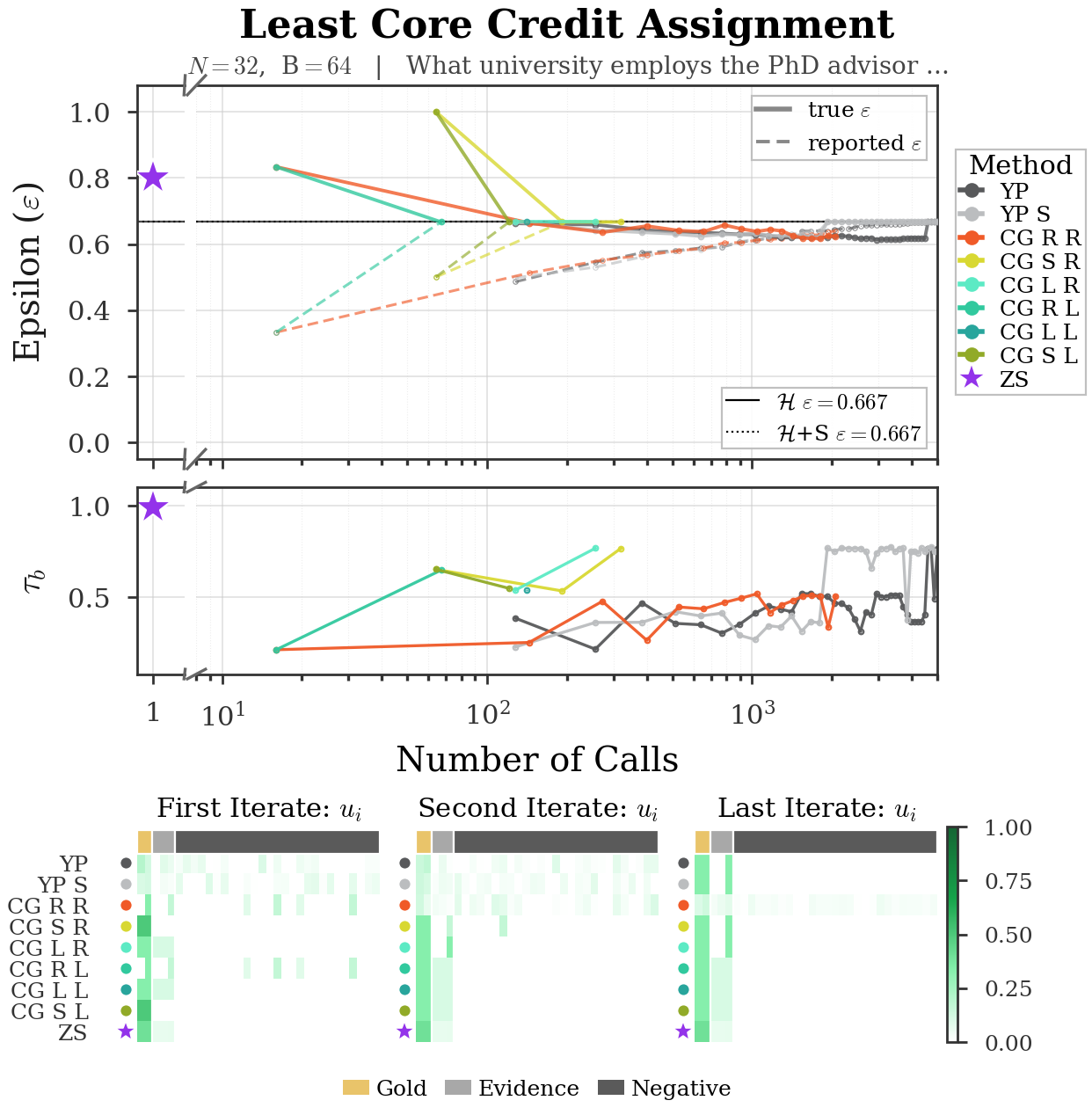}
        \caption{Row~2 (egal.)}
    \end{subfigure}\\[6pt]
    \begin{subfigure}[t]{0.49\textwidth}
        \includegraphics[width=\textwidth]{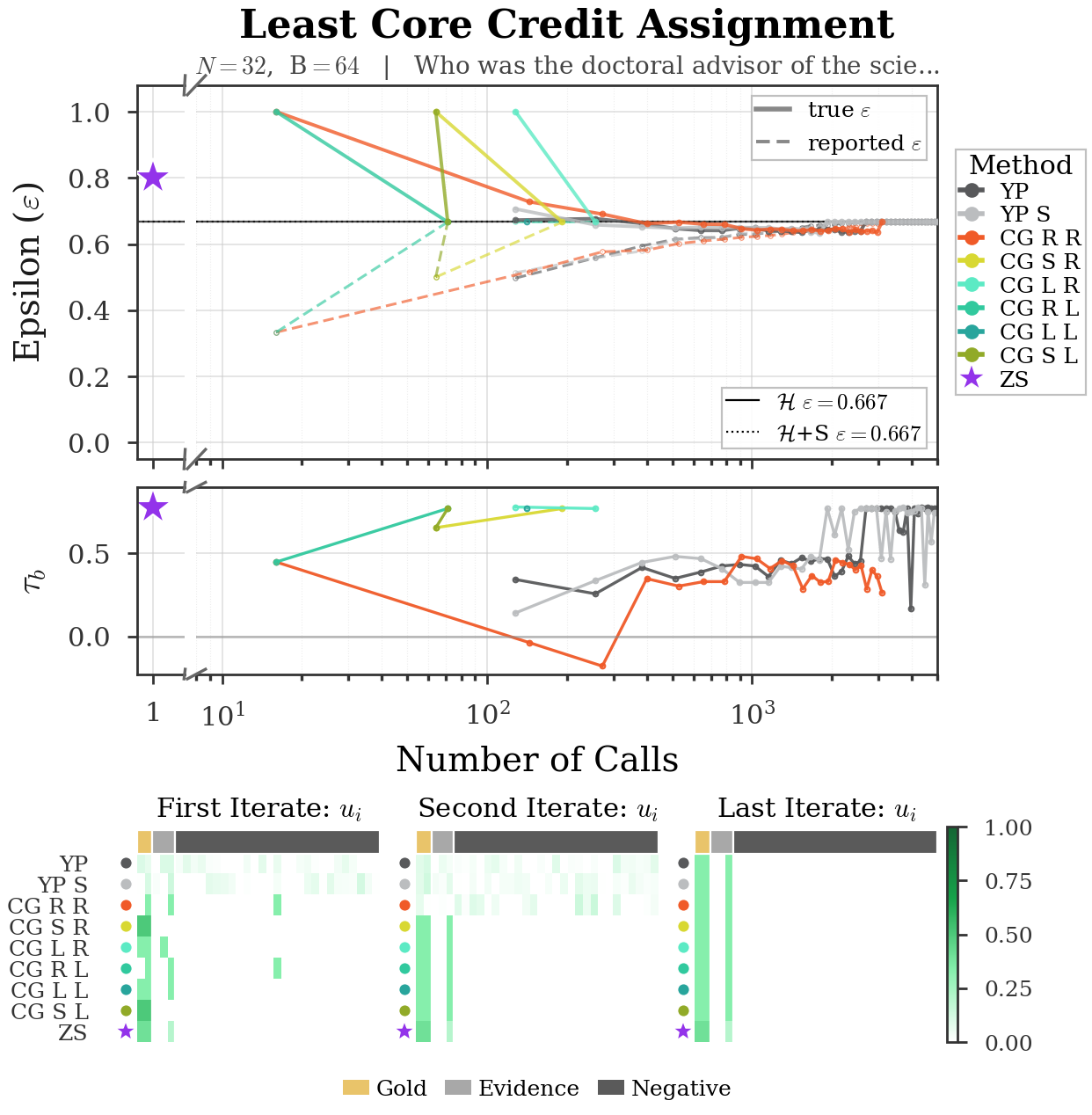}
        \caption{Row~4 (non-egal.)}
    \end{subfigure}\hfill
    \begin{subfigure}[t]{0.49\textwidth}
        \includegraphics[width=\textwidth]{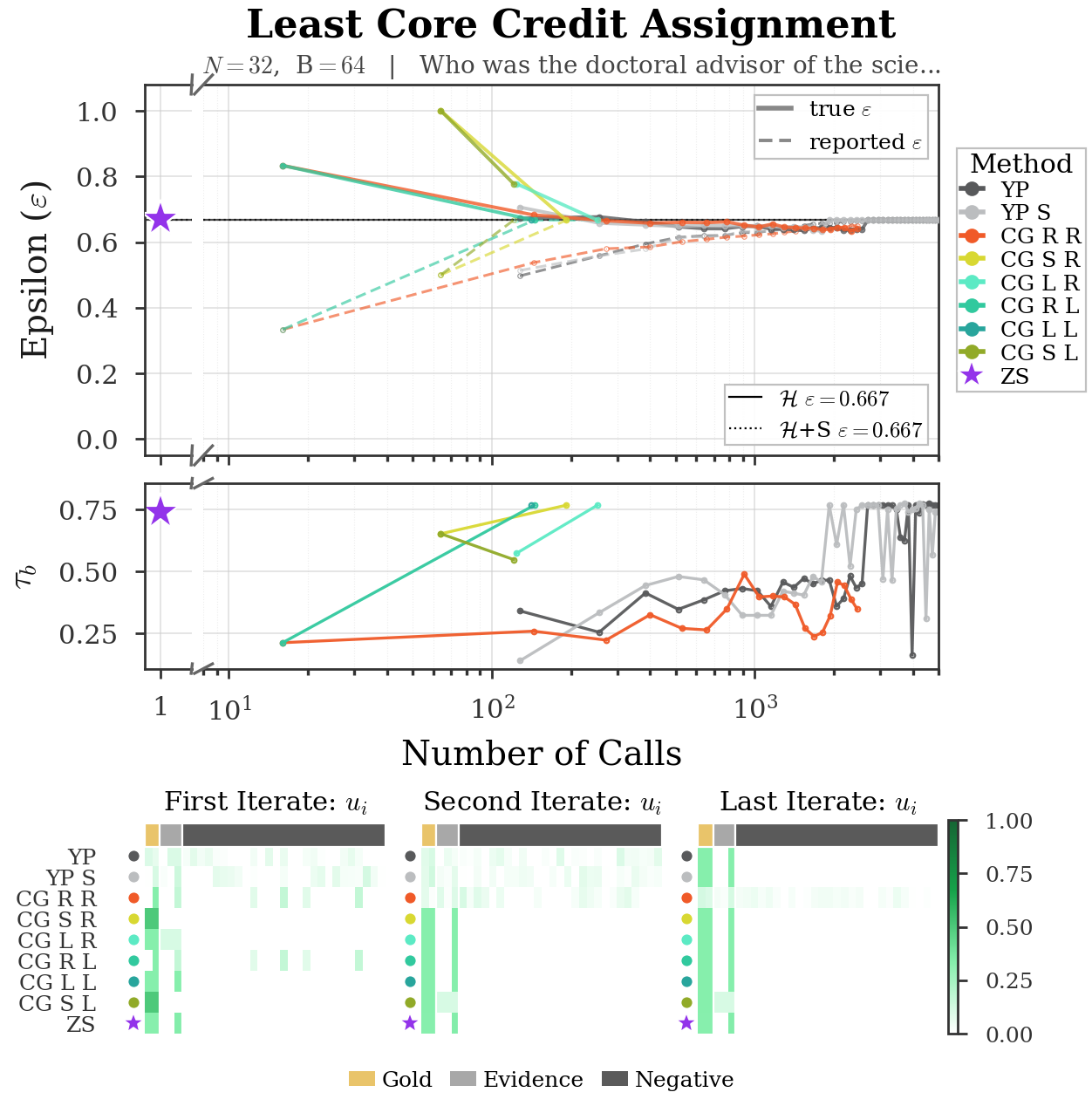}
        \caption{Row~4 (egal.)}
    \end{subfigure}
    \caption{Egalitarian effects on Synthetic $n{=}32$: non-egalitarian (left) vs.\ egalitarian (right). The QP objective redistributes credit more broadly, as visible in the heatmaps, with only minor impact on $\hat{\varepsilon}_{\mathcal{H}}$ relative to \textbf{YP}'s baseline.}
    \label{fig:egal-effects-1}
\end{figure*}

\newpage

\begin{figure*}[!htb]
    \centering
    \begin{subfigure}[t]{0.49\textwidth}
        \includegraphics[width=\textwidth]{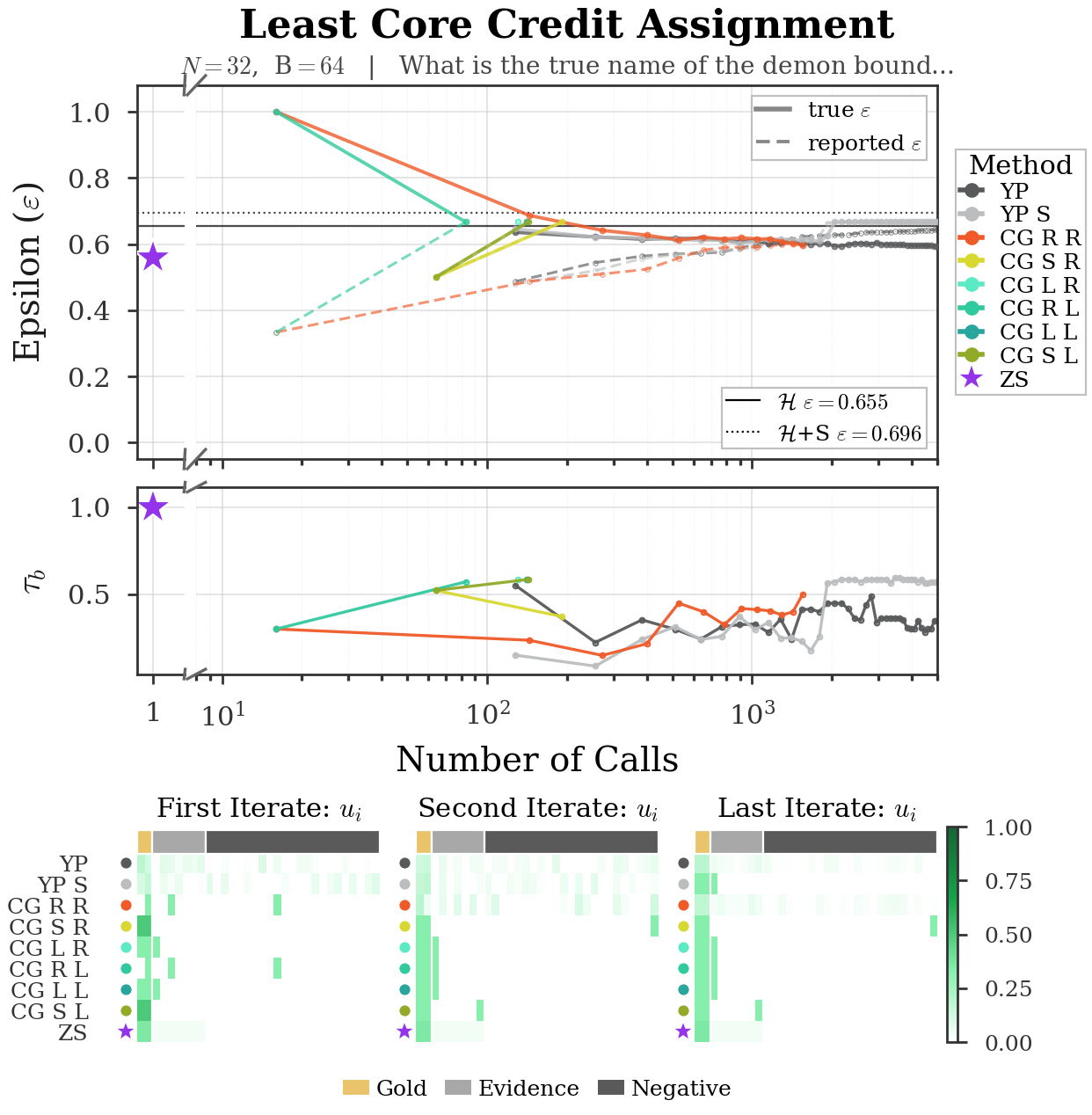}
        \caption{Row~7 (non-egal.)}
    \end{subfigure}\hfill
    \begin{subfigure}[t]{0.49\textwidth}
        \includegraphics[width=\textwidth]{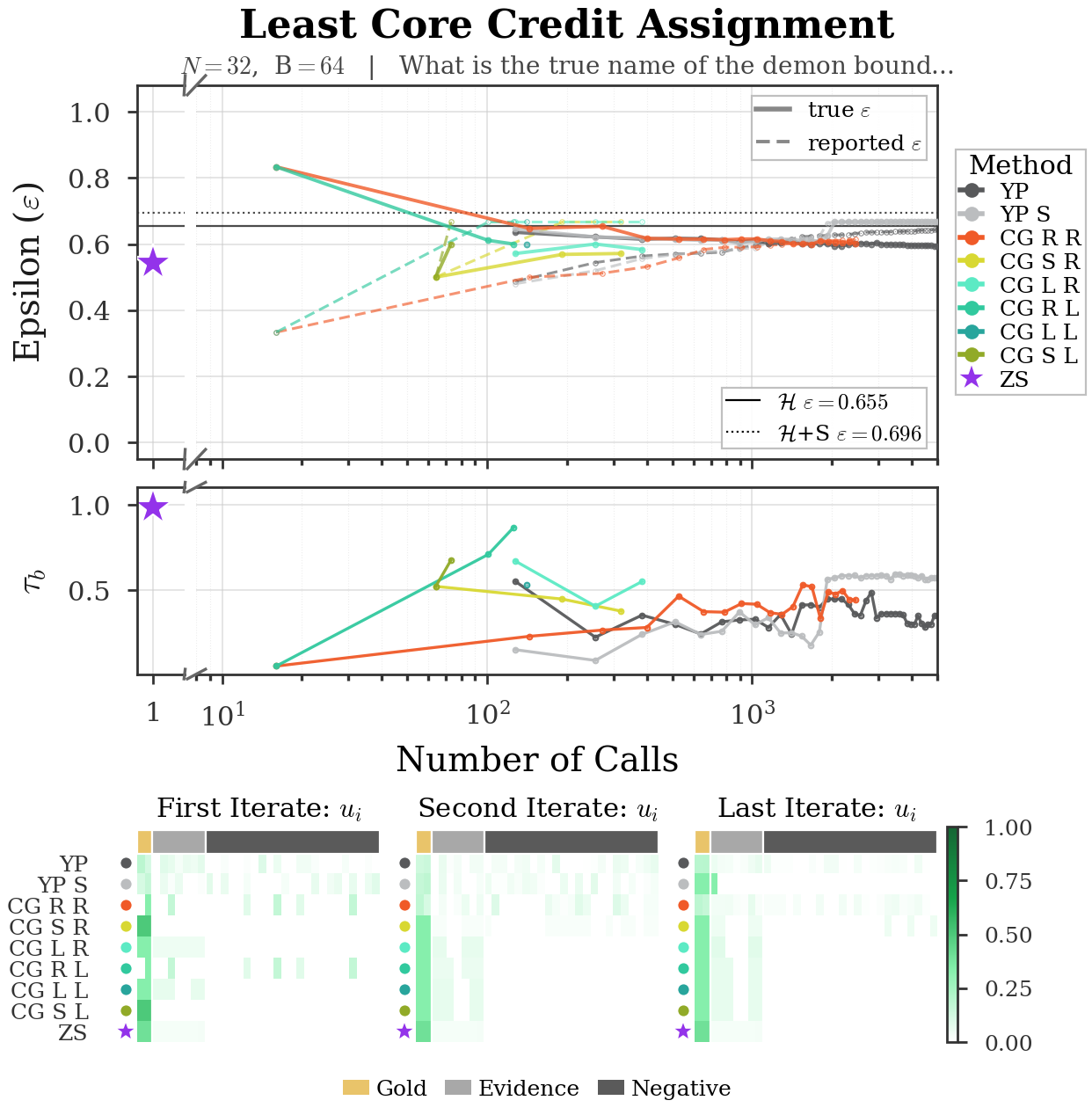}
        \caption{Row~7 (egal.)}
    \end{subfigure}\\[6pt]
    \begin{subfigure}[t]{0.49\textwidth}
        \includegraphics[width=\textwidth]{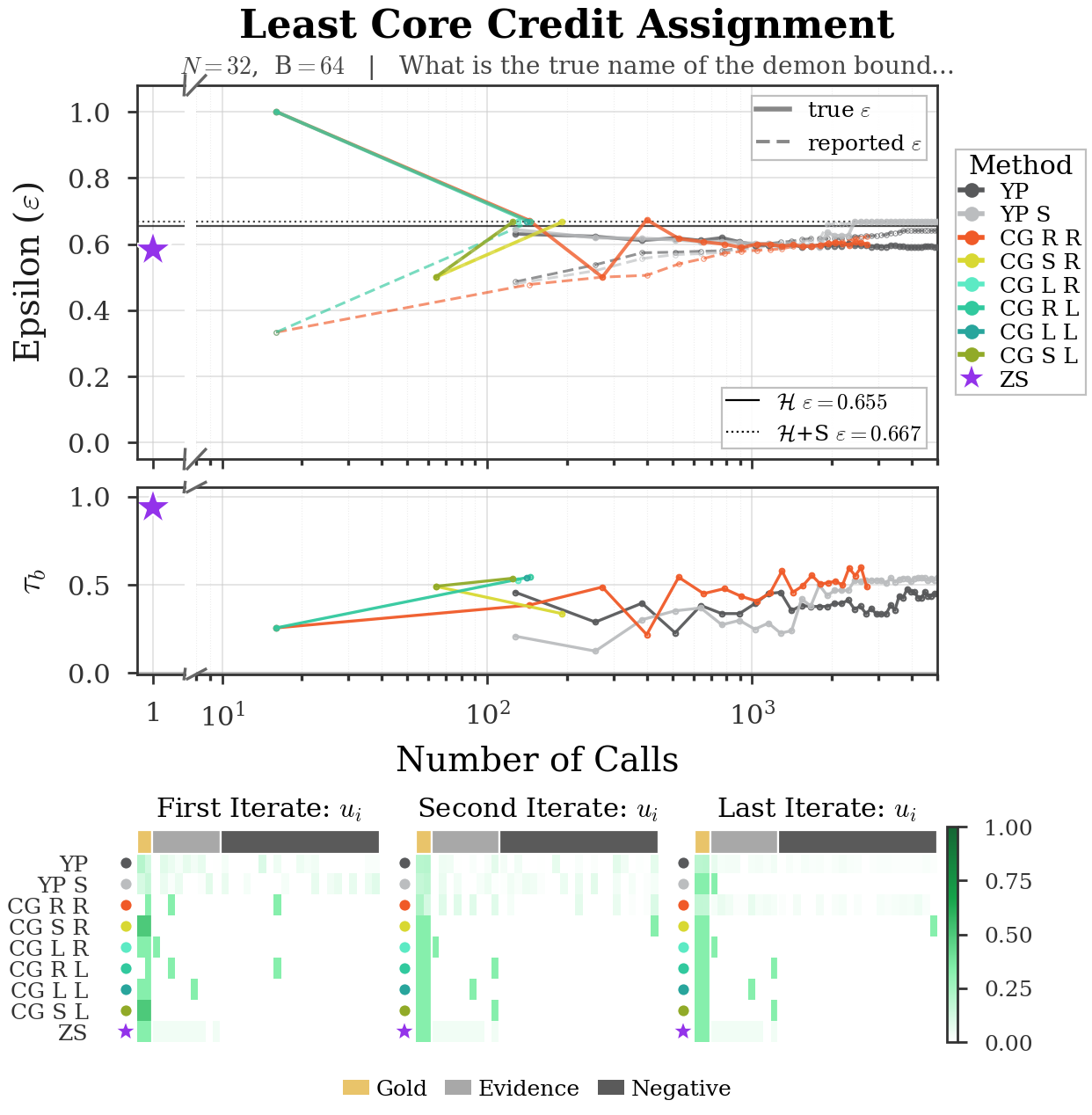}
        \caption{Row~13 (non-egal.)}
    \end{subfigure}\hfill
    \begin{subfigure}[t]{0.49\textwidth}
        \includegraphics[width=\textwidth]{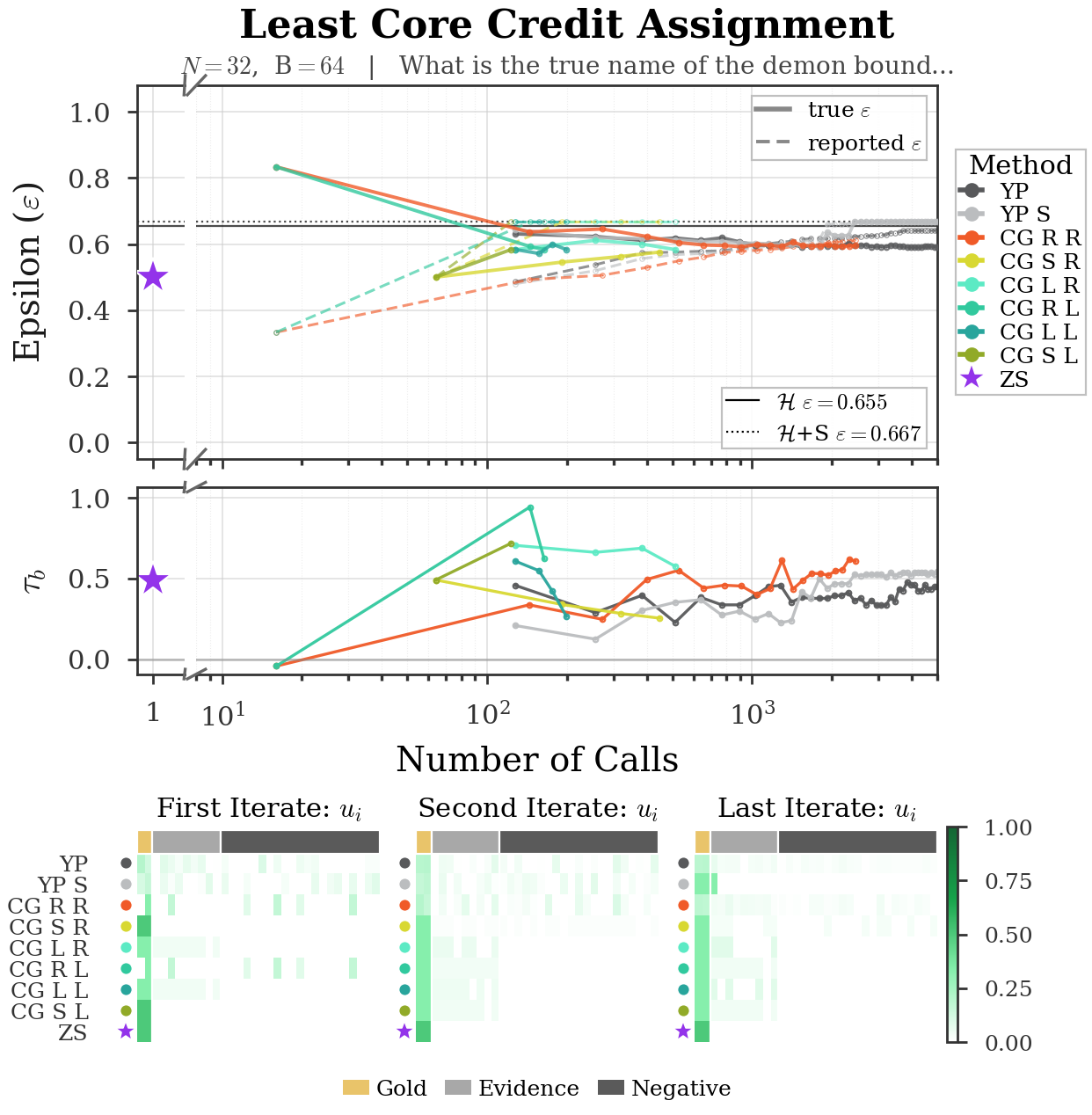}
        \caption{Row~13 (egal.)}
    \end{subfigure}
    \caption{Egalitarian effects on Synthetic $n{=}32$ (continued): non-egalitarian (left) vs.\ egalitarian (right). Credit heatmaps are visibly less sparse under the egalitarian objective.}
    \label{fig:egal-effects-2}
\end{figure*}

\newpage
\subsection{Additional BrowseComp $n{=}32$ Traces}
\label{app:bc-n32-traces}

Additional BrowseComp $n{=}32$ convergence traces are shown in \cref{fig:bc-n32-extra}, illustrating diverse \textbf{CG} dynamics across problem instances. In Row~1, \textbf{CG(R,R)} reaches within +0.01 of \textbf{YP}'s $\hat{\varepsilon}_{\mathcal{H}} = 0.65$ in ${\sim}900$ calls (a 9$\times$ reduction), while \textbf{ZS} overshoots by +0.12. Row~9 is similar: \textbf{CG(R,R)} matches \textbf{YP} within +0.01 at ${\sim}1{,}000$ calls. Row~22 shows \textbf{CG(R,R)} matching \textbf{YP}'s $\hat{\varepsilon}_{\mathcal{H}} = 0.66$ within +0.004 in ${\sim}1{,}000$ calls (an 8$\times$ reduction). In Row~31, \textbf{CG(L,ZS)} and \textbf{CG(L,S)} \emph{undershoot} \textbf{YP}'s $\hat{\varepsilon}_{\mathcal{H}} = 0.58$, achieving $\hat{\varepsilon}_{\mathcal{H}} = 0.50$ ($-0.08$) in only 49--64 calls.

\begin{figure*}[!htb]
    \centering
    \begin{subfigure}[t]{0.49\textwidth}
        \includegraphics[width=\textwidth]{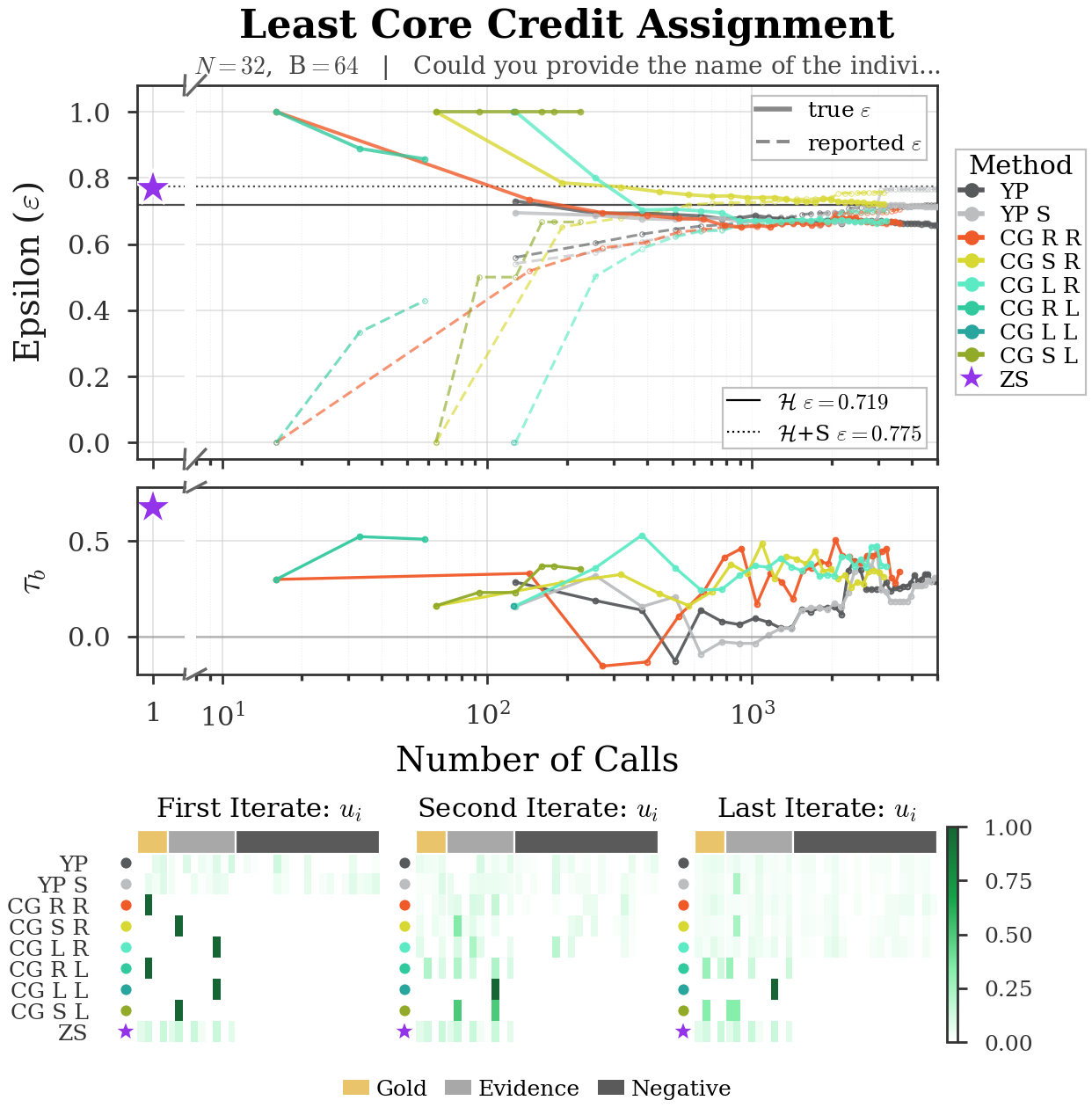}
        \caption{BrowseComp $n{=}32$, Row~1 (non-egal.)}
    \end{subfigure}\hfill
    \begin{subfigure}[t]{0.49\textwidth}
        \includegraphics[width=\textwidth]{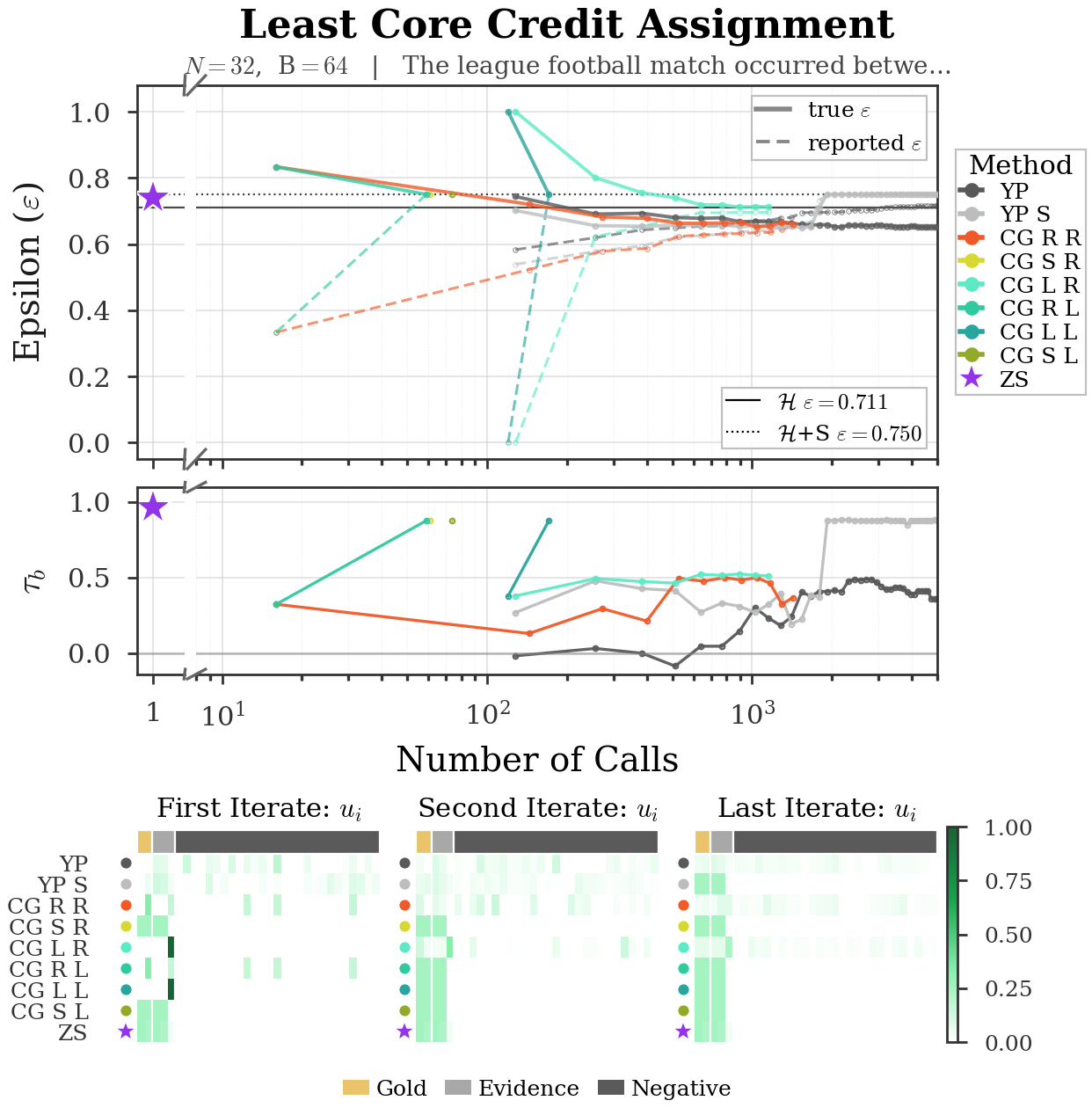}
        \caption{BrowseComp $n{=}32$, Row~9 (non-egal.)}
    \end{subfigure}\\[6pt]
    \begin{subfigure}[t]{0.49\textwidth}
        \includegraphics[width=\textwidth]{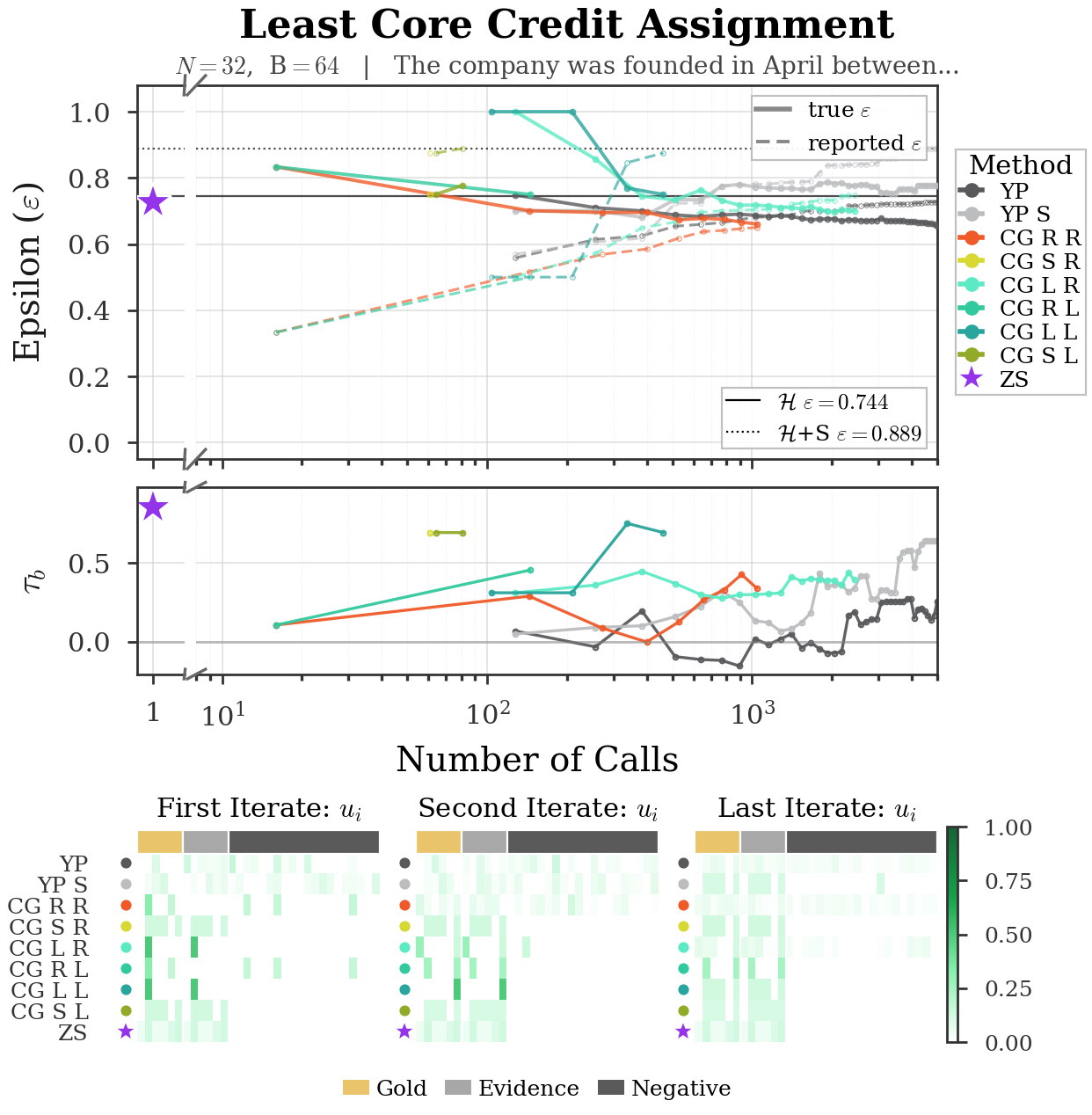}
        \caption{BrowseComp $n{=}32$, Row~22 (non-egal.)}
    \end{subfigure}\hfill
    \begin{subfigure}[t]{0.49\textwidth}
        \includegraphics[width=\textwidth]{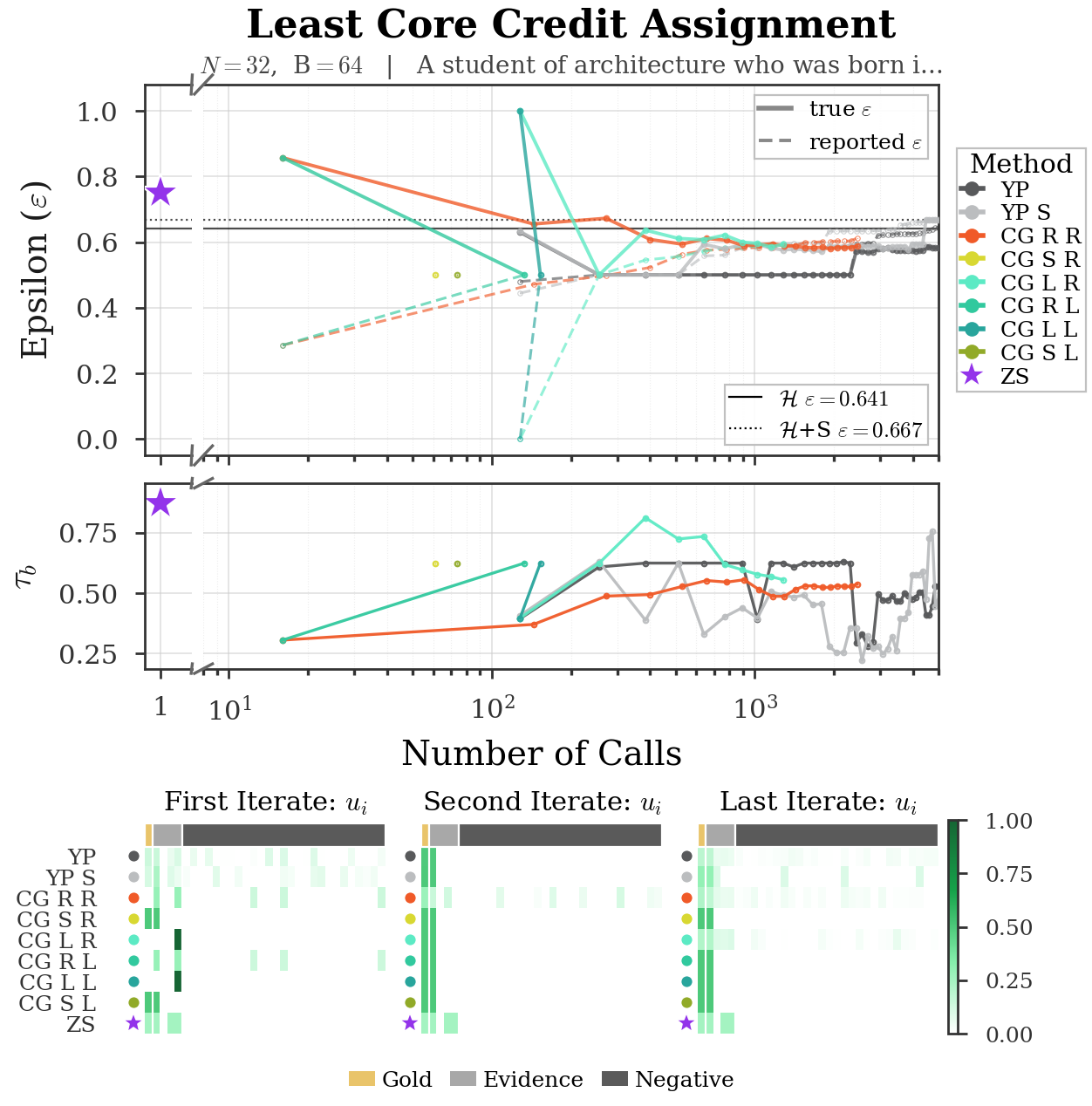}
        \caption{BrowseComp $n{=}32$, Row~31 (non-egal.)}
    \end{subfigure}
    \caption{Additional BrowseComp $n{=}32$ convergence traces. \textbf{CG(R,R)} consistently matches \textbf{YP}'s holdout $\hat{\varepsilon}_{\mathcal{H}}$ with 8--9$\times$ fewer calls; in Row~31 the LLM-oracle variants undershoot \textbf{YP} by 0.08.}
    \label{fig:bc-n32-extra}
\end{figure*}

\newpage
\subsection{$n{=}10$ Experiments}
\label{app:n10}

We repeat our experiments with $n{=}10$ documents per problem instance. At this scale, the full coalition space $|2^n| = 1024$ is tractable, allowing a more precise holdout evaluation. The \textbf{YP} baseline and $\hat{\varepsilon}_{\mathcal{H}}$ are effectively exact. Forest plots and selected convergence traces for Synthetic and BrowseComp are provided in \cref{fig:forest-synth-n10,fig:synth-n10-nonegal,fig:synth-n10-egal} and \cref{fig:forest-bc-n10,fig:bc-n10-nonegal,fig:bc-n10-egal}, respectively.

\subsubsection{Synthetic $n{=}10$}

\begin{figure*}[!htb]
    \centering
    \includegraphics[width=\textwidth]{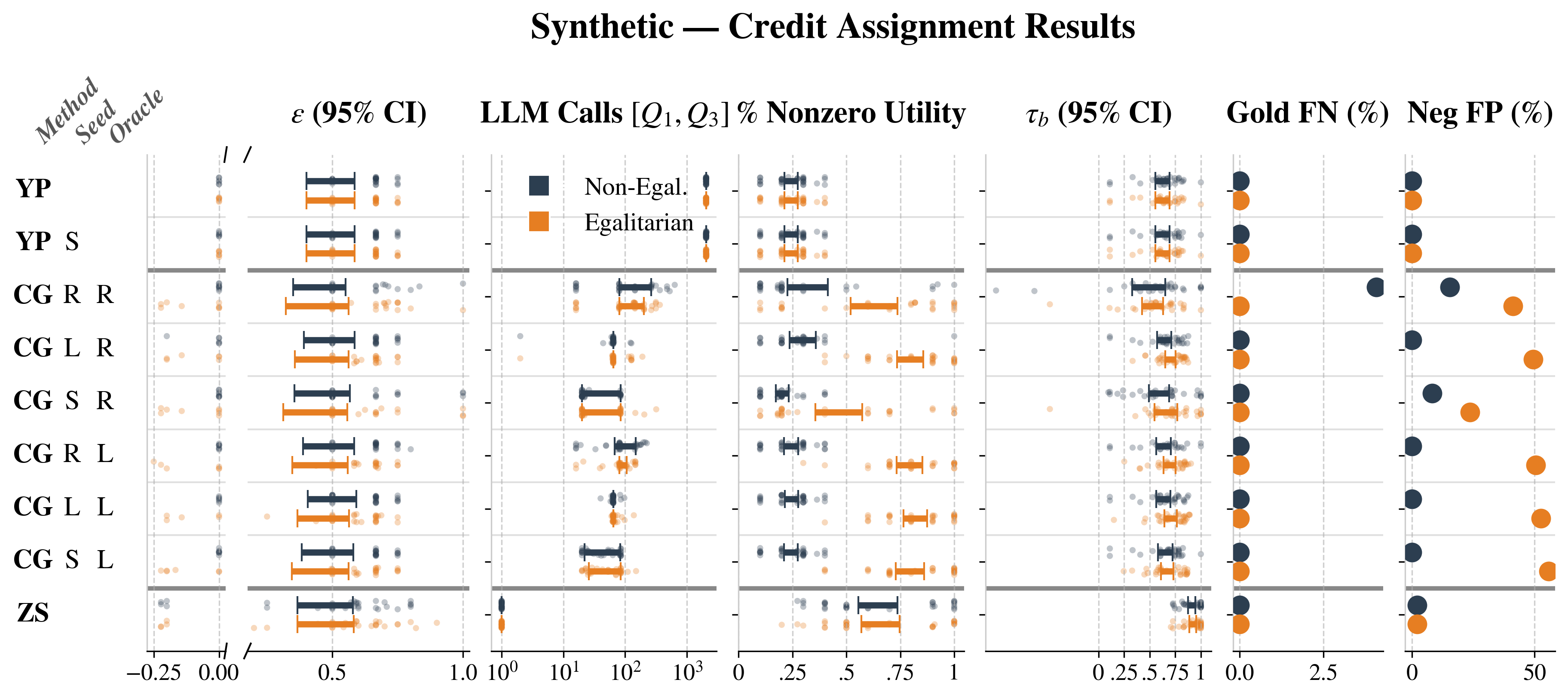}
    \caption{Credit assignment results on Synthetic ($n{=}10$) for non-egalitarian (dark) and egalitarian (light) settings. Lower $\hat{\varepsilon}_{\mathcal{H}}$ is better.}
    \label{fig:forest-synth-n10}
\end{figure*}

\newpage
\paragraph{Selected Non-Egalitarian Traces.}
At $n{=}10$, \textbf{CG} methods converge rapidly and often match \textbf{YP}'s holdout $\hat{\varepsilon}_{\mathcal{H}}$ exactly. In Row~10, all \textbf{CG} variants reach \textbf{YP}'s $\hat{\varepsilon}_{\mathcal{H}} = 0.50$ in 53--146 calls. In Row~13, \textbf{CG(R,R)} and \textbf{CG(R,S)} undershoot \textbf{YP}'s $\hat{\varepsilon}_{\mathcal{H}} = 0.67$ by 0.17, reaching $\hat{\varepsilon}_{\mathcal{H}} = 0.50$, while \textbf{ZS} undershoots modestly ($\hat{\varepsilon}_{\mathcal{H}} = 0.58$, $-0.08$). Row~6 and Row~15 show \textbf{CG(R,R)} matching or slightly beating \textbf{YP}'s baseline.

\begin{figure*}[!htb]
    \centering
    \begin{subfigure}[t]{0.49\textwidth}
        \includegraphics[width=\textwidth]{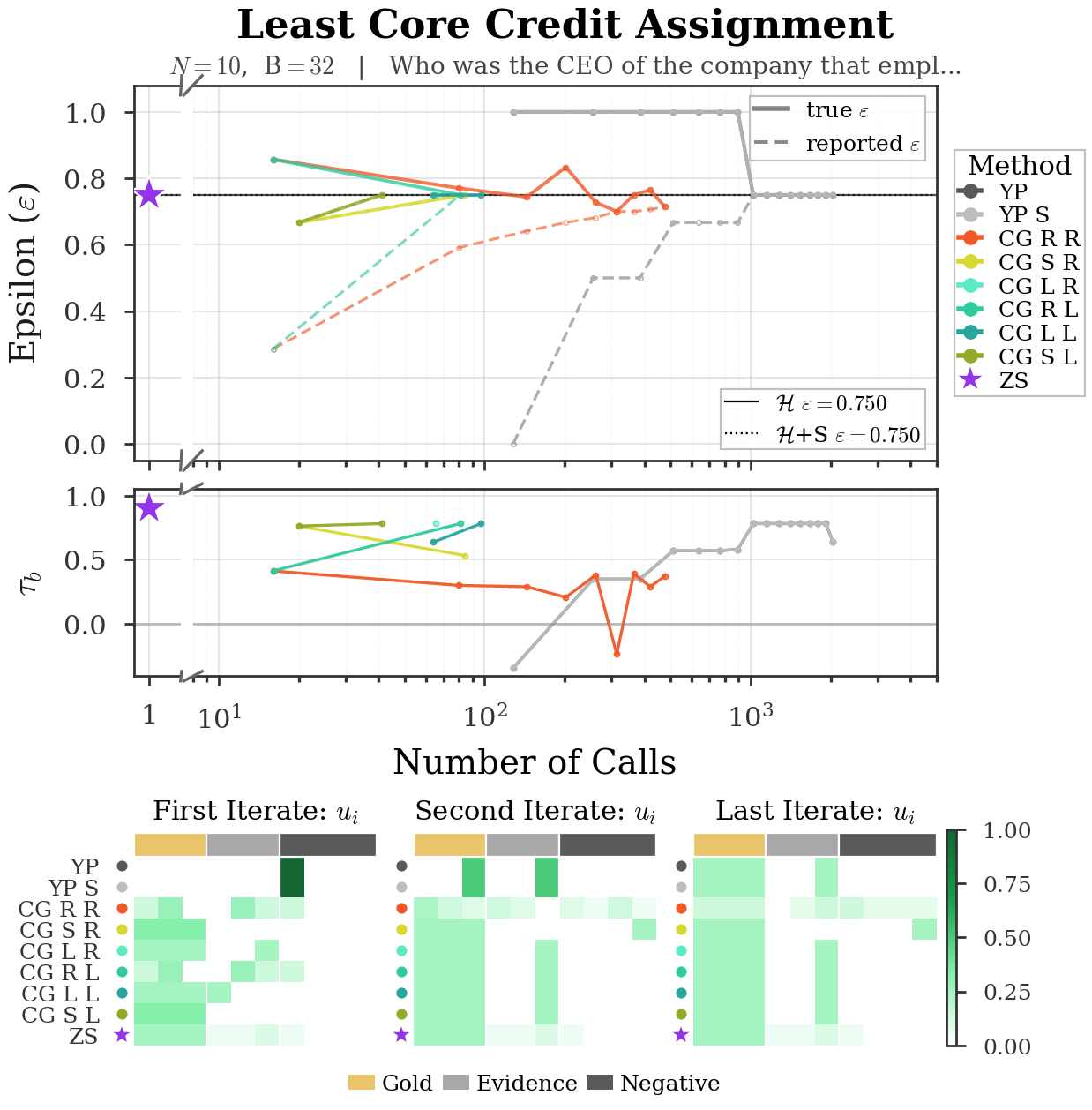}
        \caption{Row~6 (non-egal.)}
    \end{subfigure}\hfill
    \begin{subfigure}[t]{0.49\textwidth}
        \includegraphics[width=\textwidth]{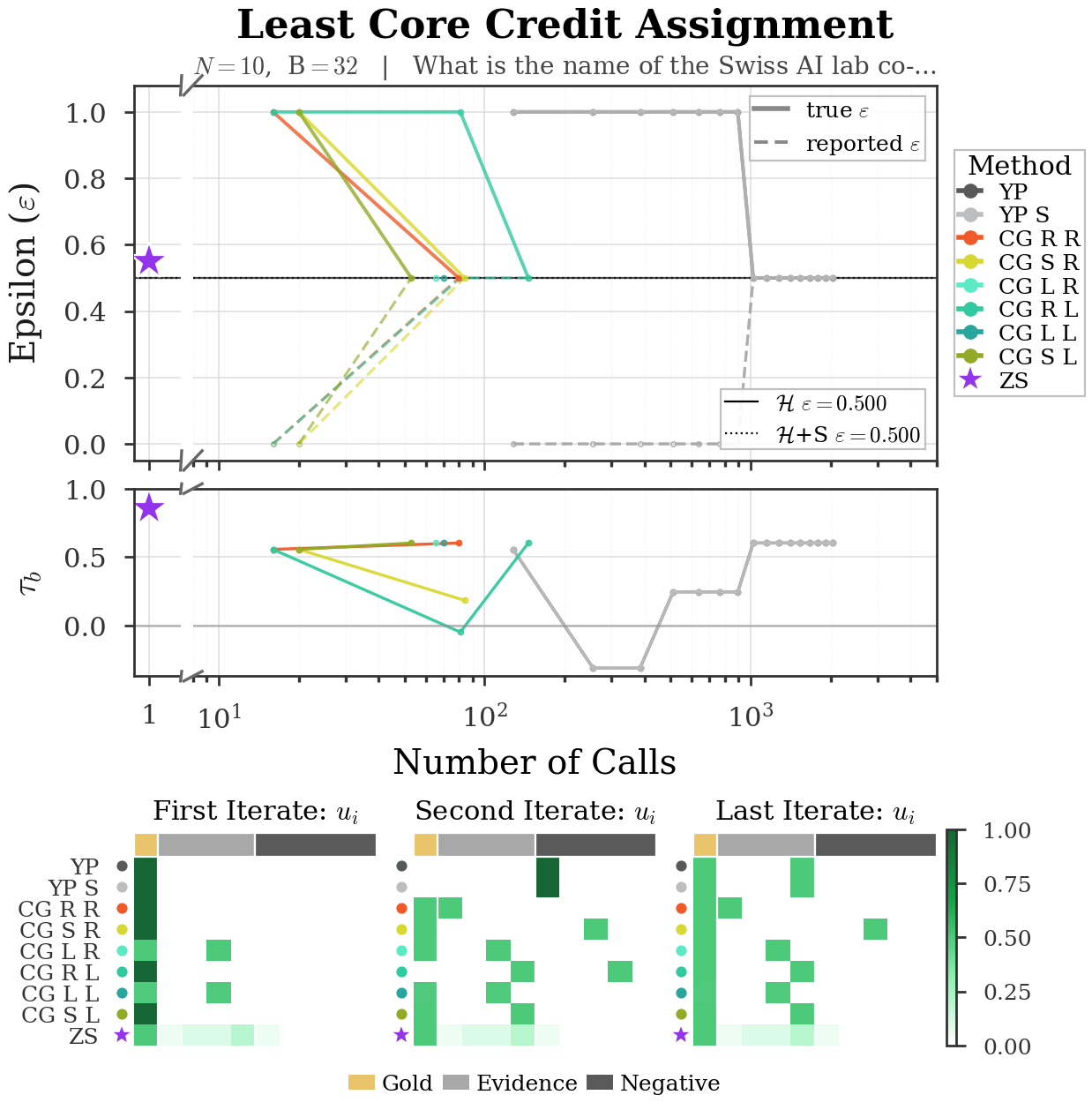}
        \caption{Row~10 (non-egal.)}
    \end{subfigure}\\[6pt]
    \begin{subfigure}[t]{0.49\textwidth}
        \includegraphics[width=\textwidth]{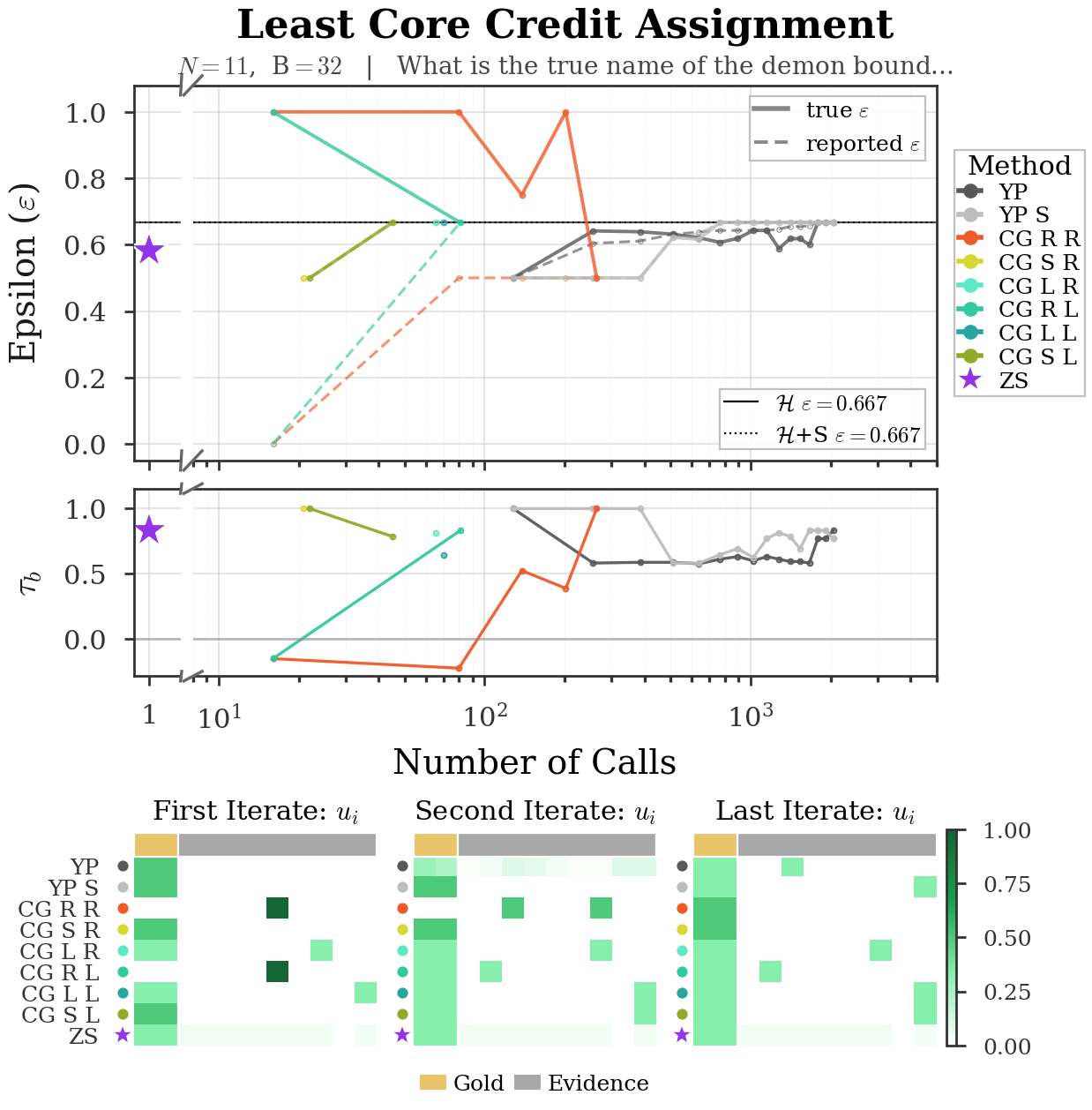}
        \caption{Row~13 (non-egal.)}
    \end{subfigure}\hfill
    \begin{subfigure}[t]{0.49\textwidth}
        \includegraphics[width=\textwidth]{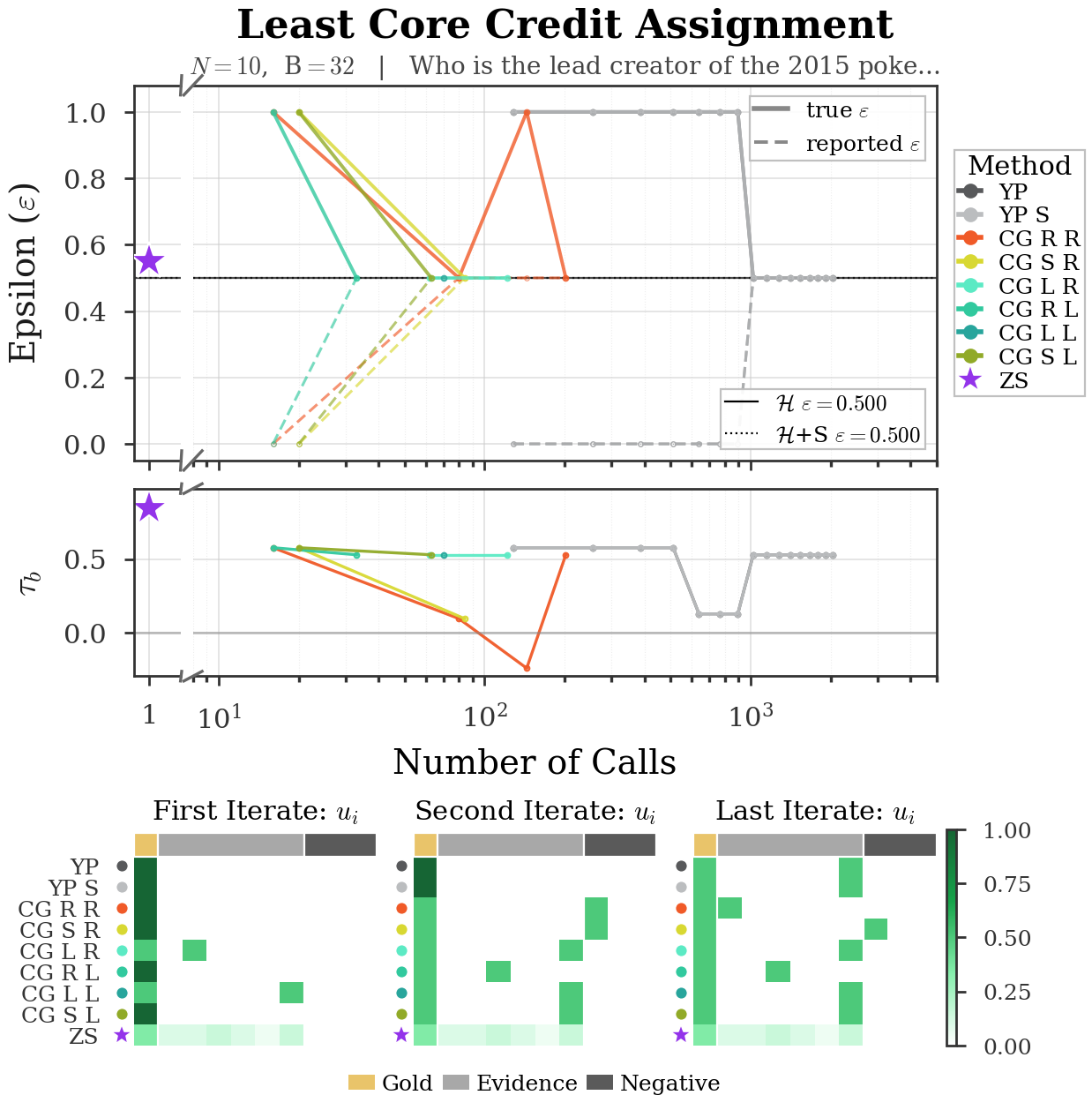}
        \caption{Row~20 (non-egal.)}
    \end{subfigure}
    \caption{Synthetic $n{=}10$, non-egalitarian convergence traces. \textbf{CG} methods match or undershoot \textbf{YP}'s holdout $\hat{\varepsilon}_{\mathcal{H}}$ within ${\sim}50$--300 calls.}
    \label{fig:synth-n10-nonegal}
\end{figure*}

\begin{figure*}[!htb]\ContinuedFloat
    \centering
    \begin{subfigure}[t]{0.49\textwidth}
        \includegraphics[width=\textwidth]{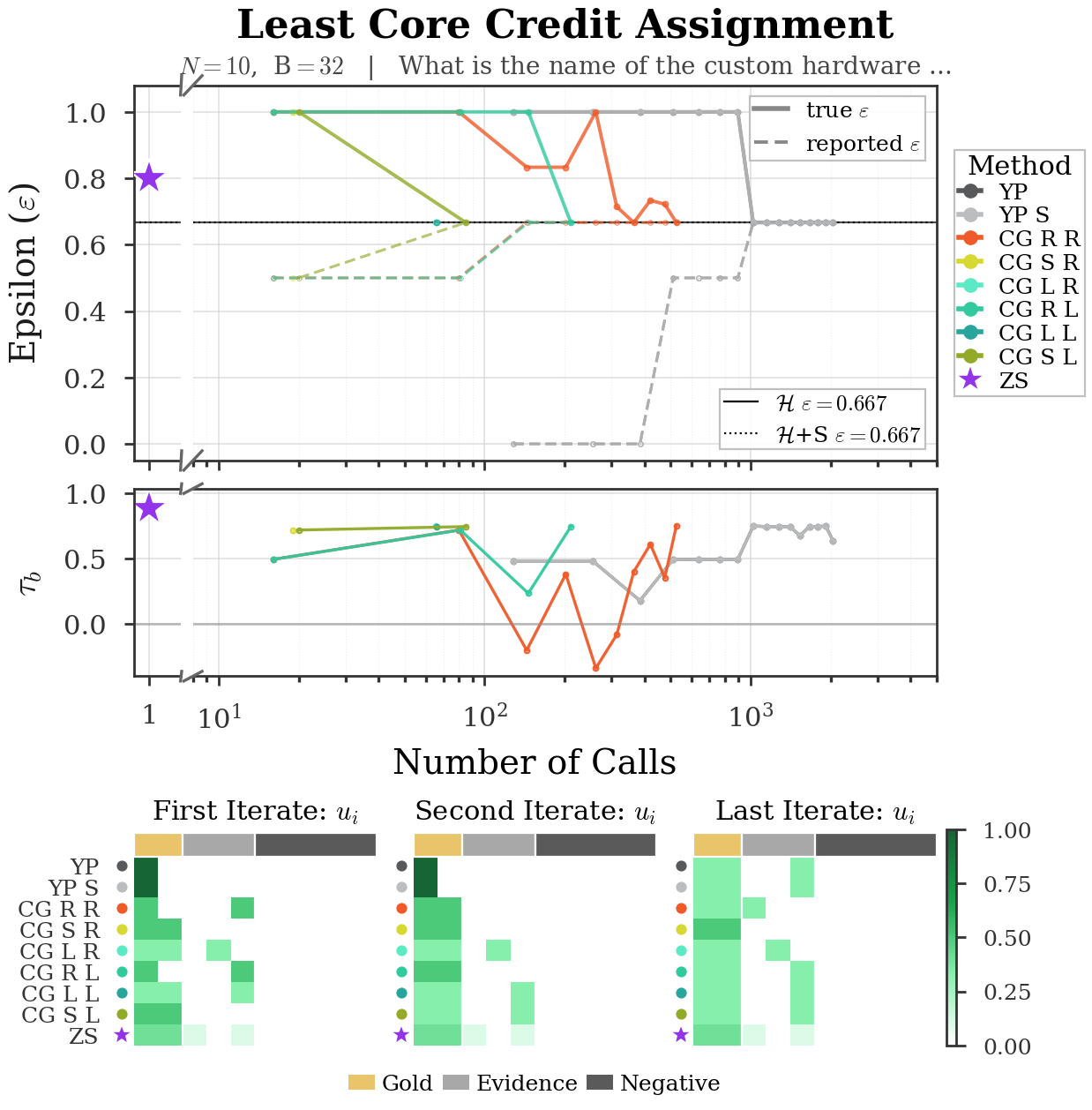}
        \caption{Row~21 (non-egal.)}
    \end{subfigure}\hfill
    \begin{subfigure}[t]{0.49\textwidth}
        \includegraphics[width=\textwidth]{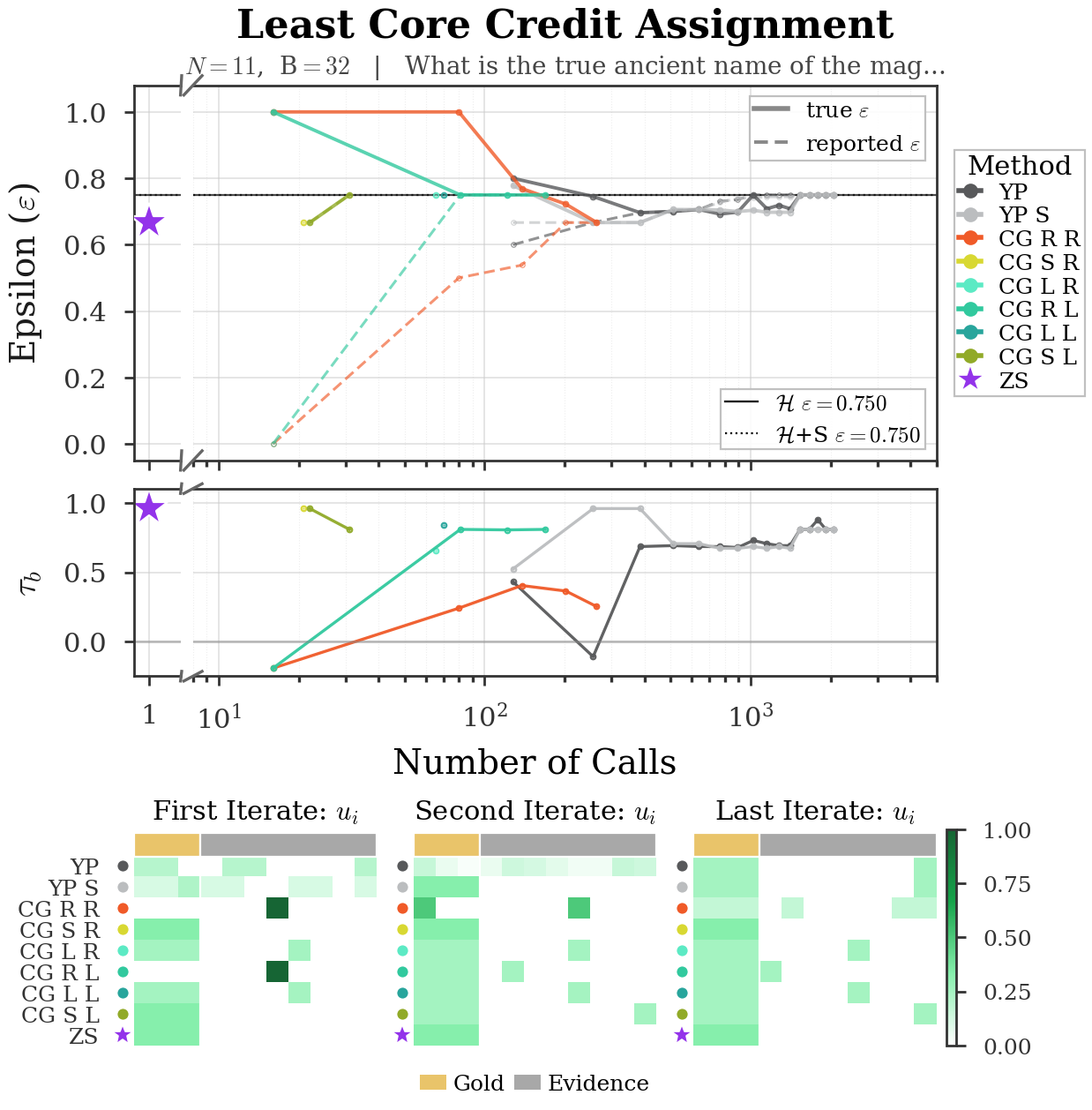}
        \caption{Row~15 (non-egal.)}
    \end{subfigure}
    \caption[]{Synthetic $n{=}10$, non-egalitarian convergence traces (cont.).}
\end{figure*}

\newpage
\paragraph{Selected Egalitarian Traces.}
The egalitarian objective redistributes credit without substantially degrading $\hat{\varepsilon}_{\mathcal{H}}$. In Row~10 (egal.), all \textbf{CG} variants still match \textbf{YP}'s $\hat{\varepsilon}_{\mathcal{H}} = 0.50$, with visibly less sparse heatmaps. Row~15 shows \textbf{CG(R,R)} matching \textbf{YP}'s baseline within 0.01, with broader credit distribution.

\begin{figure*}[!htb]
    \centering
    \begin{subfigure}[t]{0.49\textwidth}
        \includegraphics[width=\textwidth]{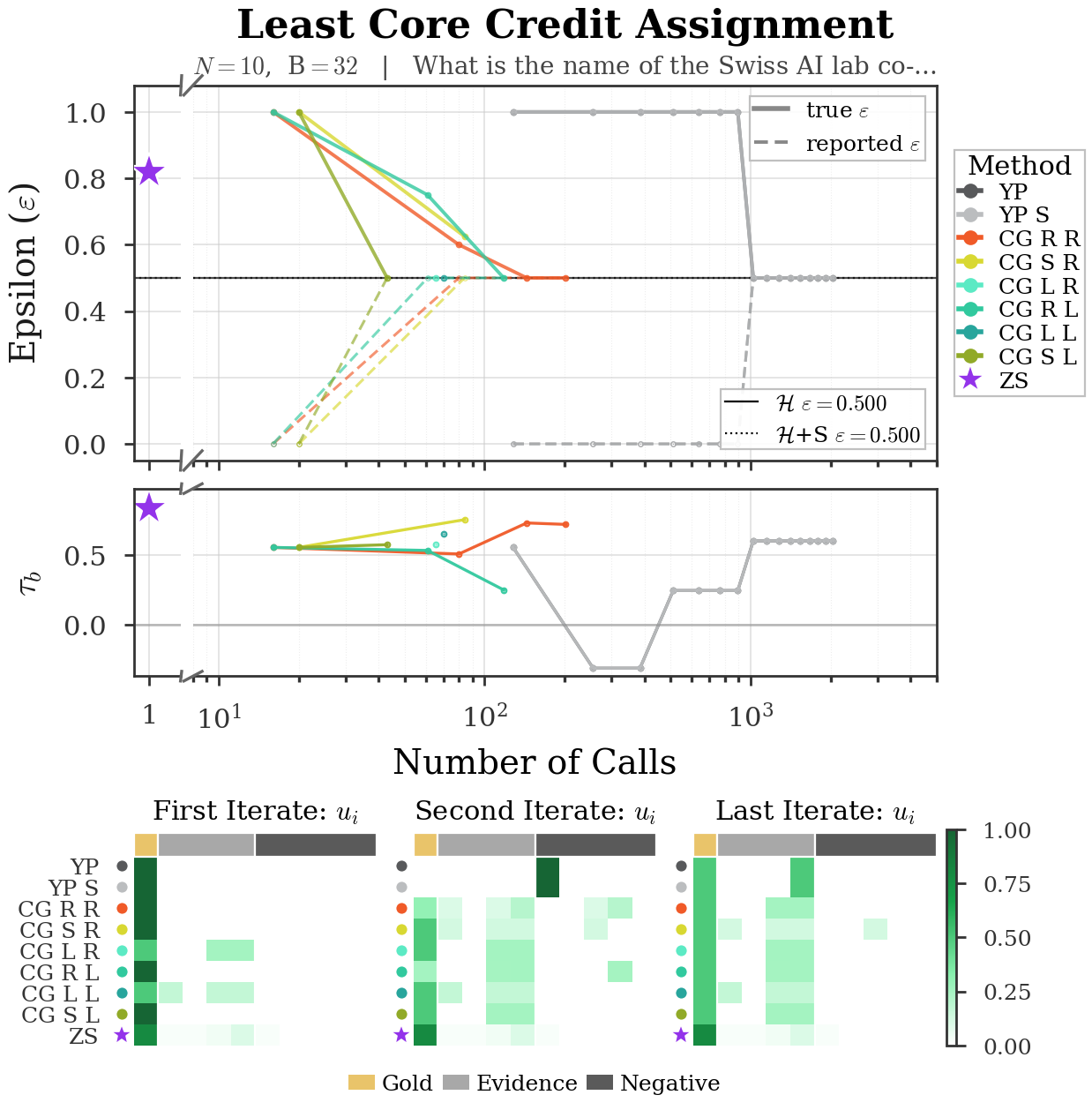}
        \caption{Row~10 (egal.)}
    \end{subfigure}\hfill
    \begin{subfigure}[t]{0.49\textwidth}
        \includegraphics[width=\textwidth]{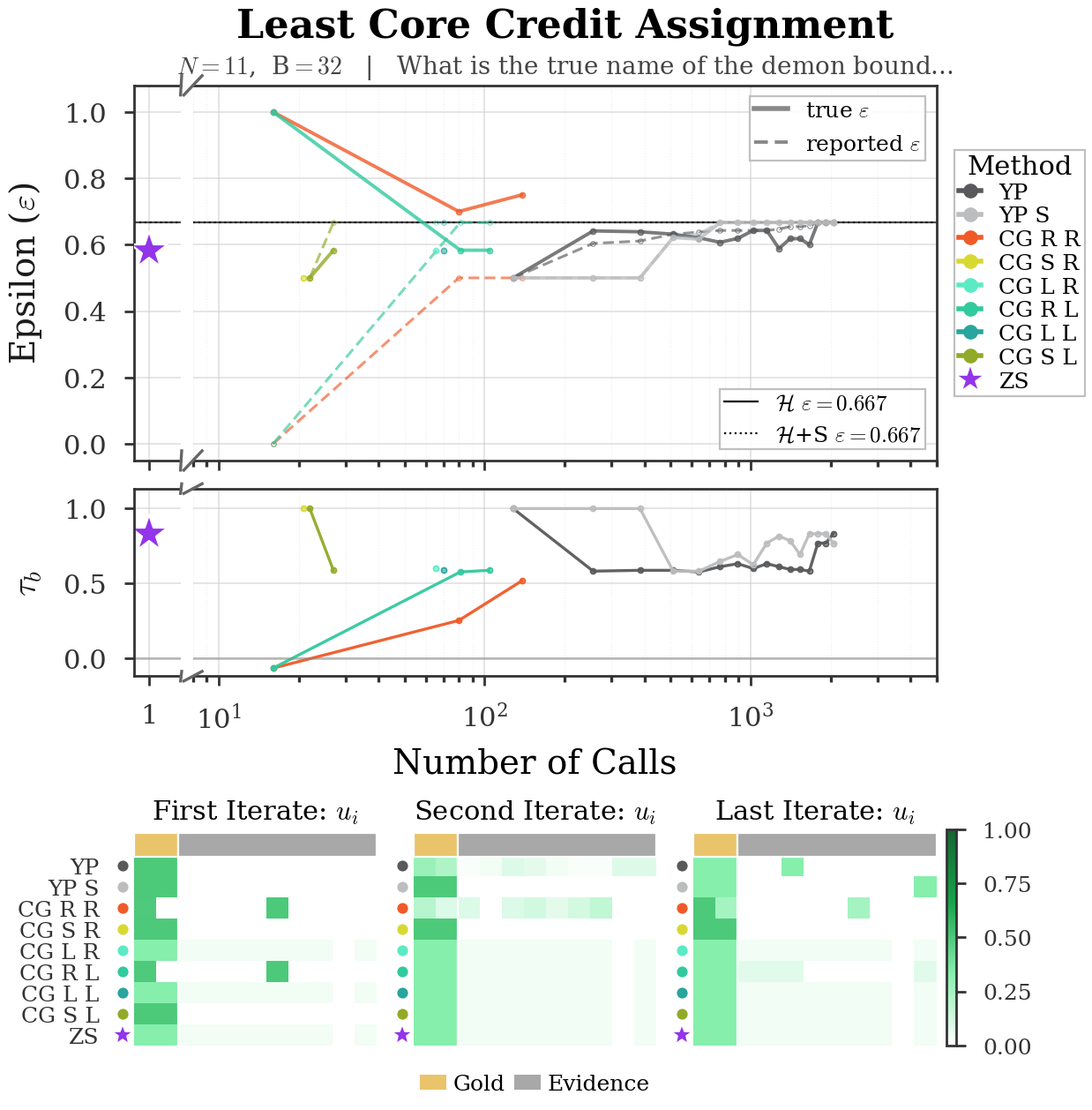}
        \caption{Row~13 (egal.)}
    \end{subfigure}\\[6pt]
    \begin{subfigure}[t]{0.49\textwidth}
        \includegraphics[width=\textwidth]{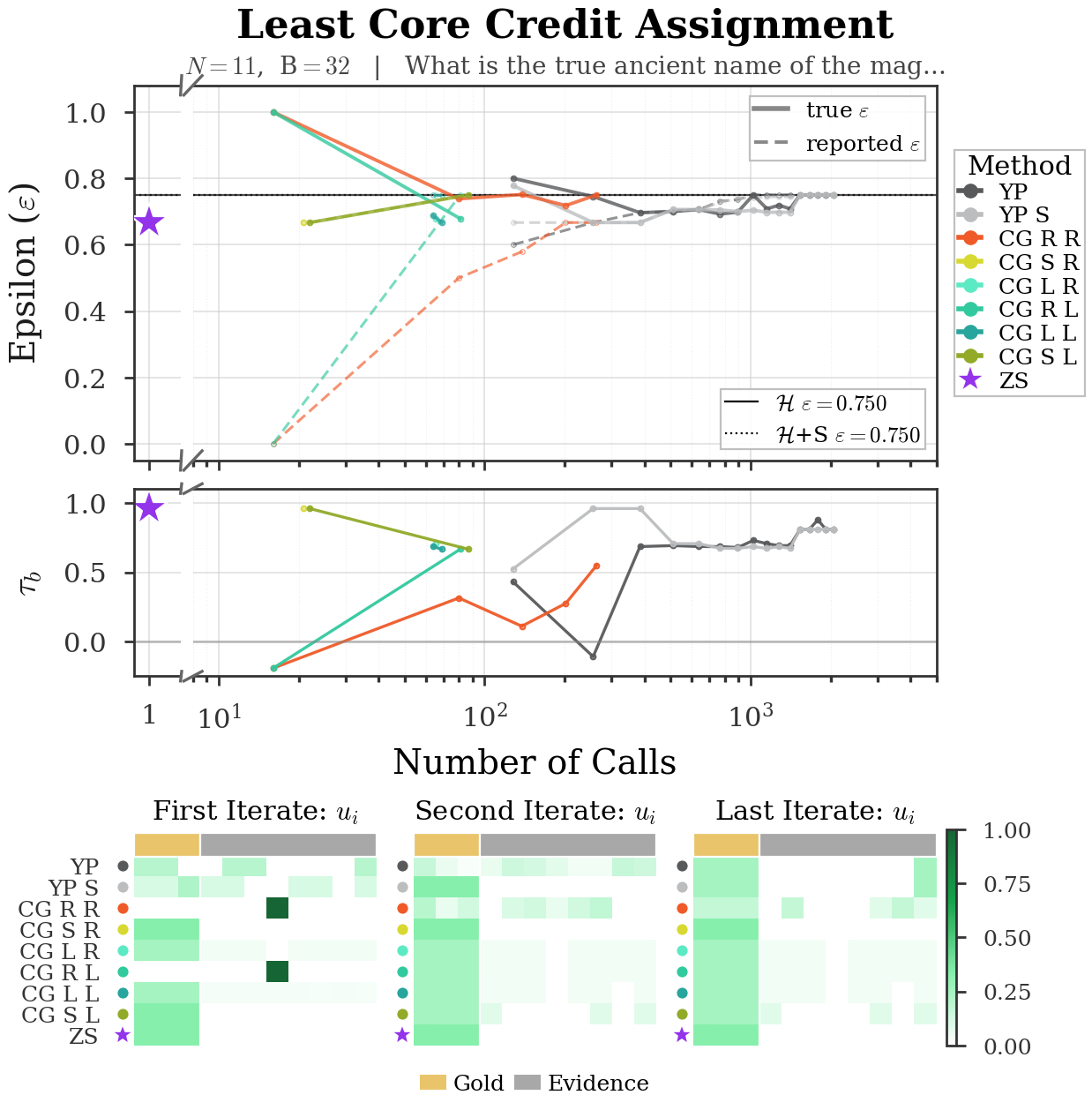}
        \caption{Row~15 (egal.)}
    \end{subfigure}\hfill
    \begin{subfigure}[t]{0.49\textwidth}
        \includegraphics[width=\textwidth]{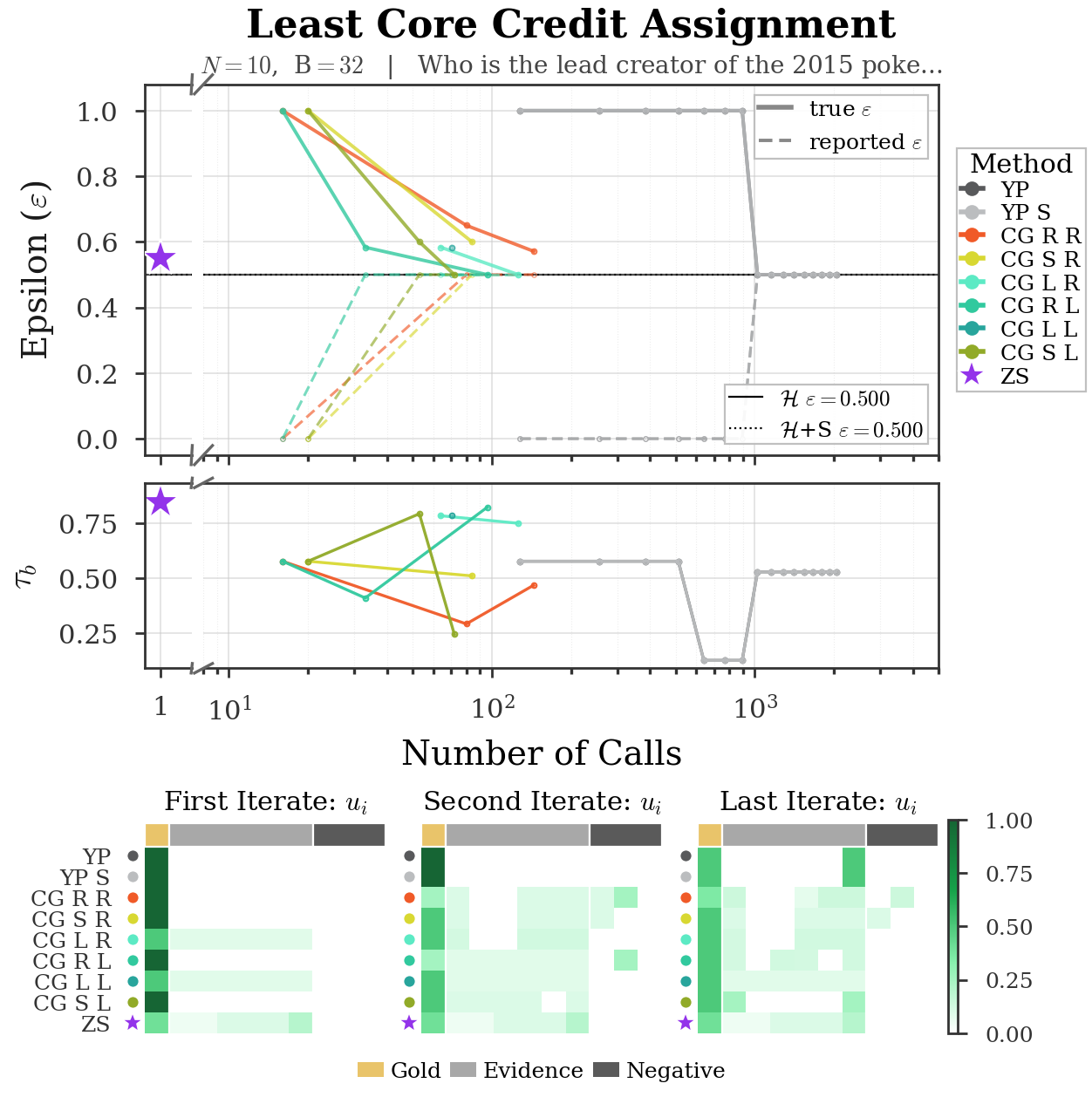}
        \caption{Row~20 (egal.)}
    \end{subfigure}
    \caption{Synthetic $n{=}10$, egalitarian convergence traces. Heatmaps show broader credit distribution with minimal $\hat{\varepsilon}_{\mathcal{H}}$ degradation relative to \textbf{YP}.}
    \label{fig:synth-n10-egal}
\end{figure*}

\newpage
\subsubsection{BrowseComp $n{=}10$}

\begin{figure*}[!htb]
    \centering
    \includegraphics[width=\textwidth]{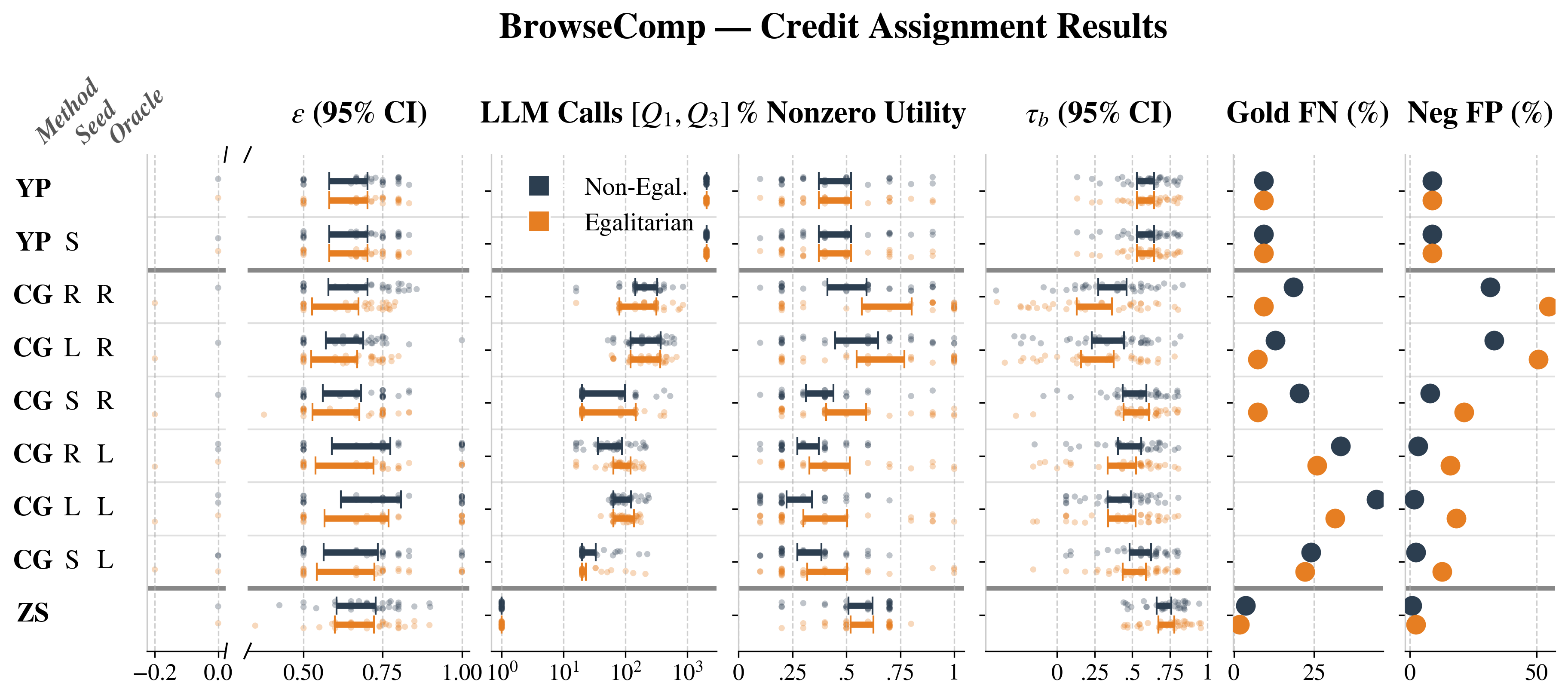}
    \caption{Credit assignment results on BrowseComp-Plus ($n{=}10$) for non-egalitarian (dark) and egalitarian (light) settings. Lower $\hat{\varepsilon}_{\mathcal{H}}$ is better.}
    \label{fig:forest-bc-n10}
\end{figure*}

\newpage
\paragraph{Selected Non-Egalitarian Traces.}
BrowseComp at $n{=}10$ shows noisier dynamics due to label-quality issues. In Row~18, \textbf{CG(R,R)} matches \textbf{YP}'s $\hat{\varepsilon}_{\mathcal{H}} = 0.67$ exactly in ${\sim}200$ calls (a 10$\times$ reduction). In Row~21, \textbf{CG(R,ZS)} undershoots \textbf{YP}'s $\hat{\varepsilon}_{\mathcal{H}} = 0.75$ by 0.05. Row~22 shows \textbf{CG(R,R)} undershooting \textbf{YP}'s $\hat{\varepsilon}_{\mathcal{H}} = 0.80$ by ${\sim}0.09$ in ${\sim}420$ calls. Row~1 demonstrates \textbf{CG(R,R)} matching \textbf{YP}'s $\hat{\varepsilon}_{\mathcal{H}} = 0.75$ exactly within ${\sim}80$ calls.

\begin{figure*}[!htb]
    \centering
    \begin{subfigure}[t]{0.49\textwidth}
        \includegraphics[width=\textwidth]{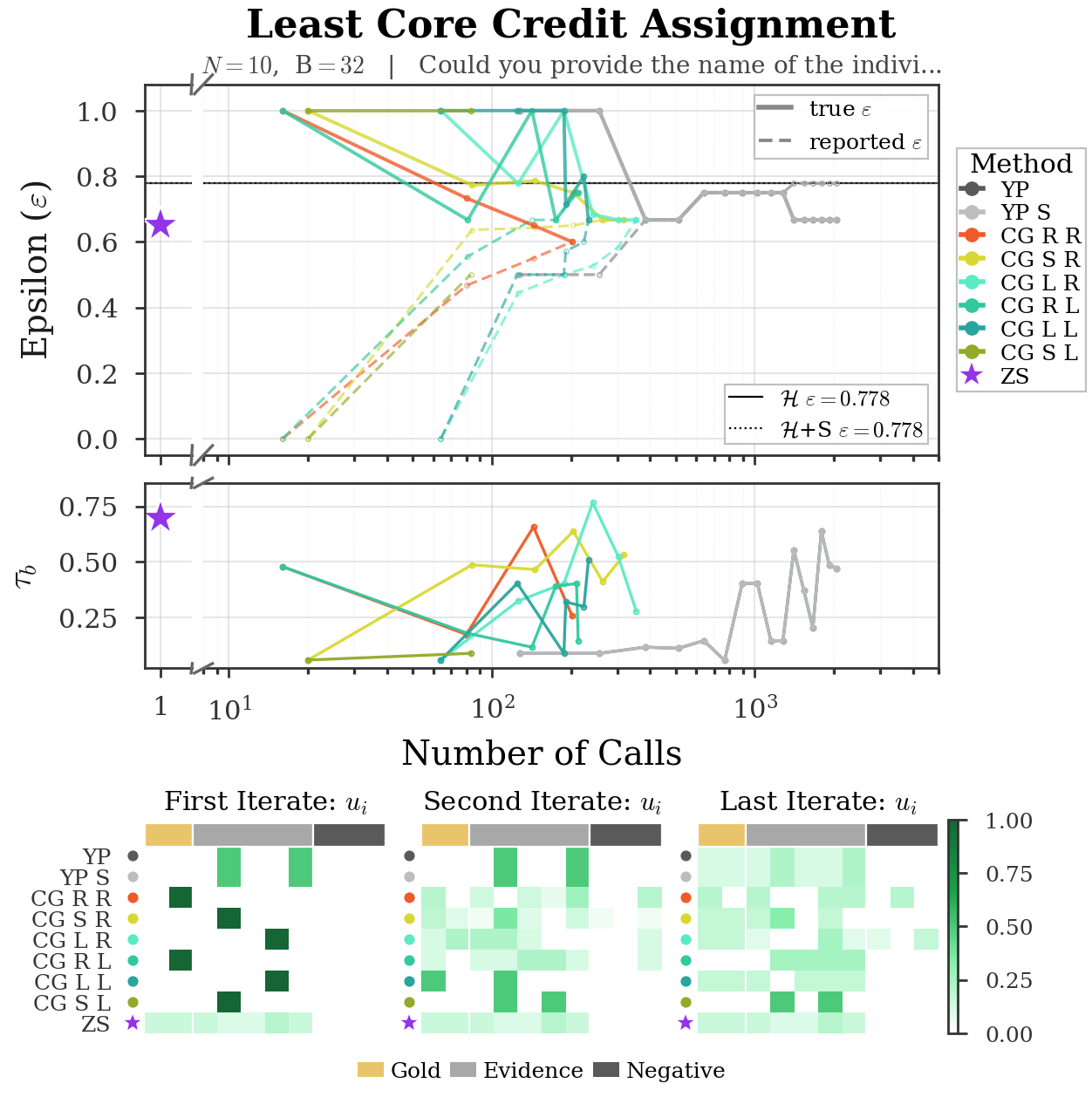}
        \caption{Row~1 (non-egal.)}
    \end{subfigure}\hfill
    \begin{subfigure}[t]{0.49\textwidth}
        \includegraphics[width=\textwidth]{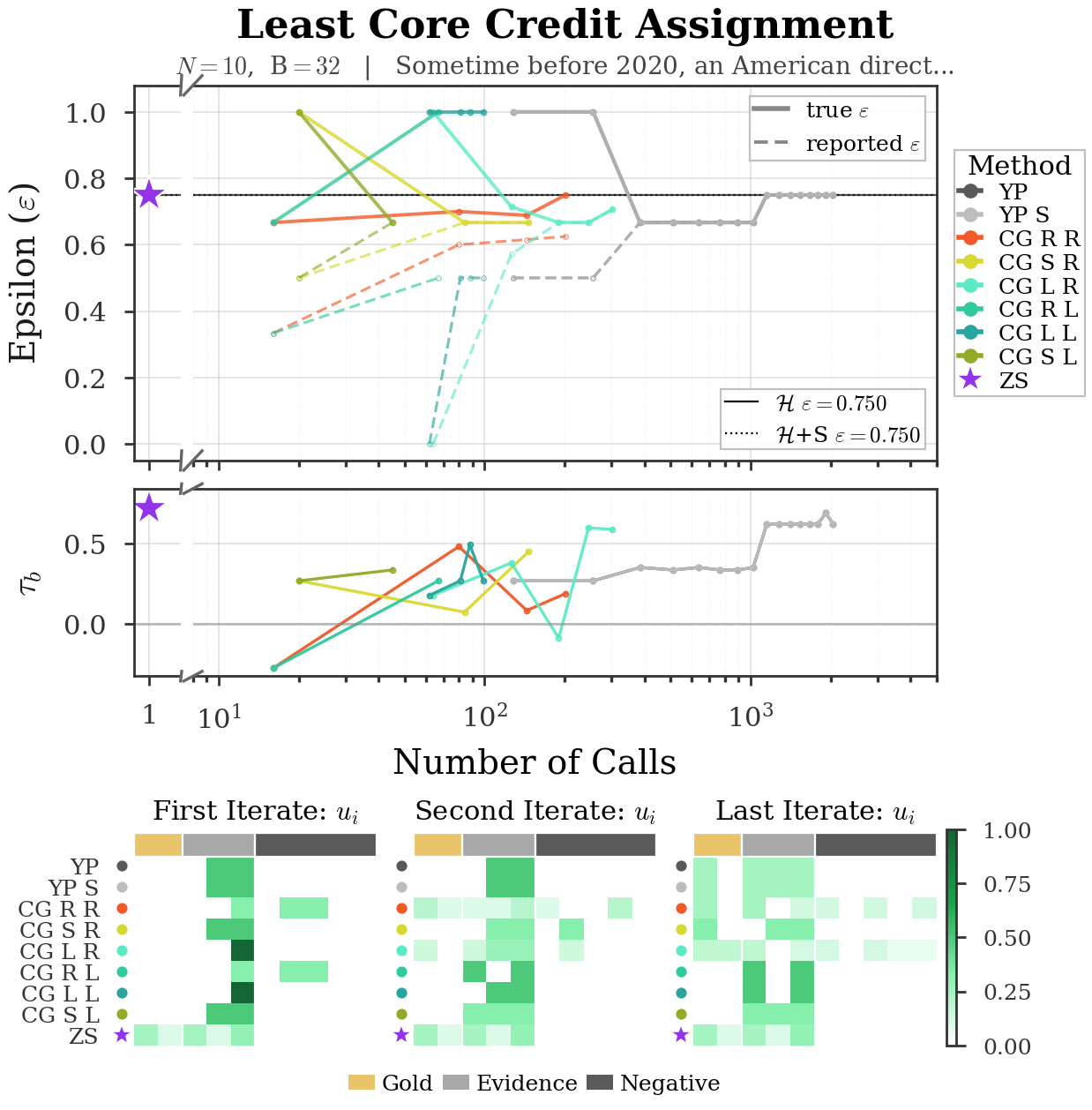}
        \caption{Row~10 (non-egal.)}
    \end{subfigure}\\[6pt]
    \begin{subfigure}[t]{0.49\textwidth}
        \includegraphics[width=\textwidth]{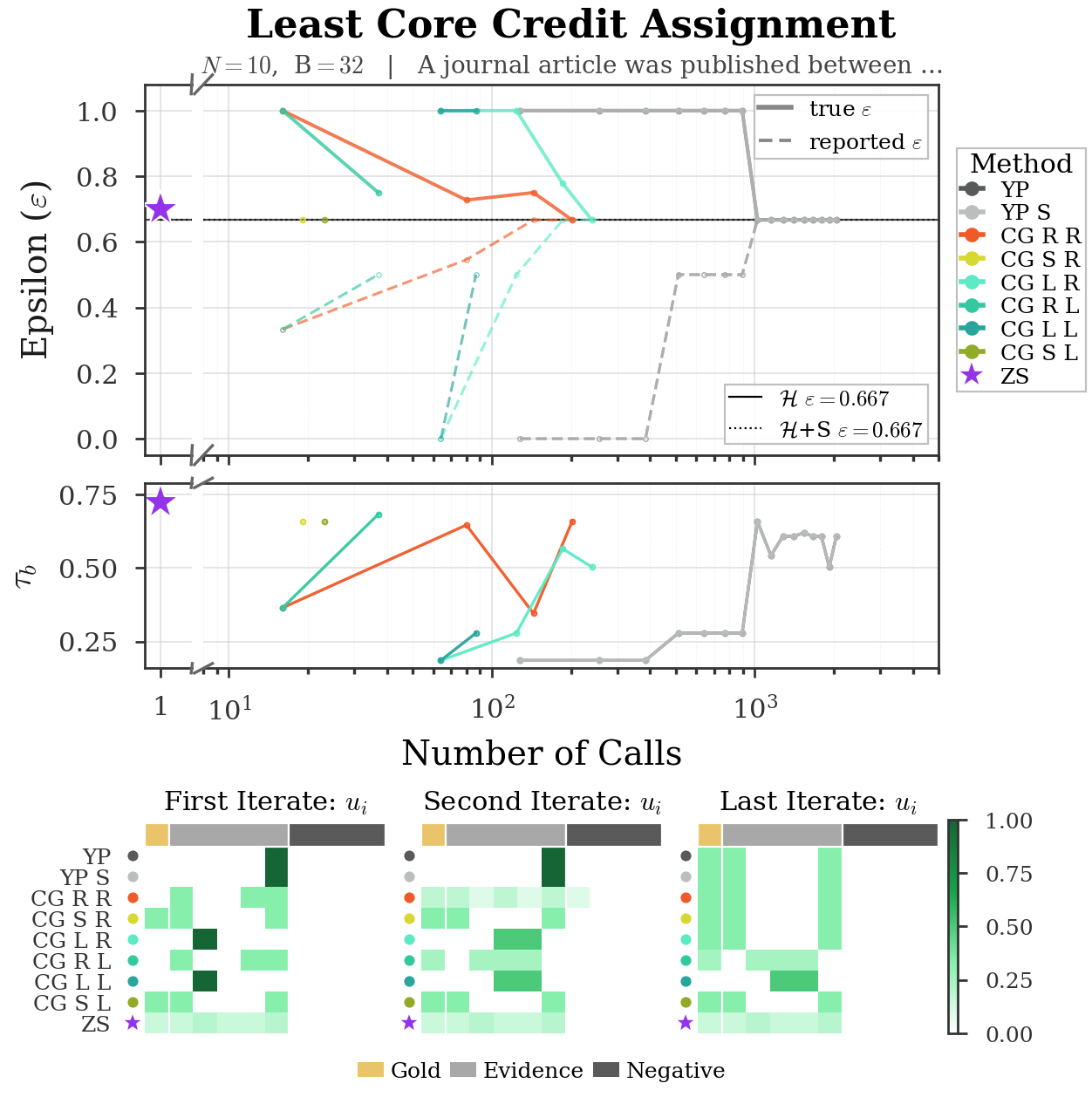}
        \caption{Row~18 (non-egal.)}
    \end{subfigure}\hfill
    \begin{subfigure}[t]{0.49\textwidth}
        \includegraphics[width=\textwidth]{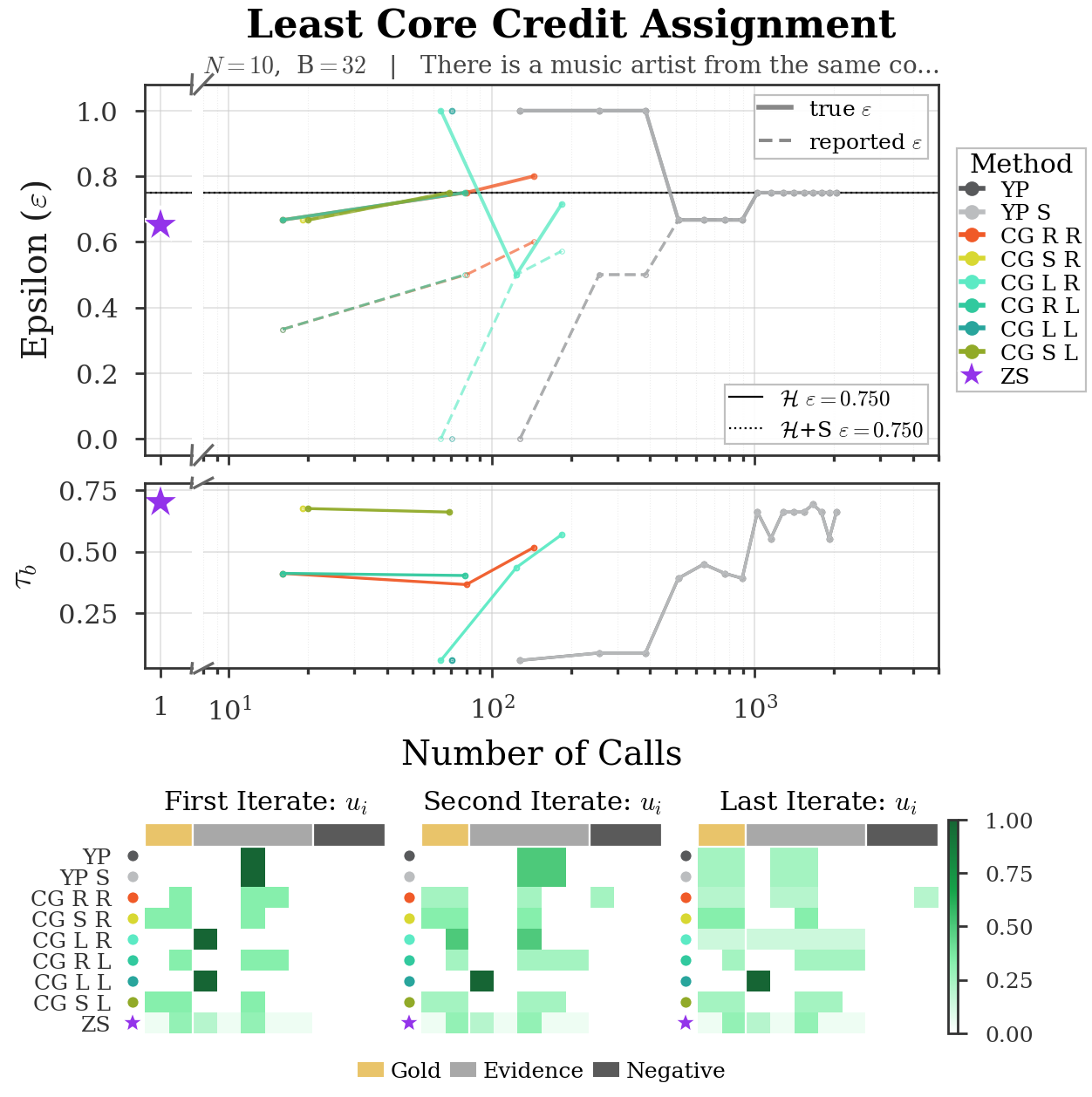}
        \caption{Row~21 (non-egal.)}
    \end{subfigure}
    \caption{BrowseComp $n{=}10$, non-egalitarian convergence traces. \textbf{CG} methods match or undershoot \textbf{YP}'s holdout $\hat{\varepsilon}_{\mathcal{H}}$ in 80--420 calls.}
    \label{fig:bc-n10-nonegal}
\end{figure*}

\begin{figure*}[!htb]\ContinuedFloat
    \centering
    \begin{subfigure}[t]{0.49\textwidth}
        \includegraphics[width=\textwidth]{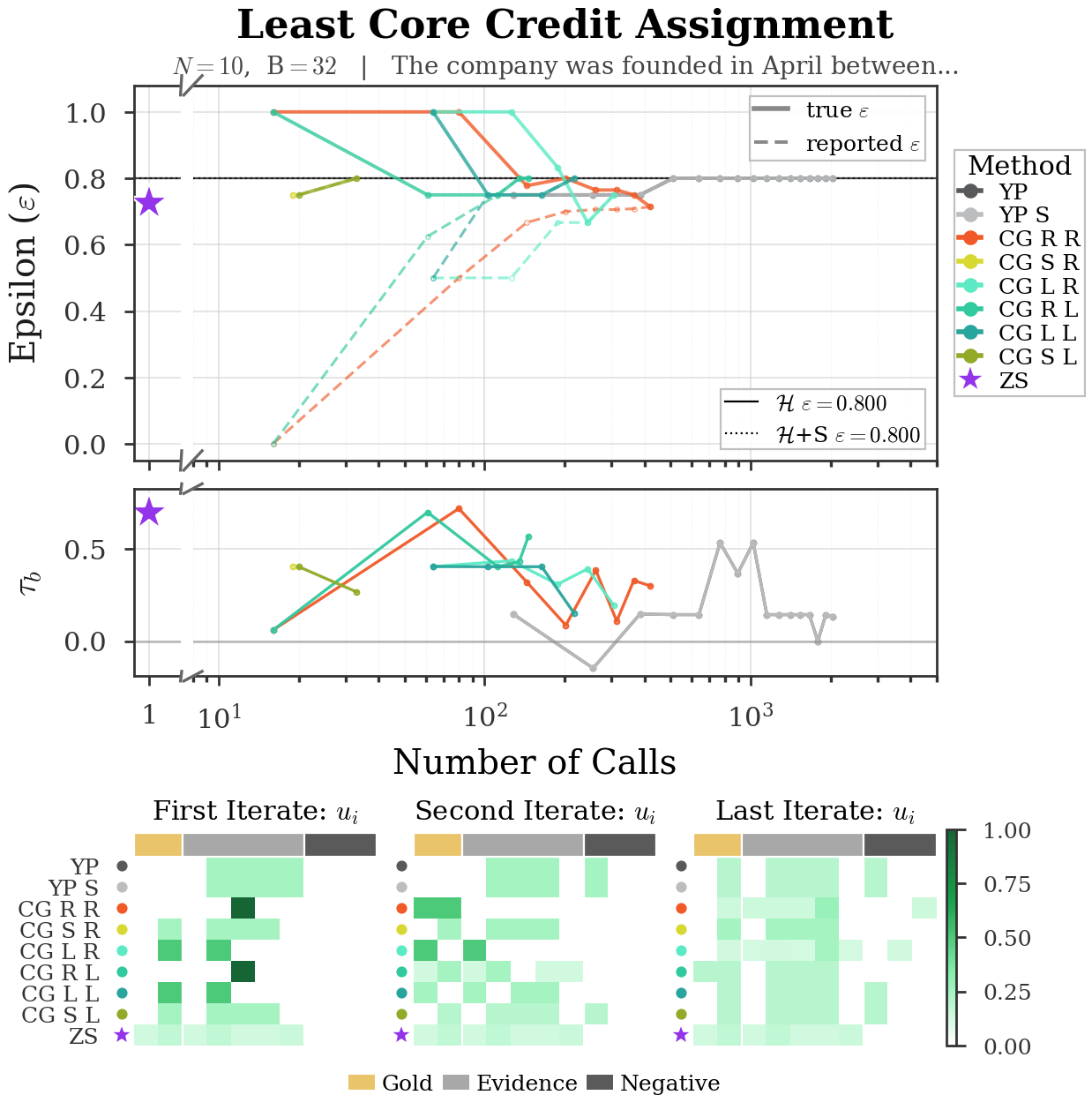}
        \caption{Row~22 (non-egal.)}
    \end{subfigure}
    \caption[]{BrowseComp $n{=}10$, non-egalitarian convergence traces (cont.).}
\end{figure*}

\newpage
\paragraph{Selected Egalitarian Traces.}
The egalitarian traces show broader credit redistribution with minimal $\hat{\varepsilon}_{\mathcal{H}}$ degradation. In Row~26, \textbf{CG(R,R)} undershoots \textbf{YP}'s $\hat{\varepsilon}_{\mathcal{H}} = 0.80$ by 0.13 ($\hat{\varepsilon}_{\mathcal{H}} = 0.67$ at ${\sim}144$ calls). Row~25 shows \textbf{CG(R,R)} slightly undershooting \textbf{YP} ($-0.03$). In Row~28, \textbf{CG(R,ZS)} matches \textbf{YP}'s $\hat{\varepsilon}_{\mathcal{H}} = 0.69$ within 0.01, while the LLM-oracle variants overshoot by ${\sim}0.31$---a reminder that \textbf{CG(L,$\cdot$)} trades convergence speed for precision.

\begin{figure*}[!htb]
    \centering
    \begin{subfigure}[t]{0.49\textwidth}
        \includegraphics[width=\textwidth]{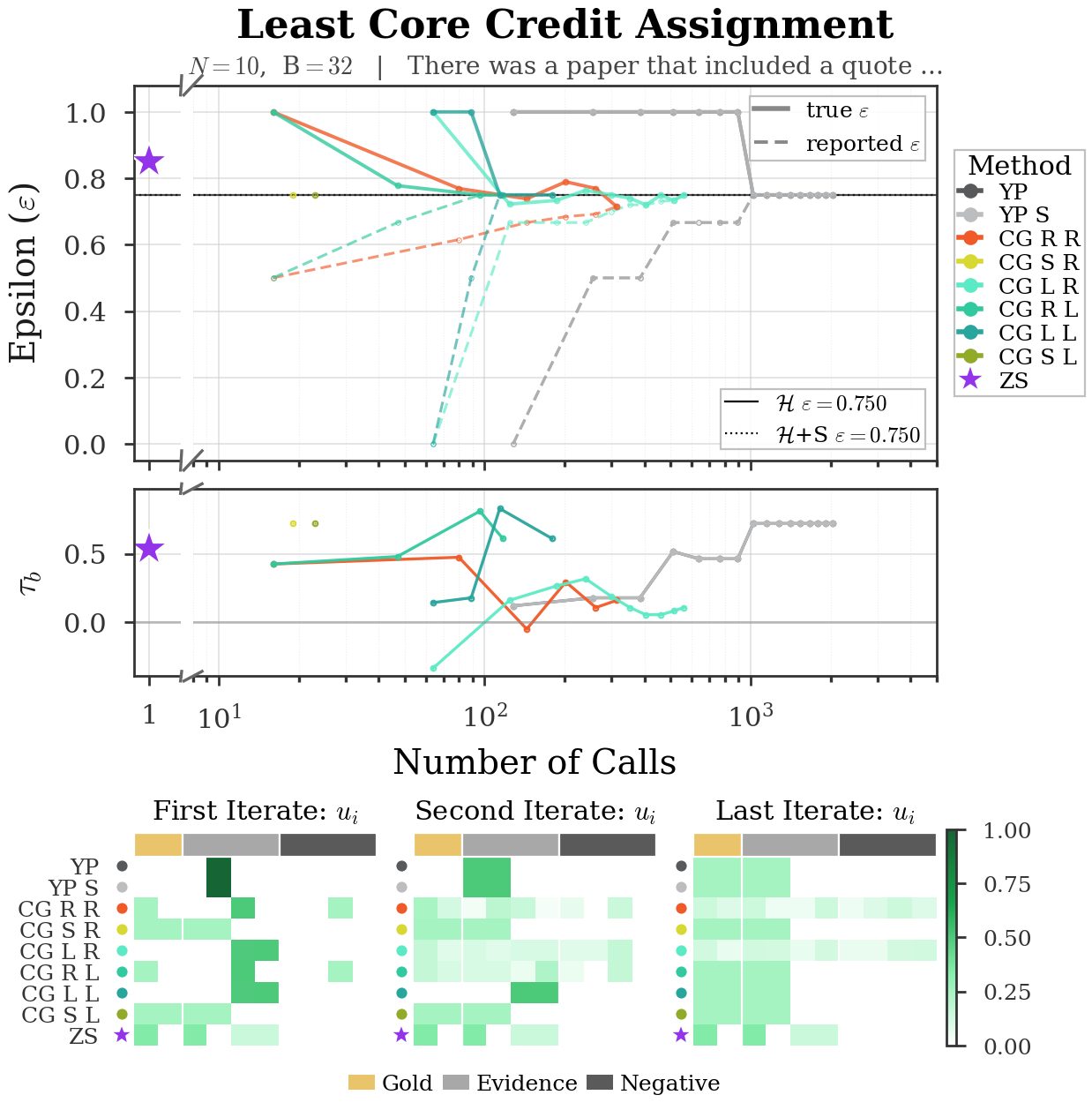}
        \caption{Row~12 (egal.)}
    \end{subfigure}\hfill
    \begin{subfigure}[t]{0.49\textwidth}
        \includegraphics[width=\textwidth]{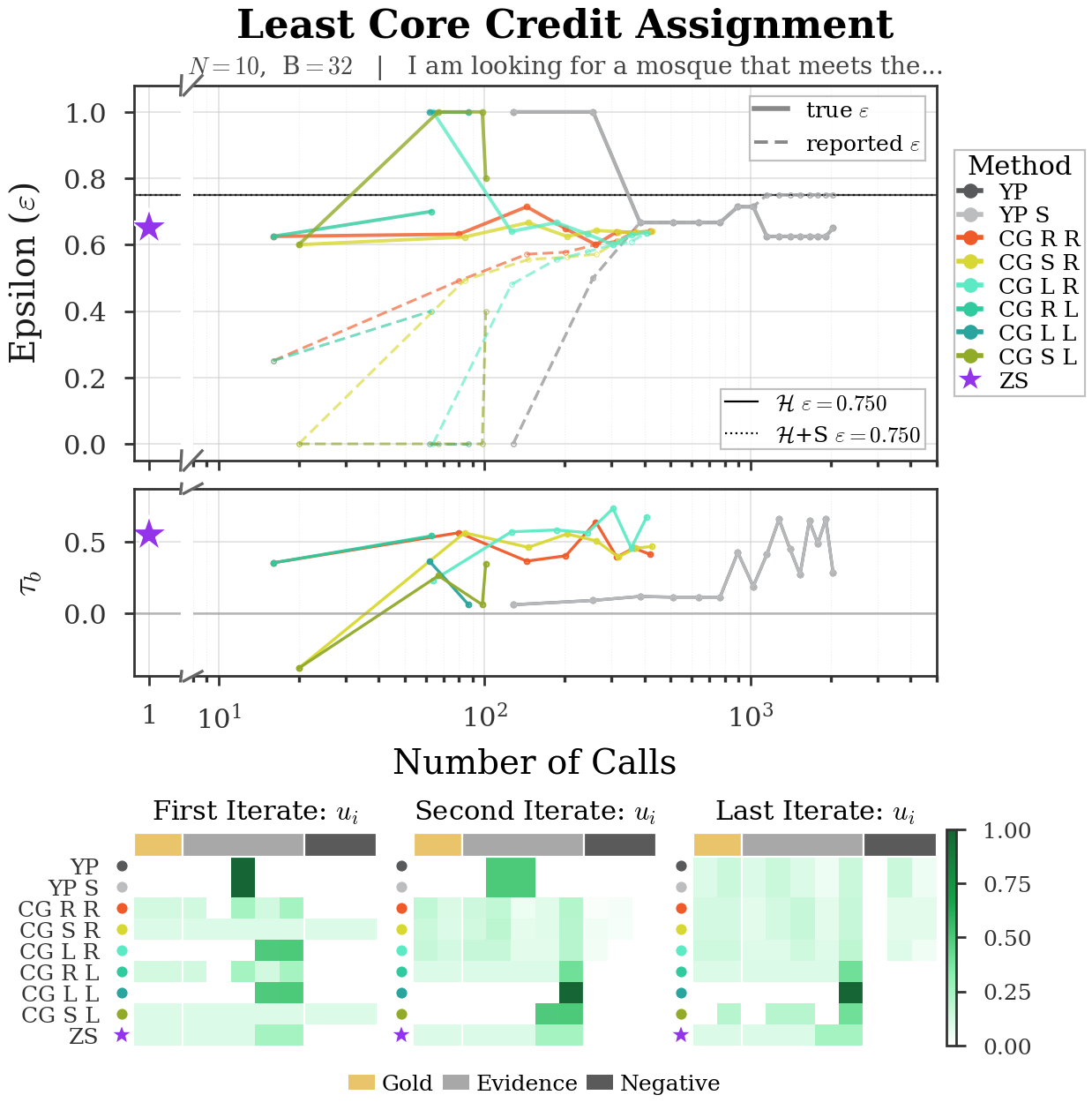}
        \caption{Row~15 (egal.)}
    \end{subfigure}\\[6pt]
    \begin{subfigure}[t]{0.49\textwidth}
        \includegraphics[width=\textwidth]{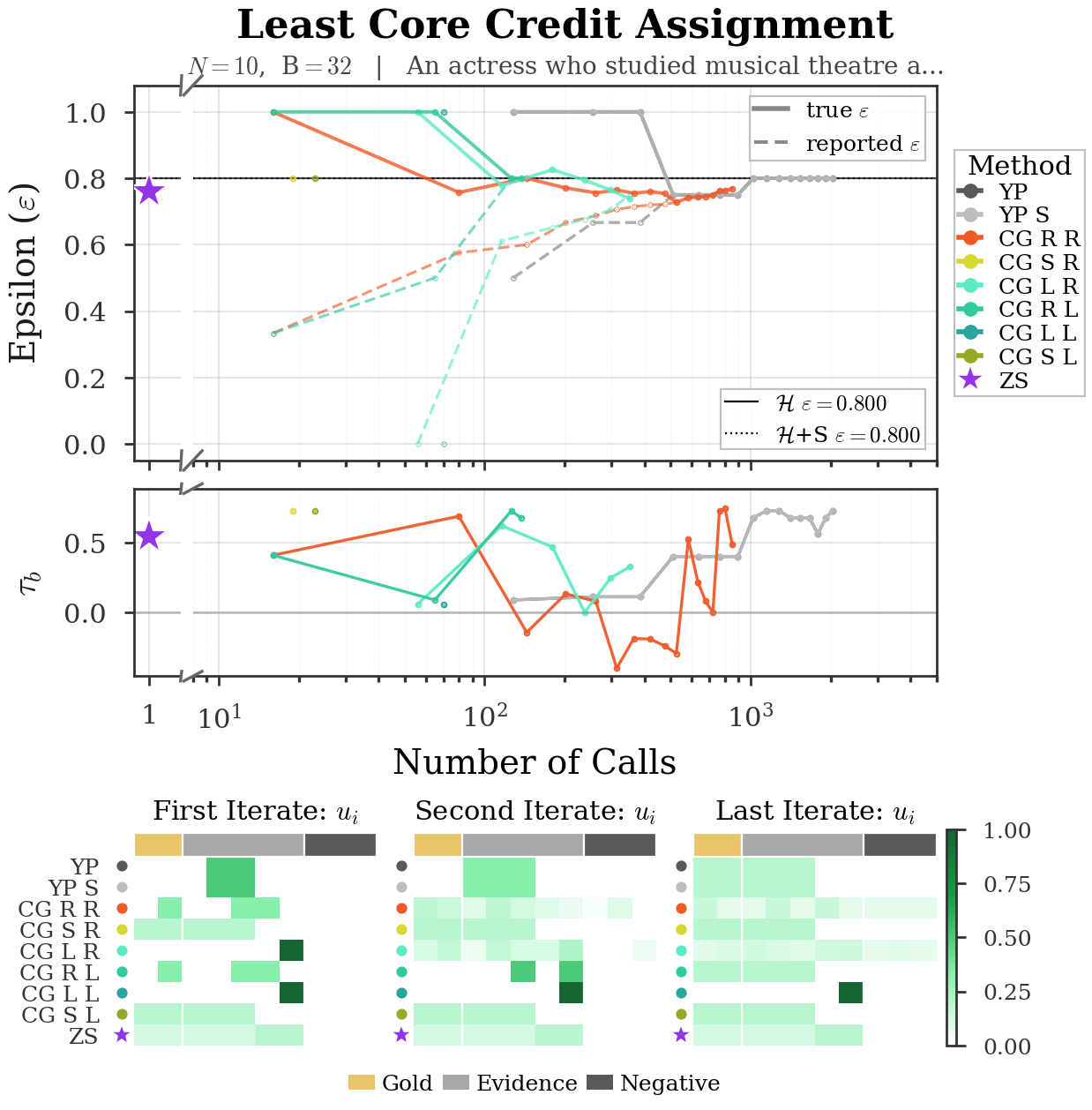}
        \caption{Row~25 (egal.)}
    \end{subfigure}\hfill
    \begin{subfigure}[t]{0.49\textwidth}
        \includegraphics[width=\textwidth]{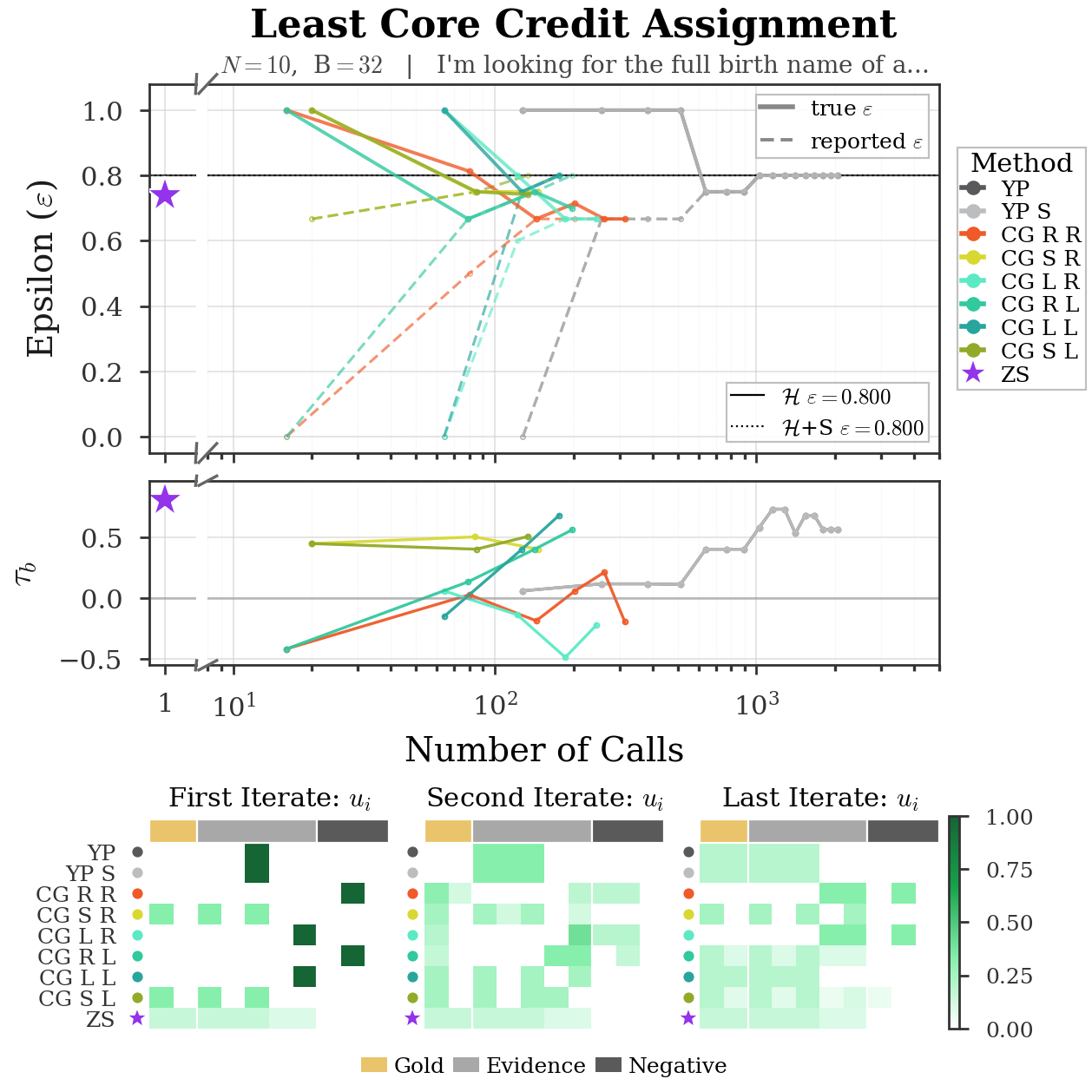}
        \caption{Row~26 (egal.)}
    \end{subfigure}
    \caption{BrowseComp $n{=}10$, egalitarian convergence traces. Heatmaps show broader credit redistribution; \textbf{CG(R,$\cdot$)} methods match or undershoot \textbf{YP}'s baseline.}
    \label{fig:bc-n10-egal}
\end{figure*}

\begin{figure*}[!htb]\ContinuedFloat
    \centering
    \begin{subfigure}[t]{0.49\textwidth}
        \includegraphics[width=\textwidth]{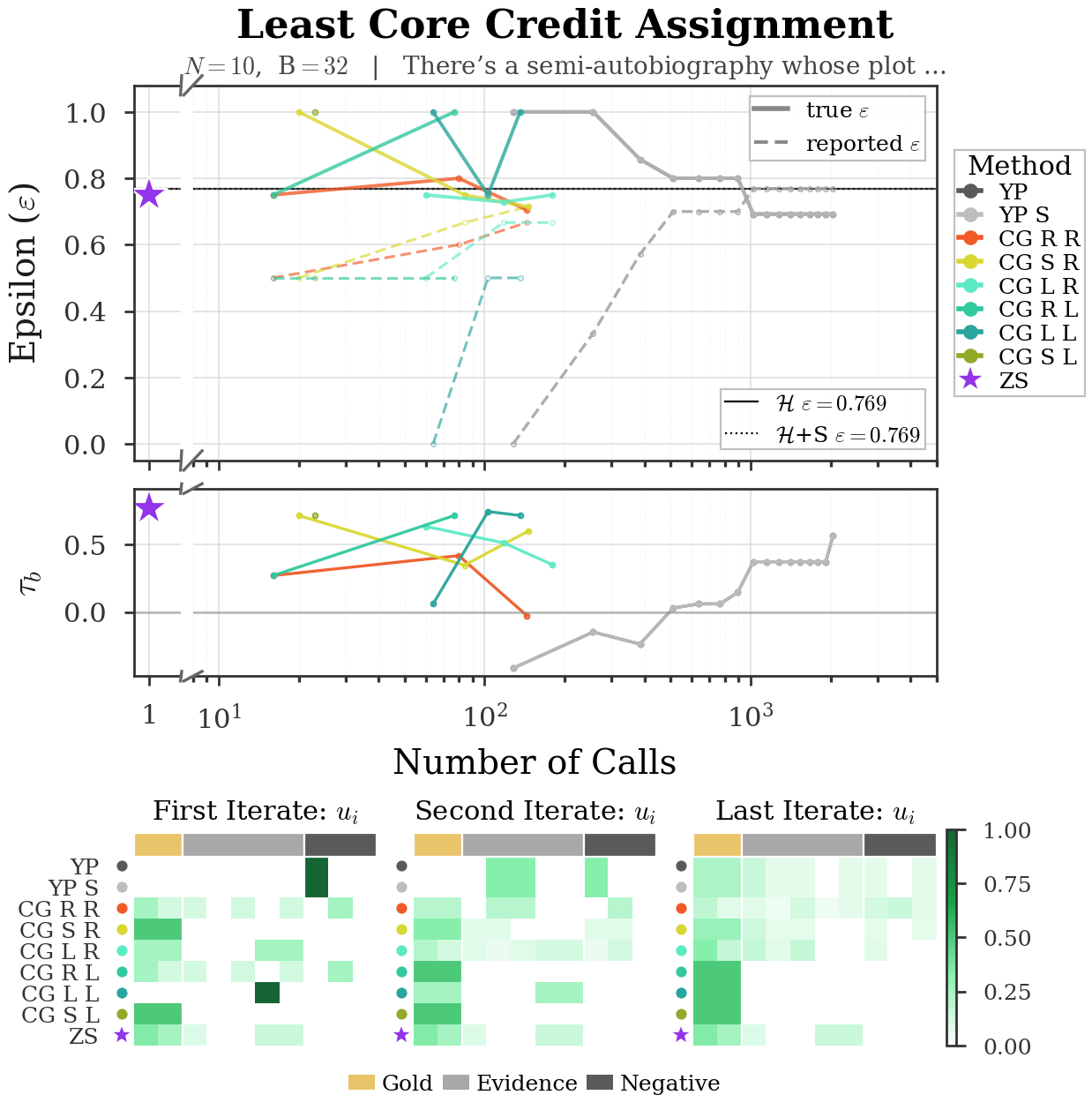}
        \caption{Row~28 (egal.)}
    \end{subfigure}
    \caption[]{BrowseComp $n{=}10$, egalitarian convergence traces (cont.).}
\end{figure*}

\newpage
\subsection{Calls-to-Target CDF Plots}
\label{app:cdf-plots}

The following plots show the fraction of $n{=}10$ instances for which each method's holdout $\hat{\varepsilon}_{\mathcal{H}}$ falls within a given tolerance of the holdout-optimal $\hat{\varepsilon}_{\mathcal{H}}^*$, as a function of LLM calls consumed. Since $n{=}10$ admits exhaustive evaluation (all $2^{10}-2$ non-trivial coalitions), $\hat{\varepsilon}_{\mathcal{H}}$ is exact (Q1.0, i.e., the true maximum violation). Baselines (\textbf{YP}, \textbf{YP\,S}) are plotted at their final call count only, since they do not support early stopping.

\begin{figure*}[!htb]
    \centering
    \begin{subfigure}[t]{0.24\textwidth}
        \includegraphics[width=\textwidth]{cdf_figs/cdf_compact_synth_nonegal_tol001.pdf}
        \caption{Synth (non-egal.)}
    \end{subfigure}\hfill
    \begin{subfigure}[t]{0.24\textwidth}
        \includegraphics[width=\textwidth]{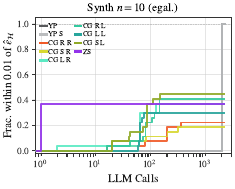}
        \caption{Synth (egal.)}
    \end{subfigure}\hfill
    \begin{subfigure}[t]{0.24\textwidth}
        \includegraphics[width=\textwidth]{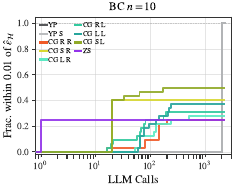}
        \caption{BC (non-egal.)}
    \end{subfigure}\hfill
    \begin{subfigure}[t]{0.24\textwidth}
        \includegraphics[width=\textwidth]{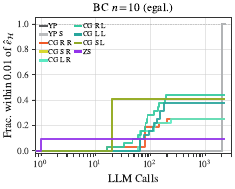}
        \caption{BC (egal.)}
    \end{subfigure}
    \caption{Calls-to-target CDF at tolerance 0.01: fraction of $n{=}10$ instances within 0.01 of $\hat{\varepsilon}_{\mathcal{H}}^*$ as a function of LLM calls.}
    \label{fig:cdf-tol001}
\end{figure*}

\end{document}